\documentclass[11pt,final,a4paper]{article}
\pdfoutput=1 

\def\draft{0}
\def\highlight{0}
\def\sigconf{0}
\def\bigfont{0}
\def\anonymous{0}
\def\shownomenclature{0}
\def\masterthesis{0}
\def\shortver{0}
\def\preprint{1}

\def\cryptology{0}
\def\llncs{0}
\def\lipics{0}
\def\acmtops{0}
\def\nextv{0}
\def\finance{0} 

\mathchardef\mhyphen="2D

    \makeatother   

\usepackage{setspace}
\ifnum\acmtops=0
\usepackage[dvipsnames,table]{xcolor}
\fi

\ifnum\preprint=1

\ifnum\sigconf=0\ifnum\lipics=0\usepackage{authblk}\fi\fi 

\usepackage{silence} 
\ifnum\lipics=0\usepackage[modulo]{lineno}\fi
\fi 

\ifnum\llncs=1

\fi 
\usepackage{amsthm}

\usepackage{wrapfig}

\ifnum\shortver=0
\fi

\ifnum\acmtops=1
\AtEndPreamble{%
\theoremstyle{acmdefinition}
\newtheorem{remark}[theorem]{Remark}
\newtheorem{claim}[theorem]{Claim}}
\fi

\usepackage{amsmath,amsfonts,latexsym,color,paralist,url,ifdraft,mathrsfs,thm-restate,comment,booktabs,dutchcal}

\usepackage{xspace}
\usepackage{blindtext} 
\usepackage{epigraph} 
\usepackage{lipsum}
\usepackage[utf8]{inputenc} \usepackage[autostyle=false, style=english]{csquotes} \MakeOuterQuote{"}
\usepackage[width=.85\textwidth]{caption}
\usepackage{pifont}
\usepackage{multirow}
\usepackage[all,cmtip]{xy}
\usepackage{wasysym}
\usepackage{soul}

\usepackage{etoolbox}
\AtBeginEnvironment{itemize}{\apptocmd{\item}{\phantomsection}{}{\errmessage{couldn't patch item}}}

\usepackage[T1]{fontenc}
\ifnum\sigconf=0
    \ifnum \shownomenclature=1
        \usepackage[refpage,prefix]{nomencl}  
        \makenomenclature
        \usepackage{xpatch}
        \xpatchcmd{\thenomenclature}{%
        \section*{\nomname}
        }{
        \section{\nomname}\label{sec:nomenclature}}{\typeout{Success}}{\typeout{Failure}}
    \fi
    \ifnum\lipics=0\ifnum\acmtops=0
    \usepackage[draft=false,pageanchor]{hyperref}
    \usepackage[draft=false,pageanchor]{hyperref}
    
    \usepackage[final]{graphicx}
    \fi\fi
\fi
\hypersetup{
    colorlinks=true,
    linkcolor=blue,
    filecolor=magenta,
    citecolor={green!75!black},
    urlcolor=cyan,
}

\usepackage{rotating}
\usepackage{array,makecell,multirow}

\ifnum\draft=1
    \usepackage{showlabels}
\fi 

\usepackage [ lambda,advantage , operators , sets , adversary , landau , probability , notions , logic ,ff, mm,primitives , events , complexity , asymptotics , keys]{cryptocode}

\theoremstyle{plain}

\newcommand{\snote}[1]{\authnote{Shai}{#1}{red}}

\ifnum\draft=1
    \newcommand{\authnote}[3]{{\color{#3} {\bf  #1:} #2}}
\else
    \newcommand{\authnote}[3]{}
\fi
\ifnum\masterthesis=0
    \ifnum\highlight=1
        \newcommand{\myhl}[1]{\hl{#1}}
    \else
        \newcommand{\myhl}[1]{#1}
    \fi
\else
    \newcommand{\myhl}[1]{#1}
\fi

\newcommand{\ket}[1]{\vert #1 \rangle}

\newcommand{\hashed}{\hyperref[col:tx_hashed]{hashed}\xspace}

\newcommand{\derived}{\hyperref[col:tx_derived]{derived}\xspace}
\newcommand{\Derived}{\hyperref[col:tx_derived]{Derived}\xspace}
\newcommand{\naked}{\hyperref[col:tx_naked]{naked}\xspace}
\newcommand{\stealable}{\hyperref[col:tx_steal]{stealable}\xspace}
\newcommand{\lost}{\hyperref[col:tx_lost]{lost}\xspace}
\newcommand{\Hashed}{\hyperref[col:tx_hashed]{Hashed}\xspace}
\newcommand{\Doomed}{\hyperref[col:tx_doom]{Doomed}\xspace}
\newcommand{\Naked}{\hyperref[col:tx_naked]{Naked}\xspace}
\newcommand{\Stealable}{\hyperref[col:tx_steal]{Stealable}\xspace}
\newcommand{\Lost}{\hyperref[col:tx_lost]{Lost}\xspace}

\newcommand{\secorssec}[1]{
    \ifnum\shortver=0
        \subsection{#1}
    \else
        \section{#1}
    \fi
}
\newcommand{\chr}[2]{
        \hyperref[#1]{#2}
}

\newcommand{\los}[2]{\ifnum\shortver=0 #1\else #2\fi}
\newcommand{\cof}[2]{\ifnum\finance=0 #1\else #2\fi}
\newcommand{\nextver}[1]{\ifnum\nextv=1 #1\fi}

\ifnum\draft=1
    \linenumbers
\fi

\usepackage{algorithm,algorithmicx,algpseudocode}
\usepackage[normalem]{ulem}

\newcommand{\sC}{\mathcal{C}}
\newcommand{\sD}{\mathcal{D}}
\newcommand{\sA}{\mathcal{A}}

\newcommand{\sK}{\mathcal{K}}

\newcommand{\sS}{\mathcal{S}}
\newcommand{\sT}{\mathcal{T}}
\newcommand{\sE}{\mathcal{E}}
\newcommand{\DD}{\mathbb{D}}

\newcommand{\xsk}{\mathsf{xsk}}
\newcommand{\xpk}{\mathsf{xpk}}
\newcommand{\msk}{\mathsf{msk}}

\newcommand{\ds}{\ensuremath{\mathsf{DS}}\xspace}

\renewcommand{\sign}{\ensuremath{\mathsf{Sign}}\xspace}
\newcommand{\ver}{\ensuremath{\mathsf{Ver}}\xspace}
\newcommand{\zo}{\{0,1\}}

\renewcommand{\seufcma}{\ensuremath{\mathsf{\mathsf{EUF\mhyphen CMA}}}\xspace}

\newcommand{\keygen}{\ensuremath{\mathsf{KeyGen}}\xspace}
\newcommand{\der}{\ensuremath{\mathsf{Der}}\xspace}
\newcommand{\kdf}{\ensuremath{\mathsf{KDF}}\xspace}

\newcommand{\sseuflcma}{\ensuremath{\mathsf{\mathsf{EUF\mhyphen LCMA}}}\xspace}

\newcommand{\utxo}{\ensuremath{\mathsf{\mathsf{UTXO}}}\xspace}
\newcommand{\utxos}{\ensuremath{\mathsf{\mathsf{UTXO}}}s\xspace}

\newcommand{\ptpk}{\ensuremath{\mathsf{\mathsf{P2PK}}}\xspace}
\newcommand{\ptpkh}{\ensuremath{\mathsf{\mathsf{P2PKH}}}\xspace}
\newcommand{\ecdsa}{\ensuremath{\mathsf{\mathsf{ECDSA}}}\xspace}
\newcommand{\secp}{\ensuremath{\mathsf{\mathsf{secp256k1}}}\xspace}
\newcommand{\tx}{\ensuremath{\mathsf{\mathsf{tx}}}\xspace}

\newcommand{\sha}{\ensuremath{\mathsf{\mathsf{SHA\mhyphen 256}}}\xspace}
\newcommand{\shaf}{\ensuremath{\mathsf{\mathsf{SHA\mhyphen 512}}}\xspace}
\newcommand{\HH}{\ensuremath{\mathsf{\mathsf{H}}}\xspace}
\newcommand{\pkec}{\ensuremath{\mathsf{PK^{EC}}}\xspace}

\usepackage{tikz}
\usetikzlibrary{arrows,chains,matrix,positioning,scopes}

\newcommand{\cmark}{\ding{51}}%
\newcommand{\xmark}{\ding{55}}%
\def\hcmark{\hspace{-0.2cm}
\parbox{10pt}{
    \begin{tikzpicture}
        \protect\node at (0, 0)   (a) {\cmark};
        \protect\draw[scale=0.4,fill=black,line width=0.4mm]  (-0.1,0.3) -- (0.25,-0.1);
    \end{tikzpicture}
    }
}

\makeatletter
\tikzset{join/.code=\tikzset{after node path={%
\ifx\tikzchainprevious\pgfutil@empty\else(\tikzchainprevious)%
edge[every join]#1(\tikzchaincurrent)\fi}}}
\makeatother
\tikzset{>=stealth',every on chain/.append style={join},
         every join/.style={->}}
\tikzstyle{labeled}=[execute at begin node=$\scriptstyle,
   execute at end node=$]
\ifnum\sigconf=0
    \ifnum\lipics=0
        \usepackage[capitalise,nameinlink]{cleveref} 
    \fi
\fi

\newcommand{\qpt}{\ensuremath{\mathsf{QPT}}\xspace}
\newcommand{\picnic}{\ensuremath{\mathsf{PICNIC}}\xspace}

\ifthenelse{\(\equal{\preprint}{1} \OR \equal{\masterthesis}{1}\) \OR \(\equal{\bigfont}{1}   \) \OR \(\equal{\cryptology}{1}   \)}
{
    \newtheorem{theorem}{Theorem} 
    
    \newtheorem{corollary}{Corollary}
    \newtheorem{proposition}{Proposition}
    \newtheorem{definition}{Definition}
    \newtheorem{example}{Example}
    
    \newtheorem{claim}{Claim}
    
    \newtheorem*{theorem*}{Theorem}
    \newtheorem*{lemma*}{Lemma}
    \newtheorem*{corollary*}{Corollary}
    \newtheorem*{proposition*}{Proposition}
    \newtheorem*{claim*}{Claim}
    \theoremstyle{definition}
    
    \theoremstyle{remark}
    \newtheorem{remark}[theorem]{Remark}

    \theoremstyle{plain}
    
}{}

\usepackage{autonum}

\expandafter\let\expandafter\savedflalignstar\csname flalign*\endcsname
\expandafter\let\expandafter\savedendflalignstar\csname endflalign*\endcsname
\AtBeginDocument{%
  \expandafter\let\csname flalign*\endcsname\savedflalignstar
  \expandafter\let\csname endflalign*\endcsname\savedendflalignstar
}

\newcommand{\ab}[1]{}  

\ifnum\masterthesis=1
    \usepackage{setspace}
    \usepackage{fancyhdr}
    \renewcommand{\headheight}{14pt}
\fi
\DeclareMathAlphabet{\mathpzc}{OT1}{pzc}{m}{it}

\usepackage{adjustbox}
\usepackage{schemabloc}
\usepackage{ulem}
\usetikzlibrary{shapes,arrows}
\usetikzlibrary{positioning}
\usetikzlibrary{shapes.geometric}
\usepackage{tcolorbox}

\ifnum\shortver=0\ifnum\lipics=0\ifnum\acmtops=0
    \usepackage{fullpage}
\fi\fi\fi

\usepackage{ stmaryrd }

\ifnum\acmtops=1
    
\fi

\begin{document}

\newcommand\tikzmark[1]{%
\tikz[remember picture,baseline] \node[inner sep=2pt,outer sep=0] (#1){};%
}

\cof{
\ifnum\llncs=1
\title{Protecting Quantum Procrastinators with Signature Lifting}
\subtitle{A Case Study in Cryptocurrencies}
\else
\title{Protecting Quantum Procrastinators with Lifted Signatures: a Case Study in Cryptocurrencies}
\fi
}
{\title{Protecting Quantum Cryptocurrency Procrastinators with Signature Lifting}}
\ifnum\anonymous=0
    \ifnum\llncs=1
        %
        \author{Or Sattath\inst{1}
        \and
        Shai Wyborski\inst{2,1}
        }
        %
        %
        \institute{Department of Computer Science, Ben-Gurion University of the Negev, Israel \\
        \email{sattath@bgu.ac.il} \and
        School of Computer Science and Engineering, The Hebrew University of Jerusalem, Israel\\
        \email{shai.wyborski@mail.huji.ac.il}
        }
        
        \authorrunning{O. Sattath and S. Wyborski}
    \fi
    \ifnum\acmtops=1
    
        \author{Or Sattath}
        \orcid{000-0001-7567-3822}
        \affiliation{%
          \institution{Department of Computer Science, Ben-Gurion University of the Negev}
          \city{Beer Sheva}
          \country{Israel}}
        
        \author{Shai Wyborski}
        \orcid{0000-0001-6847-5668}
        \affiliation{%
          \institution{School of Computer Science and Engineering, The Hebrew University of Jerusalem}
          \city{Jerusalem}
          \country{Israel}}
        \affiliation{%
          \institution{Department of Computer Science, Ben-Gurion University of the Negev}
          \city{Beer Sheva}
          \country{Israel}
          }
    \fi
    \ifnum\preprint=1
        \author[1]{Or Sattath}
        \author[1,2]{Shai Wyborski}
        \affil[1]{Computer Science Department, Ben-Gurion University of the Negev}
        \affil[2]{School of Computer Science and Engineering, The Hebrew University of Jerusalem, Israel}
    \fi
    
\fi

\ifnum\acmtops=1

\begin{abstract}
\cof{
    Current solutions to quantum vulnerabilities of widely used cryptographic schemes involve migrating users to post-quantum schemes \emph{before} quantum attacks become feasible. This work deals with protecting \emph{quantum procrastinators}: users that failed to migrate to post-quantum cryptography in time.

    To address this problem in the context of digital signatures, we introduce a technique called \emph{signature lifting}, that allows us to lift a deployed pre-quantum signature scheme satisfying a certain property to a post-quantum signature scheme that uses the \emph{same} keys. Informally, the said property is that a post-quantum one-way function is used "somewhere along the way" to derive the public-key from the secret-key. Our constructions of signature lifting relies heavily on the post-quantum digital signature scheme Picnic (Chase et al., CCS'17).

    Our main case-study is cryptocurrencies, where this property holds in two scenarios: when the public-key is generated via a key-derivation function or when the public-key hash is posted instead of the public-key itself. We propose a modification, based on signature lifting, that can be applied in many cryptocurrencies for securely spending pre-quantum coins in presence of quantum adversaries. Our construction improves upon existing constructions in two major ways: it is not limited to pre-quantum coins whose \ecdsa public-key has been kept secret (and in particular, it handles all coins that are stored in addresses generated by HD wallets), and it does not require access to post-quantum coins or using side payments to pay for posting the transaction.}
{
    Most cryptocurrencies rely on signature schemes that are susceptible to quantum attacks. The straightforward way to resolve this risk is for users to migrate to post-quantum signature schemes \emph{before} quantum attacks become feasible.
    In this work, we show how to protect \emph{quantum procrastinators}: users that failed to migrate in time\footnote{In other words, everyday users (authors included).}. 

    To address this problem we introduce \emph{signature lifting}, a technique for lifting a deployed pre-quantum signature scheme to a post-quantum scheme that \emph{uses the same keys}, given that a post-quantum one-way function is used "somewhere along the way" to derive the public-key from the secret-key. This property is common in cryptocurrencies, mostly when the public-key is generated via a key-derivation function or the public-key's hash is posted instead of the public-key itself. Our constructions rely heavily on the post-quantum digital signature scheme Picnic (Chase et al., CCS'17).

    We provide a modification, based on signature lifting, that is applicable to Bitcoin and many other cryptocurrencies. This modification furnishes the first protocol not limited to pre-quantum coins whose \ecdsa public-key has been kept secret (in particular, it handles all coins that are stored in addresses generated by HD wallets), and not requiring access to post-quantum coins (or using side payments) to pay transaction fees. The importance of the first improvement is augmented by the Taproot update---now used by more than 30\% of Bitcoin transactions---where the wallet address is an \ecdsa public-key given in the clear.

    We also discuss \emph{quantum canaries}---a mechanism designed to detect and determine \emph{in consensus} the emergence of capable quantum entities---and provide a rudimentary game-theoretic analysis thereof.
}
\end{abstract}
    \maketitle
\else
    {
    \maketitle}
    
\fi

\newpage
\ifnum\shortver=0
\fi
\ifnum\shortver=0\ifnum\lipics=0\ifnum\acmtops=0
    \setcounter{tocdepth}{2}
    \tableofcontents
\fi\fi\fi

\newpage

\snote{
Nextvers:
\begin{itemize}
     \item Next version:
    \begin{itemize}
        \item Beautify protocols 2 - use some designated environment to preset the protocols, maybe defer all explicated protocols to a designated appendix.
        \item Discuss the fact that an adversary with > 50\% certainty that a utxo is lost has positive expected profit from trying to pilfer it.
        \item Discuss how a public-key could be chosen uniformly (or close enough to uniformly) without leaking information about the public-key (e.g. sample a random $x$ coordinate, calculate $y$ coordinates and flip a coin to choose one of them. Another idea, modify the generator some way and then take a random power of the modified generator. I think there is a way to do so such that DLOG is reducible to figuring out the order w.r.t. the original generator even with knowledge of the order w.r.t. the new generator).
        \item Using FC to post fraud-proofs for LFC
    \end{itemize}
    \item More research required before I decide if to include in next version or not at all:
    \begin{itemize}
        \item Discussion about how we expect people to use our methods. Do we expect most of them to spend to post-quantum addresses?
        \item Value of deposit: discuss the possibility that the deposit could depend on the year the \utxo was created
        \item The discussion about pure vs. mixed Nash equilibria
    \end{itemize}
\end{itemize}
}


\section{Introduction} 

\cof{
The canonical solution for quantum attacks on widely used cryptographic schemes is to migrate users to post-quantum schemes before quantum attacks become feasible. The focus of this work is \emph{quantum procrastinators}, namely, users who remain dependent on the security of pre-quantum cryptography in the era of scalable quantum computation. 

There are many possible reasons for procrastinators to exist:
\begin{itemize}
    \item Most notably, the tendency of individuals and organizations\footnote{and cryptographers} to procrastinate: some users would not act in a timely manner. Even after the vulnerability is identified and a solution is available, it might take a long time for organizations to adapt to changes due to security risks. For example, recall the WEP standard for securing local wireless networks. Even though the vulnerabilities of WEP to were identified and fixed by the succeeding WPA standard within months, it took about a decade for the majority of systems their systems due to critical vulnerabilities. Criminals were able to exploit these delays to steal millions of credit cards \cite[Section~6.2.1]{Bar21}.
    \item An unexpected technological breakthrough, perhaps done in secret, could make quantum attacks faster than expected.
    \item Post-quantum technologies have not been battle-tested and scrutinized to a fraction of the extent widely used cryptographic schemes have, and are still rapidly evolving and improving in terms of security and performance, which incentivizes delaying adoption.
\end{itemize}

In this work, we focus on users who rely on pre-quantum signature schemes whose public-key is known to the adversary. To help such procrastinators, we introduce post-quantum \emph{signature lifting}: a technique for lifting an already deployed pre-quantum signature scheme to a post-quantum signature scheme \emph{with the same keys}, given that a post-quantum one-way function is applied "somewhere along the way" when computing the public-key from the secret-key. This property does not hold for commonly used signature schemes such as \ecdsa or Schnorr signatures. However, there are several applications of digital signature that could be recast such that it does hold. The two scenarios we consider are a) when a hash of the public-key is provided rather than the public-key itself, and b) when the secret-key is generated using a key-derivation function.


\subsection{Case Study: Cryptocurrency Procrastinators}
}{}

Bitcoin, Ethereum, and most other cryptocurrencies rely on the security of signature schemes that are susceptible to quantum attacks, such as the \ecdsa and Schnorr signature schemes. Even an unspent transaction (\utxo) whose public-key is kept secret, while only its hash is available on the blockchain, is susceptible to attacks while the transaction spending it is in the mempool \cite{But13}. 

The canonical solution for quantum attacks on widely used cryptographic schemes is to migrate users to post-quantum schemes before quantum attacks become feasible. The focus of this work is \emph{quantum procrastinators}, namely, users who remain dependent on the security of pre-quantum cryptography in the era of scalable quantum computation. 

\cof{}{
There are many possible reasons for procrastinators to exist:
\begin{itemize}
    \item Most notably, the tendency of individuals and organizations\footnote{and cryptographers} to procrastinate:



    some users would not act in a timely manner. Even after the vulnerability is identified and a solution is available, it might take a long time for organizations to adapt to changes due to security risks. For example, recall the WEP standard for securing local wireless networks. Even though the vulnerabilities of WEP to were identified and fixed by the succeeding WPA standard within months, it took about a decade for the majority of systems to upgrade due to critical vulnerabilities. Criminals were able to exploit these delays to steal millions of credit cards \cite[Section~6.2.1]{Bar21}.
    \item An unexpected technological breakthrough, perhaps done in secret, could make quantum attacks faster than expected.
    \item Post-quantum technologies have not been battle-tested and scrutinized to a fraction of the extent widely used cryptographic schemes have, and are still rapidly evolving and improving in terms of security and performance, which incentivizes delaying adoption.
\end{itemize}
}
We consider the following question:
\begin{tcolorbox}
    \textbf{Question 1:} how can cryptocurrencies maximize the number of coins that procrastinators could securely spend?
\end{tcolorbox}

We say that a method to spend pre-quantum \utxos is \emph{quantum-cautious} if it is secure even in the presence of quantum adversaries. 
Current quantum-cautious spending methods \cite{BM14,SIZ+18,IKS19} allow cautiously spending pre-quantum \utxos that are \emph{hashed}, that is, whose public-key has not leaked (e.g., leakage can occur by using the same address for several \utxos and only spending some of them). However, these methods have a considerable drawback: to pay a fee to the miner for including the transaction, the user must either use a side payment or have access to post-quantum \utxos.

Our cautious spending methods provide two major improvements. First, they allow cautiously spending hashed pre-quantum \utxos without requiring a side-payment or a post-quantum \utxo. Second, they extend the set of cautiously spendable \utxos to include:
\begin{itemize}
        \item \utxos whose address was generated by an HD wallet \myhl{(see }\cref{ssec:hdwallet}\myhl{)}, \emph{even if} the public-key has leaked,
        \item \utxos whose public-key has leaked \emph{even if} it was not generated by a key-derivation function, as long as the adversary is unaware\snote{footnote with reference} that the public-key was not generated by a key-derivation function, and
        \item \utxos whose secret-key was \emph{lost}, but the fact that it was lost was kept secret by the user.
\end{itemize}

The protocol is designed such that attempting to recover a  \utxo of the second or third type requires leaving a deposit that can be claimed by the true owner of the \utxo if it happens to be of the first type, making it costly for an attacker to attempt stealing arbitrary transactions, hoping to land on a stealable \utxo.

To make the latter two types spendable, all procrastinators who own pre-quantum \utxos whose public-keys have been leaked must be \emph{online} periodically (e.g., once a year) to notice attempts at stealing their coins. Allowing only spending the first type avoids this drawback while still considerably increasing the set of cautiously-spendable \utxos.

We stress that a large majority of bitcoin arguably resides in \utxos whose address is a derived public-key (or a hash thereof). In particular, all \utxos generated by a Hierarchical Deterministic (HD) wallet \cite{Wui13} fall under this category. As of 2022, all \href{https://bitcoin.org/en/choose-your-wallet?step=5}{bitcoin.org recommended wallets} are HD wallets, and have been HD wallets since as early as 2015\footnote{As could be verified by tediously looking them up one by one.}. According to an analysis by CryptoQuant \cite{Cry22}, as of September 2022, \utxos generated during and after 2015 contain more than 16.5 million bitcoins.

Furthermore, in the newly deployed Bitcoin update Taproot \cite{WNT20}, the public-key of the owner of a \utxo is posted to the blockchain in the clear. Hence, current methods are unable to quantum-cautiously spend \utxos generated by following the Taproot specification, whereas our methods can. Taproot was deployed in November 2021, and its adoption has been steadily increasing since. 
As of May 2023, about 30\% of Bitcoin's transaction throughput is spent to Taproot addresses \cite{Tra23}.

Our methods can be deployed in either \emph{restrictive} mode, in which some funds of honest procrastinators are effectively burned, or \emph{permissive} mode, in which those funds are spendable but require a large \emph{deposit} as well as a very long \emph{challenge period} (during which anyone can \emph{challenge} the transaction by posting a \emph{fraud proof}). Hence, it is desirable to defer deploying these methods as much as reasonable. Furthermore, since the current state of affairs is that quantum computers are not an immediate concern, users will tend to procrastinate as they will only realize that quantum threats are relevant \emph{after} their coins are stolen. Both considerations motivate the next question:

\begin{tcolorbox}
    \textbf{Question 2:} how can we warn users \emph{in advance} that quantum adversaries will emerge in the near future? Can we make this warning \emph{in consensus}?
\end{tcolorbox}

Towards this goal, we present \emph{quantum canaries}, a mechanism for signaling the emergence of quantum adversaries in a way that could be read off the blockchain. Quantum canaries are puzzles designed to be infeasible for classical computers but solvable for quantum computers whose scale is significantly smaller than required to attack \ecdsa signatures. We propose to use quantum canaries to determine the block from which restrictions on spending pre-quantum \utxos will take effect.

Even though we try to be as general as possible, some of our results may be hard to adjust, for example, when considering cryptocurrencies with advanced smart-contract capabilities, such as Ethereum, which may have more complicated logic; or privacy-oriented cryptocurrencies, such as Zerocash \cite{BCG+14}, where special attention is given to privacy and anonymity---aspects which we completely ignore in this work.

\subsection{Types of \textsf{UTXOs}}

\ifnum\masterthesis=0{In this paper, we discuss several methods for quantum cautious spending of \utxos. Each method is only applicable to \utxos satisfying some conditions, }\else{The different spending methods we propose apply to \utxos satisfying different conditions, }\fi so we first identify several useful subsets of the \utxo set, defined by the information available to the owner of the \utxo and to an adversary trying to steal the coins in the \utxo.
When considering addresses that were derived by an HD wallet (see \ifnum\masterthesis=0{\cref{ssec:hdwallet}}\else{\cref{ssec:prelim_other_hd}}\fi), we always assume that the adversary has no access to the seed
used to generate the address of the \utxo (that is, we consider anyone holding the derivation seed an "owner" of the \utxo). We summarize the properties of these sets in \cref{table:utxos}.\ifnum\masterthesis=0{ \footnote{\ifnum\acmtops=0 For an overview of the \utxo model and blockchains in general, we refer the reader to \cite{NBF+16}\else For an overview of the \utxo model and blockchains in general, we refer the reader to \cite{NBF+16}. For an in-depth introduction and analysis of key-derivation functions and HD wallets providing all details required for the current work, we refer the reader to the full version \cite{SW23}.\fi}}\fi

\begin{itemize}
    \item \textbf{Hashed}\label{col:tx_hashed}: a \utxo whose secret-key is known to the owner and whose public-key is not known to the adversary.
    \item \textbf{Derived}\label{col:tx_derived}: a \utxo whose owner knows a derivation seed from which the signature keys were derived (see \ifnum\masterthesis=0\cref{ssec:hdwallet}\else\cref{ssec:prelim_other_hd}\fi).
    \item \textbf{Naked}\label{col:tx_naked}: \myhl{a }\utxo\myhl{ whose public-key is known to the adversary, and whose owner knows the secret-key but \emph{not} the derivation seed (either because it was lost or because the public-key is non-derived), yet the adversary does not know that the owner does not have access to the derivation seed.}
    \item \textbf{Lost}\label{col:tx_lost}: a \utxo whose secret-key is not known to anyone, but whose public-key is known to the adversary, yet the adversary does not know the user does not know the secret-key. Some of our methods allow owners of lost \utxos to recover their funds (see \cref{ssec:pfc}).
    \item \textbf{Stealable}\label{col:tx_steal}: a \utxo whose derivation seed is not known to the owner (either because they lost it or because the signature keys were not derived from a key-derivation function), and whose public-key is known to the adversary as well as the fact that the owner does not have the derivation key. Such \utxos are called stealable as they would be stolen by a quantum adversary if unrestrictive FawkesCoin is implemented (see \cref{ssec:nfc}). If permissive FawkesCoin is implemented (see \cref{ssec:pfc}), these funds would also be stolen by \emph{classical} adversaries. 
    \item \textbf{Doomed}\label{col:tx_doom}: a \utxo whose secret-key and public-key are unknown to anyone. We call such \utxo doomed as no one (including a quantum adversary) can spend these with any of our methods.
\end{itemize}

\begin{remark}
    Note that most pairs of the two sets above are disjoint, with \hashed and \derived \utxos being the only exception. If a \utxo is both \hashed and \derived, the user could treat it as any of the two when spending it.
\end{remark}

\setlength{\tabcolsep}{6pt}
\begin{table}[!htb]
\begin{adjustbox}{max width=1.1\textwidth,center}
\begin{tabular}{c@{\hskip 0.2cm}|@{\hskip 0.1cm}llcccccc}
\multicolumn{1}{l}{} &  &  & \multicolumn{1}{c}{\rotatebox[origin=r]{270}{\Hashed}} & \multicolumn{1}{c}{\rotatebox[origin=r]{270}{\Derived}} & \multicolumn{1}{c}{\rotatebox[origin=r]{270}{\Naked}} & \multicolumn{1}{c}{\rotatebox[origin=r]{270}{\Lost}} & \multicolumn{1}{c}{\rotatebox[origin=r]{270}{\Stealable}} & \multicolumn{1}{c}{\rotatebox[origin=r]{270}{\Doomed}} \\ \toprule 
\multirow{3}{*}{\rotatebox[origin=r]{270}{Owner}} &  & Derivation Seed & \hcmark & \cmark & \xmark & (\xmark) & \xmark & (\xmark) \\
 &  & Secret-Key & \cmark & (\cmark) & \cmark & \xmark & \hcmark & (\xmark) \\
 &  & Public-Key & (\cmark) & (\cmark) & (\cmark) & \cmark & \cmark & \xmark \\ \midrule 
\multirow{2}{*}{\rotatebox[origin=r]{270}{Adv.}} &  & Public-Key & \xmark & \hcmark & \cmark & \cmark & \cmark & (\xmark) \\
 &  & Non-derived/lost & \hcmark & (\xmark) & \xmark & \xmark & \cmark & \hcmark \\ \bottomrule
\end{tabular}
\end{adjustbox}

\caption{
    The different types of \utxos are defined by the information available to the owner and an adversary. Each column represents a subset of the \utxo set, and each row represents a particular datum regarding this \utxo and whether it is known to the owner/adversary. The non-derived/lost row, in particular, represents a scenario where the adversary knows that one of the following holds: either the \utxo address was not generated by an HD wallet, or the owner lost access to the corresponding secret-key (note these scenarios are not mutually exclusive).
    \hcmark indicates that the property described in the corresponding row is irrelevant to the type defined by the column (e.g., a \utxo whose public-key is known to the user but not known to the adversary is \hashed regardless of whether it was derived or not), a value is parenthesized to indicate that it could be inferred from other rows in the same column (e.g., in a \derived \utxo the owner knows the derivation seed, so they can derive the secret- and public-keys). 
}
\label{table:utxos}
\end{table}


\subsection{Signature lifting}\label{intro:lifted}

A key ingredient in our proposed solutions is a novel technique we call \emph{signature lifting}. While we apply this technique to cryptocurrencies, it is useful in various scenarios and therefore is of independent interest. We overview lifted signatures in the current section and defer a more formal treatment to \cref{ssec:zk}.

Chase et al. \cite{CDG+17} present \emph{\picnic}: a construction that transforms any post-quantum one-way function $f$ into a post-quantum signature scheme $\picnic(f)$. In the scheme $\picnic(f)$, the secret-key is generated by uniformly sampling a point $x$ in the domain of $f$, and the corresponding public-key is $f(x)$.

Our contribution is the observation that in some cases, by instantiating the \picnic scheme with an appropriate function, we can \emph{lift} an existing \emph{pre}-quantum signature scheme to a \emph{post}-quantum signature scheme \emph{with the same keys}. Loosely speaking, the required property is that a post-quantum one-way function is applied "somewhere along the way" when generating the public-key from the secret key.

One setting where this property hold is \hashed \utxos, which (as we explain in \cref{ssec:keylift}) can be interpreted as a slight modification of \ecdsa: instead of using the \ecdsa public-key $\pk$ as a public-key, the hash thereof $\HH(\pk)$ is used; when signing a message, $\pk$ is attached to the \ecdsa signature $\sigma$. Since we assume the hash function is post-quantum one-way, we can lift the resulting scheme to obtain a quantum-cautious method for \hashed \utxos. We call this method \emph{key lifting}.

In \cref{ssec:seedlift} we introduce \emph{seed lifting} -- an application of signature lifting towards cautiously spending \derived \utxos. The idea is to instantiate \picnic with the \emph{key-derivation function} used to generate the master secret-key (see \cref{ssec:kdf}). We call this method \emph{seed lifting},

\subsection{Lifted FawkesCoin}

As we explain in \cref{ssec:liftedsize}, one problem with lifted signatures is that the size of such a signature is at least a few dozen kilobytes. Consequently, using such signatures directly will be expensive and will have a strong detrimental effect on the throughput of the network.

We overcome this by combining signature lifting with the \emph{FawkesCoin} protocol of \cite{BM14} to obtain a method we call \emph{lifted FawkesCoin}. 

The core idea of the FawkesCoin protocol is that the payer first posts a \emph{post-quantum commitment} to their transaction, then she \emph{waits} a predetermined number of blocks, and only then does she \emph{reveal} the transaction. Since the public-key is only leaked during the reveal phase, the payer obtains security against a front-running adversary, as the adversary would have to complete an entire commit-wait-reveal phase before the revealed transaction is posted to the blockchain.

The FawkesCoin protocol (along with follow-up work leveraging it towards quantum cautious spending \cite{SIZ+18,IKS19}) has a prohibitive limitation: it does not provide a spam-secure way to pay the miner for including the commitment (see \cref{sec:lfc} for an in-depth discussion of FawkesCoin and the limitations thereof).

Lifted FawkesCoin solves this problem: instead of using a standard signature, the payer uses a lifted signature. The signature can be verified by the miner, but is only posted on the blockchain in case the user fails to reveal the transaction in time. Thus, as long as all parties are honest, lifted FawkesCoin requires as much storage as vanilla FawkesCoin, while protecting miners from spammers.

We fully review the Lifted FawkesCoin protocol for \hashed \utxos in \cref{chap:lfc}. Combined with seed lifting, the Lifted FawkesCoin could be similarly used for spending \derived \utxos (see \cref{cell:lfc_derived}).

\subsection{Restrictive and Permissive FawkesCoin}

\myhl{Restrictive FawkesCoin is an extension of Bonneau's and Miller original algorithm (or rather, their proposed application of this algorithm to Bitcoin) that additionally allows spending }\derived \utxos,\myhl{ using the derivation seed rather than the }\ecdsa\myhl{ public-key as the secret hashed data.}

A pleasant consequence of signature lifting is that it allows us to extend the applicability of the original (non-lifted) FawkesCoin protocol to spend \derived, \naked, and even \lost \utxos.

The idea is simple: allow anyone to spend any \utxo without providing any proof of ownership. However, transactions spending a \utxo this way must provide a \emph{deposit} and wait a long \emph{challenge period} during which the owner of the \utxo is allowed to post a lifted signature as a \emph{proof of ownership}. If they do so, the transaction is invalidated, and the deposit is paid to the owner of the \utxo.

The advantage of the honest user over the adversary is that the adversary cannot know whether a \utxo is \derived or not. An adversary attempting to steal a \utxo will lose the deposit unless the \utxo is non-\derived or \lost.

We flesh out the permissive FawkesCoin protocol in full detail in \cref{sec:fawkes}.

\myhl{Unrestrictive FawkesCoin is a compromise between Permissive and restrictive Fawkescoin, allowing to spend} \naked \utxos\myhl{ but not }\lost \utxos\myhl{. Its advantage over permissive FawkesCoin is that it can be implemented in a soft-fork (see }\cref{table:intro}\myhl{).

We note that these solutions do not allow using the spent }\utxo\myhl{ to pay the transaction fees (we discuss this further in }\cref{sec:lfc}\myhl{). This motivates us to propose a combined solution that allows FawkesCoin and Lifted FawkesCoin to work together (see }\cref{ssec:putting}\myhl{).}

\begin{remark}
    \myhl{The proof of ownership in both permissive and lifted FawekesCoin is an example of a \emph{fraud-proof}. That is, instead of prohibiting foul play we design the protocol such that a fraud is detectable and provable. This allows us to penalize foul play. Fraud-proofs have been used before in cryptocurrencies, so for example Al-Bassam et al. }\cite{AS+21}.
\end{remark}

\subsection{Quantum Canaries}\label{intro:canaries}

A necessary precaution against quantum looters is to decide on a point in time after which pre-quantum spending is no longer allowed. The simplest way to do so is for the community to agree on an arbitrary point. However, since disabling pre-quantum spending effectively implies burning many coins, isolating such a point could be a source of great contention. Moreover, since the progress of quantum technologies is very hard to predict, the chosen time might be too early (leading to unnecessary loss of coins) or too late (allowing quantum adversaries to loot the chain). We propose \emph{quantum canaries}, a mechanism designed to detect, and respond to, increasing quantum capabilities \emph{in consensus}. We stress that our spending methods are agnostic to how the time to initiate them was determined.

In the early 19th century, miners used caged canaries to notice leakages of deadly gasses that were otherwise hard to detect. The miners would go into the mine carrying a cage with a canary, and as long as the canary lived, they would assume that the air was safe to breathe. Using canaries over human beings had two advantages: loss of canary life was considered much preferable to loss of human life, and canaries were more sensitive to lethal gasses such as carbon monoxide, so the death of a canary signified the presence of hazardous gasses well before the human holding the cage would be exposed to a lethal dosage.

Inspired by this solution, \ifnum\acmtops=0\los{in \cref{sec:quantum_canaries} }{}\fi we propose to set up \emph{quantum canaries} on the blockchain. Quantum canaries are bounties locked away behind classically infeasible challenges devised to require quantum computers considerably smaller than those required to steal unsafe coins. 
The cost of paying the bounty (be it by raising funds from the community or by a hard-fork to mint bounty funds) is far smaller than the consequences of a full-blown quantum thief capable of stealing unsafe coins. Killing the canary (that is, claiming the bounty) requires a much smaller "dosage" of quantum capability. Most importantly, the canary would act as a global beacon providing a warning that quantum adversaries are about to emerge. 

Entities capable of killing the canary and claiming the bounty might rather wait for their quantum technology to ripen sufficiently for stealing coins. However, while they wait, another entity could emerge that would kill the canary, depriving the original entity of \emph{both} the bounty and the loot. This dynamic incentivizes quantum entities to claim the bounty. \ifnum\acmtops=0 We provide an analysis of this dynamic in \cref{ssec:canarygame}\fi.

Since this act of ornithicide is recorded on the blockchain, policies could be implemented with respect to the well-being of the canary, such as "do not accept transactions spending \naked \utxos if the canary has been dead for at least ten thousand blocks."

\ifnum\acmtops=1
In the extended version of this paper \cite{SW23} we provide a thorough discussion of quantum canaries and the cryptographic subtleties one has to be conscious of when designing them, and provide a rudimentary game theoretic analysis to provide evidence that an intermediate-scale quantum adversary would prefer to kill the canary rather than wait for their quantum technology to sufficiently scale to loot the network.
\fi

\begin{remark}
    While completing this work, we became aware that quantum canaries were already proposed by Justin Drake in the Ethereum research forum \cite{Dra18} as a general method to protect the network from exploits discovered in widely used cryptographic primitives. \ifnum\acmtops=0 We retain the discussion about canaries as it is useful for our overall solution, and since we expand upon the discussion in \cite{Dra18} in two regards: in \cref{ssec:canarypuzzle} we argue that canaries need to be as similar as possible to the scheme they are protecting, pointing out existing research particularly relevant to the special case where the canary is used to protect \ecdsa signatures from quantum adversaries; in \cref{ssec:canarygame} we provide a game theoretic analysis of entities competing for the bounty furnished by the quantum canaries. We also borrowed the name "quantum canaries" from \cite{Dra18} and adapted the text accordingly.\else In the full version \cite{SW23} we discuss aspects of this idea that are particular to our application.\fi \ifnum\anonymous=0{ We thank Andrew Miller for bringing this to our attention.}\fi

    The idea of posting \utxos that could be spent by solving a cryptographic challenge was already suggested and implemented by Peter Todd in \cite{Tod13}. He created several \utxos, each only spendable by providing a collision in a particular hash function, including \textsf{SHA1}, $\mathsf{RIPEMD\mhyphen 160}$, \sha and several compositions thereof. Of these \utxos, only the first has been claimed so far.
\end{remark}

\secorssec{Properties and Comparison of Spending Methods}


{
\newcommand{\advtg}[1]{
\begin{tikzpicture}[baseline=(char.base)]
\node(char)[draw, color=black!30!green,fill=white, line width=0.4mm,
  shape=rounded rectangle, text=black]
  {#1};
\end{tikzpicture}
}
\newcommand{\dadvtg}[1]{
\begin{tikzpicture}[baseline=(char.base)]
\node(char)[draw, color=black!30!red,fill=white, line width=0.4mm,
  shape= rectangle, text=black]
  {#1};
\end{tikzpicture}
}

\newcommand\link[2]{%
\begin{tikzpicture}[remember picture, overlay, >=stealth, shift={(0,0)}, bend left]
  \draw[->] (#1) to (#2);
\end{tikzpicture}%
}
\setlength{\tabcolsep}{6pt}
\renewcommand{\arraystretch}{1.5} 
\begin{table}[!htb]
\begin{adjustbox}{max width=\los{1.1}{1.3}\textwidth,center}
\begin{tabular}{lccccc@{\hskip 3pt}cllccccccl}
\textbf{} & \multicolumn{4}{c}{\textbf{\begin{tabular}[c]{@{}c@{}}\hyperref[col:cautious]{Cautiously}\\ \hyperref[col:cautious]{Spendable}\end{tabular}}} & \textbf{} & \multicolumn{3}{c}{\textbf{\begin{tabular}[c]{@{}c@{}}\hyperref[col:conf_blocks]{Confirmation}\\ \hyperref[col:conf_blocks]{Times}\end{tabular}}} & \multicolumn{1}{l}{} & \multicolumn{1}{l}{} & \multicolumn{1}{l}{} & \multicolumn{1}{l}{} & \multicolumn{1}{l}{} & \multicolumn{1}{l}{} &  \\ \cline{2-5} \cline{7-9}
 & \rotatebox[origin=r]{270}{\hyperref[col:tx_hashed]{Hashed}} & \rotatebox[origin=r]{270}{\hyperref[col:tx_derived]{Derived}} & \rotatebox[origin=r]{270}{\hyperref[col:tx_naked]{Naked}} & \rotatebox[origin=r]{270}{\hyperref[col:tx_lost]{Lost}} &  & \rotatebox[origin=r]{270}{\hyperref[col:tx_hashed]{Hashed}/\hyperref[col:tx_derived]{Derived}} & \multicolumn{1}{c}{\rotatebox[origin=r]{270}{\hyperref[col:tx_naked]{Naked}}} & \multicolumn{1}{c}{\rotatebox[origin=r]{270}{\hyperref[col:tx_lost]{Lost}}} & \rotatebox[origin=r]{270}{\hyperref[col:tx_size_inc]{Transaction Size Increase}} & \rotatebox[origin=r]{270}{\hyperref[col:spend_thres]{Spendability Threshold}} & \rotatebox[origin=r]{270}{\hyperref[col:pq_utxo]{Works Without PQ \utxos}} & \rotatebox[origin=r]{270}{\hyperref[col:delay]{No Delay Attacks}} & \rotatebox[origin=r]{270}{\hyperref[col:soft-fork]{Soft-Fork}} & \rotatebox[origin=r]{270}{\hyperref[col:online]{Offline Users Not Risked}} &  \\ \cline{1-15}
\tikzmark{cur}Current~State & \xmark & \xmark & \xmark &  &  & $\times 1$ & \multicolumn{1}{c}{} & \multicolumn{1}{c}{} & $\times 1$ & $\times 1$ & \cmark & \cmark &  & \cmark &  \\
\tikzmark{FC}\hyperref[row:fc]{FawkesCoin} &  &  &  &  &  &  &  &  &  &  &  &  &  & \multicolumn{1}{l}{} &  \\
\quad\tikzmark{CFC} \hyperref[row:cfc]{Restrictive} & \advtg{\chr{cell:cfc_hashed}{\cmark}} & \advtg{\chr{cell:cfc_derived}{\cmark}} &  &  &  & \chr{cell:cfc_hdconf}{$\sim\times 10$} &  &  & \chr{cell:cfc_tx}{$\sim + 10$} & \chr{cell:cfc_spendthres}{$\sim\times 1$} & \dadvtg{\chr{cell:fc_pqutxo}{\xmark}} & \cmark & \chr{cell:cfc_softfork}{\cmark} & \cmark &  \\
\quad\tikzmark{NFC} \hyperref[row:nfc]{Unrestrictive} & \chr{cell:nfc_hashed}{\cmark} & \chr{cell:nfc_derived}{\cmark} & \advtg{\chr{cell:nfc_naked}{\hcmark}} &  &  & \chr{cell:nfc_hdconf}{$\sim\times 10$} & \multicolumn{1}{c}{\chr{cell:nfc_nconf}{$\sim\times 10^3$}} &  & \chr{cell:nfc_tx}{$\sim + 10$} & \chr{cell:nfc_spendthres}{$\sim\times 1$} & \chr{cell:fc_pqutxo}{\xmark} & \cmark & \chr{cell:nfc_softfork}{\cmark} & \dadvtg{\chr{cell:pfc_offline}{\xmark}} &  \\
\quad\tikzmark{PFC} \hyperref[row:pfc]{Permissive} & \chr{cell:pfc_hashed}{\cmark} & \chr{cell:pfc_derived}{\cmark} & \chr{cell:pfc_naked}{\hcmark} & \advtg{\chr{cell:pfc_lost}{\hcmark}} &  & \chr{cell:pfc_hdconf}{$\sim\times 10$} & \multicolumn{1}{c}{\chr{cell:pfc_nconf}{$\sim\times 10^3$}} & \multicolumn{1}{c}{\chr{cell:pfc_lconf}{$\sim\times 10^3$}} & \chr{cell:pfc_tx}{$\sim + 10$} & \chr{cell:pfc_spendthres}{$\sim\times 1$} & \chr{cell:fc_pqutxo}{\xmark} & \cmark & \dadvtg{\chr{cell:pfc_softfork}{\xmark}} & \chr{cell:nfc_offline}{\xmark} &  \\
\tikzmark{LS}\hyperref[row:ls]{Lifted~Spending} & \chr{cell:ls_hashed}{\cmark} & \chr{cell:ls_derived}{\cmark} &  &  &  & \chr{cell:ls_hdconf}{$\times 1$} &  &  & \dadvtg{\chr{cell:ls_tx}{$\sim\times 10^3$}} & \chr{cell:ls_spendthres}{$\sim\times 10^3$} & \advtg{\chr{cell:ls_pqutxo}{\cmark}} & \cmark & \chr{cell:ls_softfork}{\xmark} & \cmark &  \\
\tikzmark{LFC}\hyperref[row:lfs]{Lifted~FawkesCoin} & \chr{cell:lfc_hashed}{\cmark} & \chr{cell:lfc_derived}{\cmark} &  &  &  & \chr{cell:lfc_conf}{\dadvtg{$\sim\times 10$}} &  &  & \advtg{\chr{cell:lfc_tx}{$\sim\times 1$}} & \chr{cell:lfc_spendthres}{$\sim\times 10^3$} & \chr{cell:lfc_pqutxo}{\cmark} & \chr{cell:lfc_delay}{\hcmark} & \chr{cell:lfc_softfork}{\xmark} & \cmark & \link{LS}{CFC}
\end{tabular}
\end{adjustbox}
\caption{
    A comparison of the quantum-cautious spending methods with the current state as a baseline. Each row represents a method, and the columns represent different properties of the methods reviewed below the table. \myhl{Each method is compared to the preceding method, except the ``Lifted Spending'' method is compared with \emph{restrictive} FawkesCoin, as indicated by the arrow}. Green circles and red boxes designate what we consider the most significant advantages and disadvantages respectively of each method with respect to the method it is being compared to. Cells that are not applicable (e.g., using the current state to spend a \lost \utxo) are left blank. The meaning of each column is discussed below. The \hcmark symbol indicates that the property corresponding to the column is partially satisfied in a manner explicated in the description of the relevant column. The table headers and cells hyperlink to relevant sections in the text.
    \myhl{All methods below ``Current State'' are original to the current work. The FawkesCoin protocol was first introduced in }\cite{BM14}\myhl{, but even our most restrictive version expands on the original protocol by allowing spending }\derived \utxos\ifnum\masterthesis=1{.}\else\myhl{. See} \cref{ssec:int_rel}\myhl{ for details}.\fi
}

\label{table:intro}
\end{table}
}

\ifnum\masterthesis=0
In this work we provide several different quantum cautious spending methods, that can be compared to each other through several different properties. In this section we describe the properties we consider relevant, which are summarized in \cref{table:intro}. 
\fi

We use Bitcoin as the baseline for our comparison. In particular, we consider the number of confirmation blocks in the "current state" to be six blocks, which is the default number of blocks required by the standard Bitcoin wallet.

We stress that the properties stated in \cref{table:intro} are correct under the assumption that no two solutions are implemented together. In \cref{intro:compat} we discuss how combining different solutions might render them insecure, and propose ways to combine them securely.

\paragraph*{Cautiously Spendable} \label{col:cautious} A pre-quantum \utxo is considered \emph{cautiously-spendable} with respect to a given method if implementing this method allows spending it in a way that is not vulnerable to quantum attackers. 
We use \hcmark to denote the subtle scenario where a \utxo can be stolen, but an attempt to do so requires the adversary to take a considerable monetary risk. This happens in unrestrictive and permissive FawkesCoin for \naked and \lost \utxos. If an adversary claims such a \utxo, they will manage to steal the coins therein. However, if an adversary claims a non-\hashed \derived \utxo, the rightful owner of the \utxo would post a fraud-proof, and the adversary would lose her deposit. Hence, the owners of \naked and \lost \utxos are protected by the adversary's inability to distinguish them from \derived \utxos.

\paragraph*{Confirmation Times}
The number of blocks a user must wait after a transaction has been posted to the network before the transaction is considered confirmed. We write the multiplicative factor by which the number of blocks increases compared with the current state (where we consider the current confirmation times in Bitcoin to be 6 blocks). In methods where spending a \utxo requires posting two messages to the blockchain (a commit and a reveal), we count the number of blocks since the first message.\label{col:conf_blocks}

\paragraph*{Transaction Size Increase} The size of a transaction, compared with the size of a current transaction spending \ecdsa signed \utxos. Multiplicative notation $\times x$ means that the size of a transaction is $x$ times that of a current transaction. 
Additive notation $+ x$ means that spending several \utxos is possible by adding a fixed amount of data larger than a current transaction by a factor of $x$. In particular, the amount of added data does not depend on the number of \utxos spent.\label{col:tx_size_inc}

\paragraph*{Spendability Threshold}\label{col:spend_thres} The minimal value a \utxo should have so it could be spent using this method without the user or the miner risking losing money. The values in the table represent the factor of increase compared to the current spendability threshold. In most methods, this is the same as the cost of spending the transaction (i.e., the transaction fee). However, in Lifted FawkesCoin, the spendability threshold is actually much higher. This follows since the \utxo must be valuable enough to cover the costs of posting a proof of ownership in case the user fails to reveal the transaction. In methods where many \utxos may be spent at once (e.g., in non-Lifted FawkesCoin, one can post \emph{many} commitments in the payload of a \emph{single} post-quantum \utxo), we consider the threshold in the "many \textsf{UTXO}s" limit. We stress that spendability threshold is not the same as dust (see \ifnum\masterthesis=0{\cref{par:dust}}\else{\cref{ssec:perlim_bitcoin_dust}}\fi): dust refers to \utxos whose value is in the same order of magnitude as it would cost to spend them, \emph{regardless} of what method is used; on the other hand, a \utxo could be valuable enough to be above the spendability threshold of one method while being below the spendability threshold of another method.

\paragraph*{Works Without Post-Quantum \utxos}\label{col:pq_utxo} Some of our methods require the user to have a post-quantum \utxo (namely, a \utxo whose address is associated with a post-quantum signature public key), for paying fees or deposits. Such methods are marked with \xmark, while methods that work even without a post-quantum \utxo are marked with \cmark.

\paragraph*{No Delay Attacks}\label{col:delay} A \emph{delay attack} is a way to make an arbitrary \utxo unspendable for a fixed but meaningful period of time (say, a few hours) at a low cost (say, the cost of posting a transaction on the blockchain). We do not consider methods that allow cheap delay attacks suitable. Hence, no row in \cref{table:intro} has \xmark\xspace in this column. Lifted FawkesCoin is the only method marked with \hcmark\xspace to indicate that in this method, delay attacks are possible but are automatically detected, and the subject of the attack is compensated by the perpetrator.

\paragraph*{Can be implemented in a Soft-fork}\label{col:soft-fork} It is true if the method can be implemented as a soft-fork, that is, in a way that does not require the entire network to upgrade their nodes (see \ifnum\masterthesis=0{\cref{par:softfork}}\else{\cref{ssec:perlim_bitcoin_forks}}\fi).

\paragraph*{Offline Users not Risked} \label{col:online} Some methods require owners of \derived \utxos which are not \hashed to monitor the network for attempts to spend their coins. We use \cmark\xspace in methods for which this is \emph{not} the case.


\subsection{Putting it all Together}\label{ssec:putting}

It is evident from \cref{table:intro} that different protocols provide different trade-offs, which are both suitable for different users. However, combining solutions carelessly can render them insecure. For example, allowing spending pre-quantum \utxos as they are spent today renders FawkesCoin insecure, since once the transaction is revealed the adversary can use the public-key without having to wait. For identical reasons, implementing lifted spending also renders FawkesCoin insecure. We therefore recommend \emph{not} to implement lifted spending.

Security issues also arise when trying to combine FawkesCoin and lifted FawkesCoin together. In \cref{ssec:combining} we explain in detail how FawkesCoin and Lifted FawkesCoin can be securely combined by dividing the network into \emph{epochs}, during which only one protocol is allowed.

We propose to implement a quantum canary. The death of the canary would initiate a countdown, after which the consensus rules of the \emph{quantum era} are enforced:
\begin{itemize}
    \item FawkesCoin and Lifted FawkesCoin are activated and rotated in epochs.
    \item Directly spending pre-quantum \utxos becomes prohibited.
    \item Permissive FawkesCoin is allowed for \naked \utxos whose address was not posted to the blockchain prior to 2013.
\end{itemize}

In \cref{intro:proposal} we specify a fully fleshed protocol and provide and motivate concrete values for the canary bounty, countdown length, epoch lengths, etc.

\subsection{Paper organization}

The paper is organized to introduce the technique of signature lifting, followed by an application thereof to construct the lifted FawkesCoin protocol. In \cref{prelim:ds} and \cref{prelim:bitcoin} we provide a brief exposition of digital signatures and blockchains, respectively, sufficient to render the current version self-contained. In \cref{ssec:zk} we define the notion of a signature lifting, and apply it to two common types of signature schemes: schemes with a hashed public-key, and schemes whose secret-key was created by a key-derivation function. In \cref{sec:lfc} we introduce and thoroughly discuss the lifted FawkesCoin protocol. In \cref{sec:fawkes} we present permissive FawkesCoin, and discuss its advantages and disadvantages over lifted FawkesCoin. Finally, in \cref{intro:proposal}, we propose a holistic solution by carefully combining quantum canaries, permissive FawkesCoin, and lifted FawkesCoin into a concrete protocol.

This text is a summarized version of our work, where many details and discussions were omitted for the sake of brevity, with the full version available online \cite{SW23}. While the current version is mostly self contained, the full version provides a more in-depth discussion of most aspects of our work, including a more complete discussion of the quantum threats on blockchain and the solutions proposed so far, a fully detailed proof of the security of seed-lifting (provided with a formal introduction and treatment of key-derivation functions and HD wallets), and an extended discussion of quantum canaries and the game theory of the incentives thereof. Throughout the text, we provided references to the full version when applicable.

\secorssec{Related Works}\label{ssec:int_rel}

The loot currently waiting for a quantum adversary on the Bitcoin network was explored in \cite{IKK20}.

The observation that secure signature schemes could be constructed from a one-way function (not necessarily a cryptographic hash function) given reliable time-stamping (or append-only log) was first made by Anderson et al. \cite{ABC+98}, who used it to construct \emph{Fawkes signatures}. A crucial property of this construction is that the signature only requires a single preimage as a key and a single image as a signature, unlike Lamport signatures and similar hash-based one-time signatures whose signatures are typically much larger.

\myhl{FawkesCoin was first introduced by Bonneau and Miller }\cite{BM14}\myhl{, who realized that the blockchain itself could be used as a reliable time-stamping service (at least for blocks sufficiently temporally separated) and thus can be used to implement Anderson et al.'s Fakwes signatures. Bonneau and Miller concede that the security properties of Bitcoin are preferable to their protocol (especially in the way it responds to chain forks) but argue that their protocol could be combined with Bitcoin to mitigate ``a catastrophic algorithmic break of discrete log on the curve P-256 or rapid advances in quantum computing.'' The authors note the problem of paying fees for transfer messages and suggest a few possible solutions. However, these solutions do not prevent denial-of-service attacks, as we discuss in }\cref{sec:lfc}.

As was noted in \cite{But13}, spending \hashed \utxos directly is not quantum-cautious, as a quantum adversary listening to the mempool could recover the public-key from the transaction and use it to post a competing transaction spending the same \utxo. Applying the ideas of \cite{BM14} to solve the problem of safely migrating funds from a pre-quantum to a post-quantum \utxo was first discussed in \cite{SIZ+18,IKS19} (the authors thereof point out several threads from Twitter and the Bitcoin developer mailing list to which they attribute first suggesting this approach) by means of \emph{key surrogacy}. They use the techniques of \cite{BM14} to build a mechanism allowing a user to publicly attach to any existing pre-quantum public-key a post-quantum \emph{surrogate key}, so that the network would expect \utxos signed with the pre-quantum key to pass verification with respect to the post-quantum key. This approach achieves some minor improvements upon FawkesCoin but still only allows spending \hashed \utxos and fails to address the problem of paying the fees from the pre-quantum \utxo.

Coladangelo and Sattath \cite[Section~4.1]{CS20} originally proposed the idea of allowing users to claim lost coins by leaving a deposit, which could be claimed by the rightful owner in case of a theft attempt. We use the same approach in unrestrictive and permissive FawkesCoin (see \cref{sec:fawkes}).

The technique of using arbitrary post-quantum one-way functions to construct post-quantum digital signatures was presented in \cite{CDG+17},\los{ we review this work in more depth \cref{ssec:zk}.}{} A follow-up work \cite{KZ20} reduces the signature size considerably.

The canary introduced in \cref{intro:canaries} has some shared features with bug bounties. Breidenbach et al.~\cite{BDTJ18} introduces \emph{Hydra}, a systematic approach for bug bounties for smart-contracts. Their approach is tailored to finding implementation bugs in smart contracts and, therefore, cannot be used in our setting. Ref.~\cite{RMD+21} presents and analyzes a framework for prediction markets, which can be used to estimate the future risk of an attack on ECDSA, which would break a cryptocurrency. The motivation and goals for their prediction market are far removed from the quantum canary discussed in this work.

\ifnum\anonymous=0\ifnum\lipics=0
\secorssec{Acknowledgements}

We would like to thank Andrew Miller for valuable discussions and, in particular for pointing out the existence of quantum canaries, and Ori Newman and Elichai Turkel for their invaluable feedback and discussion.

This work was supported by the Israel Science Foundation (ISF) grant No. 682/18 and 2137/19, and by the Cyber Security Research Center at Ben-Gurion University.
\BeforeBeginEnvironment{wrapfigure}{\setlength{\intextsep}{0pt}}

\begin{wrapfigure}{r}{90px}
    \includegraphics[width=40px]{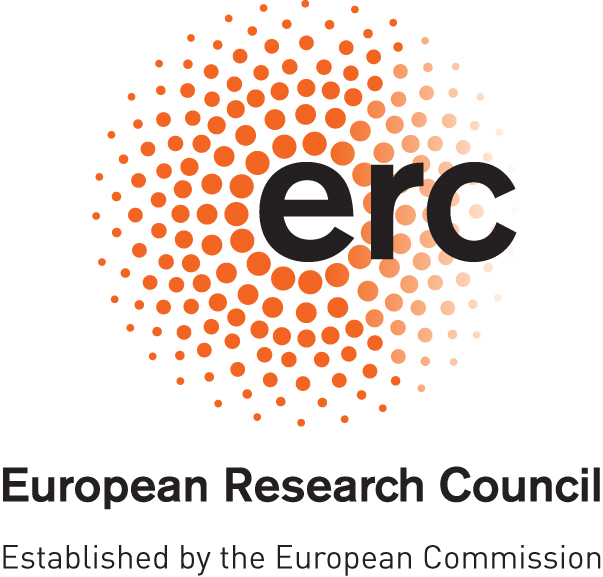}  \includegraphics[width=40px]{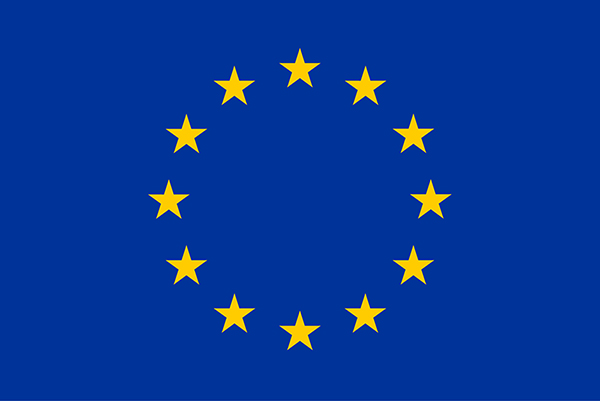}
\end{wrapfigure}This work was funded by the European Union (ERC-2022-COG, ACQUA, 101087742). Views and opinions expressed are however those of the author(s) only and do not necessarily reflect those of the European Union or the European Research Council Executive Agency. Neither the European Union nor the granting authority can be held responsible for them.
\fi\fi

\ifnum\acmtops=1
\section{Preliminaries: Digital Signatures}\label{prelim:ds}

A \emph{digital signature scheme} is a cryptographic primitive that allows users to sign messages such that only they can sign the message, but anyone can verify the authenticity of the message. In digital signature schemes, the signer creates a \emph{secret signature key} that can be used to sign messages and a \emph{public-key} that can be used to verify these messages.

\ifnum\masterthesis=0
\subsection{Security of Digital Signatures}\label{ssec:dssec}
\fi

The security notion we consider is \emph{existential unforgeability under chosen message attack} ($\mathsf{EUF\mhyphen CMA}$). In this notion, the adversary is given access to the public-key as well as access to a \emph{signing oracle} which she could use to sign any message she wants. Her task is to produce a signature for any message she did not use the oracle to sign. A scheme is \emph{\seufcma secure} if an efficient adversary cannot achieve this task with more than negligible probability (for a more formal treatment of the security of digital signatures we refer the reader to \cite{Gol04,KL14}). 


Extending the security notions of digital signatures to accommodate general quantum adversaries is not quite straightforward. The main difficulty is in the setting where the adversary has access to sign messages of their own, and they can sign a \emph{superposition} of messages. The established classical notions of security do not generalize directly to this setting, since the notion of "a message she did not use the oracle to sign" becomes ill-defined when discussing superimposed queries.

However, in the setting of cryptocurrencies, the signature oracle reflects the adversary's ability to read transactions signed by the same public-key off the blockchain. This ability is captured even when assuming that the adversary, albeit quantum, may only ask the oracle to sign classical messages. A scheme that remains secure against quantum adversaries with classical oracle access is called \emph{post-quantum} $\mathsf{EUF\mhyphen CMA}$ secure. For brevity, we use the term $\mathsf{EUF\mhyphen CMA}$ security to mean post-quantum $\mathsf{EUF\mhyphen CMA}$ security, unless stated otherwise.

\subsection{Elliptic Curve based Signatures}\label{sssec:ec}

The \ecdsa and Schnorr signature schemes are based on a particular mathematical object called an \emph{elliptic curve}. Bitcoin uses the \ecdsa signature scheme instantiated with the \secp curve, which is considered to admit 128-bit security against \emph{classical} attackers \cite{Bit22b}. The Taproot update \cite{WNT20} replaces \ecdsa with Schnorr signatures instantiated with the same \secp curve, which is also considered to admit 128-bit security \cite{WNR20}. The key-generation procedure is hence identical in the \ecdsa and Schnorr variants implemented in Bitcoin.

For the purpose of this work, it is not required to understand the details of how neither of the schemes work, it suffices to know that:
\begin{itemize}
    \item Given the secret-key, one can efficiently calculate the corresponding public-key. We use $\pkec$ to denote the function mapping secret-keys to public-keys.
    \item Under common hardness assumption, both schemes are secure. However, these assumptions do not hold when assuming a quantum adversary. In fact, an efficient quantum algorithm is known for inverting $\pkec$ \cite{BL95}, proving that an efficient adversary with access only to the public-key (and to no signatures whatsoever) can recover the secret-key, whereby she can sign arbitrary messages.
\end{itemize}

\subsection{\picnic Signatures}\label{ssec:picnicov}

The principal tool we use for signature lifting is the \emph{\picnic} signature scheme of Chase et al. \cite{CDG+17}. The \picnic scheme signature scheme can be instantiated using \emph{any} post-quantum one-way function $f$  to obtain a signature scheme which is post-quantum \seufcma secure in the \textsf{QROM}. In the obtained scheme, a secret-key is a random point $x$ in the domain of $f$, and the corresponding public key is $f(x)$.

The \picnic scheme, instantiated with a particular block cipher called \emph{LowMC}, was submitted to NIST for standardization. Their design prevailed the first two rounds of the competition. It was decided that \picnic will not proceed to the third round due to the novelty of techniques it applies compared with other candidates, however, it was decided to retain \picnic as an alternative candidate \cite{AASA+20}. This means that, while \picnic was not chosen to be NIST standardized, it successfully withstood heavy scrutiny.

\section{Preliminaries: Bitcoin and Blockchain}\label{prelim:bitcoin}

In this section, we briefly introduce some of the aspects of Bitcoin relevant to our discussion. We assume the reader is familiar with core concepts of Bitcoin such as transactions and \utxos, chain reorganizations, etc. For a review of these concepts, we refer the reader to \cite{NBF+16}. For a more complete review of the aspects of Bitcoin and blockchains relevant to the discussion, see the full version \cite{SW23}.

\subsection{Quantum Threats on Blockchains}\label{ssec:qthreats}

Roughly speaking, quantum computers affect Bitcoin on two different fronts: quantum mining and attacks on pre-quantum cryptography.

A common misconception is that quantum computers have no drastic effects on Bitcoin mining (beyond increased difficulty due to Grover's quadratic speedup). This was debunked in~\cite{Sat20,LRS19}. Our work is orthogonal to aspects related to quantum mining, especially since quantum attacks on \secp are projected to occur a few years before quantum mining starts~\cite{ABL+17}.  

The most immediate risk is in the form of \utxos with leaked public-keys. A quantum adversary could use the public-key to sign arbitrary messages and, in particular, spend any \utxo whose address is this public-key. The public-key can be exposed in a variety of ways, including (but not limited to) writing it in the transaction in-the-clear (as in \ptpk and Taproot \cite{WNT20} addresses), accepting coins to the same hashed address from several sources and then spending one of these \utxos (whereby revealing the non-hashed secret-key), spending coin on one fork of the chain (whereby revealing the key for the same \utxo on the other fork), and so on. A survey by Deloitte estimates \cite{Del22} that, as of 2022, the addresses of at least 2 million hashed \utxos are exposed this way. As of August 2023, about  25\% of newly created \utxos are spent to Taproot addresses \cite{Tra23}.

\subsection{Hierarchical Deterministic wallets}\label{ssec:hdwallet}\label{ssec:kdf}

\emph{Hierarchical Deterministic wallets}, commonly abbreviated as \emph{HD wallets}, were first introduced in BIP-32 \cite{Wui13} to simplify the task of storing secret-keys. The idea is that a single secret called the \emph{seed}, that is sampled from a distribution with sufficient entropy (commonly in the form of a \emph{mnemonic phrase}, as specified in BIP-39 \cite{PRV+13}), is used to derive \emph{many} secret-keys using a \emph{key-derivation function}. In the full version \cite{SW23} we provide a formal treatment of key-derivation functions and hardware wallets. For the sake of exposition, we introduce the function \der, mapping a seed $s$ and a derivation path $P$ to an \ecdsa secret-key, without specifying it explicitly.

\subsection{Software Forks}\label{par:softfork} A \emph{software fork} is a change to the code that nodes are expected to run that affects the conditions under which a block is considered valid. There are two types of software forks, a \emph{soft-fork} and a \emph{hard-fork}. The difference is that in a hard-fork, there exist blocks that the new version considers valid while the old version does not. In order to adopt a soft-fork, only the miners are required to update their nodes, and non-mining nodes will operate correctly even with the outdated version. A hard-fork, on the other hand, causes the chain to split into two separate chains that can not accept each other's blocks. Notable examples of hard-forks include Bitcoin Cash (forked from Bitcoin) and Ethereum Classic (forked from Ethereum).

One particularly relevant scenario where hard-forks are required, is when the protocol allows spending coins without access to the corresponding secret-key.

\subsection{Dust} \label{par:dust}  \utxos whose value is too small to cover the costs of spending them are called \emph{Dust}. The difference between dust and unspendability is that a \utxo can be spendable with respect to one spending method but not the other. For a \utxo to be considered \emph{dust}, it must not be spendable using \emph{any} method available. Since the transaction fees fluctuate, the set of dust \utxos changes in time.
\fi

\section{Signature Lifting and Lifted Spending}\label{ssec:zk}

Chase et al. \cite{CDG+17} introduce a signature scheme called \emph{\picnic}, that could be instantiated using \emph{any} post-quantum one-way function $f$ to obtain a signature scheme which is post-quantum \seufcma secure in the \textsf{QROM}\ifnum\acmtops=0 (see \cref{ssec:rom})\fi. In the obtained scheme, a secret-key is a random point $x$ in the domain of $f$, and the corresponding public key is $f(x)$.

In \cref{intro:lifted} we gave an overview of how this might be useful to protect procrastinators. The purpose of the current \los{section}{appendix} is to provide an explicit construction and formally treat its security properties.

We stress that while the resulting method is secure, it is highly inefficient in terms of block space (see \cref{ssec:liftedsize}). Therefore, we do not recommend using it directly, but we rather use it as a building block for the \emph{lifted FawkesCoin} protocol we present in \cref{chap:lfc}.

\subsection{Lifting Signature Schemes}
\ifnum\masterthesis=0{Recall that a \emph{correct signature scheme} is a tuple of three polynomial time procedures $\ds = (\keygen, \sign, \ver)$ such that if $(\pk,\sk)\gets \keygen(\secparam)$ then it holds for any $m$ that if $\sigma\gets \sign_\sk(m)$ then $\ver_\pk(m,\sigma)$ accepts (see \ifnum\masterthesis=0\cref{prelim:ds}\else\cref{sec:prelim_ds}\fi).}\fi

\begin{definition}
    Let $\ds = (\keygen, \sign, \ver)$ be a correct digital signature scheme\ifnum\masterthesis=1{ (recall \cref{sec:prelim_ds})}\fi. A \emph{lifting} of $\ds$ is two procedures $\widetilde{\sign},\widetilde{\ver}$ such that $(\keygen, \widetilde{\sign},\widetilde{\ver})$ is a correct digital signature scheme.
\end{definition}

In \ifnum\masterthesis=0{\cref{prelim:ds} }\else{\cref{ssec:prelim_ds_sec} }\fi we described the notion of \seufcma security, we extend this notion in a way appropriate for liftings:

\begin{definition}\label{defn:lifting}
    Let $\ds = (\keygen, \sign, \ver)$ and let $(\widetilde{\sign},\widetilde{\ver})$ be a lifting, we define the \sseuflcma security game as follows:
    \begin{itemize}
        \item The challenger $\sC$ samples $(\sk,\pk)\gets \keygen(\secparam)$ and gives $\pk$ and $\secparam$ as input to the adversary $\sA$
        \item $\sA$ is allowed to make \emph{classical} queries both the oracles $\sign_\sk$ and $\widetilde{\sign}_\sk$
        \item $\sA$ outputs a tuple $(m,\sigma)$
        \item $\sA$ wins the game if $\widetilde{\sign}_\sk$ was never queried on $m$, and $\widetilde{\ver}_\pk(m,\sigma)$ accepts.
    \end{itemize}
    Note that if we remove the access to the oracle $\sign_\sk$ we recover the post-quantum \seufcma game (see \ifnum\masterthesis=0\cref{ssec:dssec}\else\cref{def:euf-qcma}\fi) for the scheme $(\keygen, \widetilde{\sign},\widetilde{\ver})$.
    
    The lifting is a \emph{post-quantum (strong) lifting} if the winning probability of $\sA$ in the \seufcma game (\sseuflcma game) is $\negl$ for any \qpt $\sA$.
\end{definition}


\subsection{Key-lifted Signature Schemes}\label{ssec:keylift}


Let $(\keygen,\sign,\ver)$ be a correct signature scheme (e.g. the \ecdsa (or Schnorr) digital signature scheme over the curve \secp), let $\HH$ be a hash function modeled as a random oracle, and define the following modified scheme:
\begin{itemize}
    \item $\keygen'(\secparam)$: samples $(\pk,\sk)\gets \keygen(\secparam)$, and outputs $(\sk,\pk' = \HH(\pk))$. 
    \item $\sign'_{\sk}(m)$: outputs $\sigma' = (\sigma,\pk)$ where $\sigma \gets \sign_{\sk}(m)$ (recall that $\pk = \pkec(\sk)$, see \ifnum\masterthesis=0\cref{sssec:ec}\else\cref{ssec:prelim_ds_ec}\fi).
    \item $\ver'_{\pk'}(m,\sigma')$: interprets $\sigma'$ as $(\sigma,\pk)$, accepts iff $\HH(\pk) = \pk'$ and $\ver_{\pk}(m,\sigma)$ accepts.
\end{itemize}

This modification reflects how \ecdsa is used in Bitcoin, where the user typically posts a \sha hash of their public-key, and only posts the public-key when signing a message.

\ifnum\masterthesis=0{

We now define a lifting of $(\keygen',\sign',\ver')$, using $\picnic(\HH)$ to denote the $\picnic$ scheme instantiated with $\HH$ (see \ifnum\masterthesis=0\cref{ssec:picnicov}\else\cref{ssec:prelim_ds_picnic}\fi):
}\else{
We use the \picnic scheme (recall \cref{ssec:prelim_ds_picnic}) to define a lifting of $(\keygen',\sign',\ver')$:
}\fi
\begin{itemize}
    \item $\widetilde{\sign}_{\sk} \equiv \picnic(\HH).\sign_{\pkec(\sk)}$
    \item $\widetilde{\ver} \equiv \picnic(\HH).\ver$
\end{itemize}
It is straightforward to check that $(\widetilde{\sign},\widetilde{\ver})$  is a lifting of $(\keygen',\sign',\ver')$. We call this scheme the \emph{key-lifted scheme}.

\begin{proposition}\label{prop:keylift}
    If $\HH$ is modeled as a random-oracle, then $(\widetilde{\sign},\widetilde{\ver})$ is a post-quantum lifting (see \cref{defn:lifting}).
\end{proposition}

\begin{proof}

    Let $\sA$ be an \seufcma adversary for $(\keygen',\widetilde{\sign},\widetilde{\ver})$ which wins with probability $\varepsilon$, we construct an \seufcma adversary $\sA_p$ for $\picnic(\HH)$ which wins with probability $\varepsilon$.

    $\sA_p$ uses the $\pk$ she received from $\sC$ as input to $\sA$, and responds to oracle queries on $m$ by querying her own oracle. Since $\pkec$ is a bijection and $\sk$ distributes uniformly it follows that $\pkec(\sk)$ also distributes uniformly. Hence, the view of $\sA$ distributes identically in the simulation and in the real \seufcma game. Hence, with probability $\varepsilon$ the output $(\sigma,m)$ wins the \seufcma game for $\picnic(\HH)$. But since $\sA$ and $\sA_p$ make the same queries, $(\sigma,m)$ is winning output for $\sA$ iff it is a winning output for $\sA_p$.

    We now invoke \cite[Theorem~4]{CLQ20}, which asserts that a random oracle is post-quantum one-way (even against a non-uniform \qpt adversary with a quantum advice), and \cite[Corollary~5.1]{CDG+17}, which asserts that $\picnic(\HH)$ is post-quantum \seufcma secure whenever $\HH$ is post-quantum one-way, to conclude that $\varepsilon = \negl$.
\end{proof}

We note that $(\widetilde{\sign},\widetilde{\ver})$ is \emph{not} a strong lifting, since the output of $\sign'_{\sk}$ (on any message) contains a copy of $\pk$, from which a \qpt adversary can calculate $\sk$.


\subsection{Seed Lifted Signature Schemes}\label{ssec:seedlift}

\Derived \utxos (either \naked or \hashed) also have the property that a post-quantum hash was applied "somewhere along the way", namely, to produce the master secret-key from the seed (see \ifnum\masterthesis=0{\cref{ssec:hdwallet}}\else{\cref{ssec:prelim_other_hd}}\fi). This suggests that lifted signatures could also be provided to such addresses by using the seed-phrase as a secret-key.



The main difficulty in realizing this approach is that the \emph{same} seed could be used to generate many \emph{different} public-keys. More precisely, the function mapping the secret-key to the public-key depends on a \emph{derivation path}, and we cannot reasonably assume that the adversary only has access to signatures corresponding to a particular derivation path. Security against an adversary augmented with access to signatures corresponding to many different derivation paths doesn't follow in a black-box fashion from the security of \picnic (indeed, \picnic could be modified to leak a single bit of the secret-key, depending on the derivation path. This does not affect security against a standard adversary, but allows the augmented adversary to recover the secret-key completely).

Like in key-lifting, we modify the $\ecdsa$ scheme to reflect how such addresses are used in practice. Unlike key-lifting, here the \emph{public}-key remains the same, whereas we modify the \emph{secret}-key. Instead of using the \ecdsa secret-key $\sk$, we use the seed $s$ and password $pw$ used to generate the master public-key (see \ifnum\masterthesis=0\cref{ssec:kdf}\else{\cref{ssec:prelim_other_kdf}}\fi) and the derivation path $P$ corresponding to $\pk$ as a secret-key. Since $\sk$ can be easily calculated from $(s,pw,P)$, defining the modified scheme is straightforward (\myhl{given a string $s$ of even length, we use $s_L$ to denote the left half of $s$}):

\begin{itemize}
    \item $\keygen(\secparam)$ samples $\sk = (s,pw,P)$ where $s\gets \zo^\secpar$, and $pw$ and $P$ are chosen arbitrarily such that $P$ is independent from $(s,pw)$, sets $\pk = \pkec(\der(\kdf(s,pw),P)_L)$, outputs $(\sk,\pk)$.
    \item $\sign_{(s,pw,P)}(m)$ outputs $\sigma\gets\ecdsa.\sign_{\sk'}(m)$ where $\sk' = \der(\kdf(s,pw),P)_L$.
    \item $\ver \equiv \ecdsa.\ver$
\end{itemize}
Where $\kdf = \HH^{2048}$ (that is, $2048$ applications of $\HH$) as in BIP-32, and $\der$ is the key-derivation function \ifnum\masterthesis=0 (for a formal treatment of BIP-32 and the key derivation function used therein, see the full version \cite{SW23})\else (see \cref{ssec:sl_seed_hd})\fi.

We now define a lifting for this scheme. The lifting is constructed in such a way that its security assures that the adversary is unable to forge a signature verifiable by \emph{any} key in the wallet, not just the key used by the honest user. Set $\kdf^{pre} = \HH^{2047}$ and $\kdf^{pq} = \HH$ (here $\kdf^{pq}$ captures the "post-quantum step" we will apply signature lifting to, and $\kdf^{pre}$ captures any other processing that preceded it).

\begin{definition}\label{con:seedlift}
    The \emph{seed-lifting} of the scheme above is:
    \begin{itemize}
        \item $\widetilde{\sign}_{(s,pw,P)}$ outputs $(\sigma,\msk,P)$ where $\sigma \gets \picnic(\kdf^{pq}).\sign_{\kdf^{pre}(s,pw)}(m,P)$ and $\msk = \kdf(s,pw)$ 
        \item $\widetilde{\ver}_\pk(m,(\sigma,\msk,P))$ accepts if $\pkec(\der(\msk,P)_L) = \pk$ and $\picnic(\kdf^{pq}).\ver_{\msk}((m,P),\sigma)$ accepts.
    \end{itemize}
\end{definition}

We wish to prove the following:

\begin{theorem}\label{thm:seedlift}
    If $\HH$ is modeled as a random oracle, then the seed-lifting is a strong lifting.
\end{theorem}


We provide an overview of the proof here, and delegate a complete formal argument to \ifnum\acmtops=0\cref{ssec:seedlift}\else the extended version \cite{SW23}\fi. The crux of the argument is that the key-derivation function used in HD wallets admits a particular structure that implies that we need only consider three possibilities:
\begin{itemize}
    \item The derivation path provided by the adversary is identical to one of the derivation paths exposed to him.
    \item The derivation path provided by the adversary is a suffix of one of the derivation paths exposed to him.
    \item The derivation path provided by the adversary is a prefix of one of the derivation paths exposed to him.
\end{itemize}
The reason for that is that any other type of output by the adversary immediately implies that the adversary can find a collision in the underlying hash function. The proof proceeds by showing that this also holds in the three cases above. The two first cases are rather straightforward. however, the third case requires proving that the \picnic scheme admits a particular extractability property, and use it to show a collision is extractable. In particular, the proof appeals to the details of the construction and is not a black-box argument.

\subsection{Size of Lifted Signatures}\label{ssec:liftedsize}
The scheme resulting from lifting a one-way function might admit prohibitively large signatures. The size of the signature is a consequence of the size of the circuit calculating the one-way function $f$.

Chase et al. \cite{CDG+17} isolate a one-way function particularly suitable for the task (namely, the encryption circuit of a block cipher called LowMC \cite{ARS+15}). They use this cipher to construct the Fish and \picnic schemes. The Fish scheme has a signature length of about 120KB, but its security is only known in the ROM, whereas the signature sizes of \picnic are about 195 KB, and it is proven to be secure in the QROM.

A follow-up work \cite{KZ20} optimizes the \picnic scheme \emph{with respect to the same one-way function} to below 50KB.

In our setting, we are not free to choose the function. We always instantiate $\picnic$ with either \sha or \shaf, depending on the context. Chase et al. manage to apply their technique to \sha to obtain signatures of size 618KB with 128 bits of security (the performance of the optimized version of \cite{KZ20} when instantiated with \sha is not analyzed). We are unaware of empirical data for $\picnic$ signature sizes when instantiated with \shaf.

A further advantage of the constructions in \cite{CDG+17,KZ20}  is that they could be instantiated with various parameters. The signature size could hence be decreased by reducing the number of bits of security. For example, in \cite{CDG+17}, it is noted that the size of the signature instantiated with \sha and admitting 80 bits of security is 385KB long, which is an improvement by a factor of almost a half.

We are hopeful that the scheme in \cite{KZ20} affords shorter signatures when instantiated with the functions we require, and further optimizations could further reduce signature sizes.


\subsection{Lifted Spending} \label{row:ls}
\emph{Lifted spending} is the same as regular spending, except the lifted signature is used instead of the pre-quantum signature.

Say that a user with access to a secret-key $\sk$ (and consequentially, to the corresponding public-key \pk\footnote{Technically, the definition of a signature scheme does not imply that $\pk$ can be computed from $\sk$ alone. However, this property holds in all signature schemes the authors are aware of, and in particular in all signature schemes used in Bitcoin. We hence use this assumption freely. See the discussion in \ifnum\masterthesis=0\cref{sssec:ec}\else\cref{ssec:prelim_ds_ec}\fi}) wants to spend a \utxo whose address is $\HH(\pk)$ without leaking \pk.

They can sign the transaction \textsf{tx} with $\sigma\gets \picnic(\HH).\sign_{\pk}(\mathsf{tx})$ \myhl{(note that the public-key of the original scheme acts as the secret-key of the lifted scheme)}. The correctness of \picnic implies that it suffices to know $\HH(\pk)$ in order to verify $\sigma$, and \cref{prop:keylift} implies that publicly posting $\sigma$ does not compromise the security of the user. Similarly, seed lifting can be used to spend a \derived \utxo, see \ifnum\masterthesis=0{\cref{ssec:seedlift}}\else{\cref{sec:lifting_seed}}\fi.

\ifnum\masterthesis=0
We remind that, due to the signature sizes, we do not propose to implement lifted spending directly, but we rather use it as a building block for a larger protocol in \cref{chap:lfc} protocol we introduce in \cref{chap:lfc}. For completeness, we survey the properties of lifted-spending.
\else
Recall that lifted signatures are prohibitively large -- see \cref{sec:lifting_size} for a discussion about their sizes, and \cref{ssec:perlim_bitcoin_pqsig} for a the effects of signature sizes on throughput. We hence to not propose to implement lifted spending directly, but we rather use it as a building block for the \emph{lifted FawkesCoin} protocol we introduce in \cref{chap:lfc}. For completeness, we survey the properties of lifted-spending.
\fi





\paragraph*{Spending \Hashed \utxos}\label{cell:ls_hashed} To spend a \hashed \utxo, the user signs it with a key-lifted signature (see \ifnum\masterthesis=0{\cref{ssec:keylift}}\else{\cref{sec:lifting_key}}\fi).

\paragraph*{Spending \Derived \utxos}\label{cell:ls_derived} To spend a \derived \utxo, the user signs it with a seed-lifted signature (see \ifnum\masterthesis=0{\cref{ssec:seedlift}}\else{\cref{sec:lifting_seed}}\fi).

\begin{remark}
    We stress that the security of lifted spending relies on the policy that key lifting could only be used to spend \hashed \utxos, and leaked \utxos must be spent using seed lifting. This follows as quantum adversaries exposed to the public-key can compute the secret-key, which they could then use to create valid key-lifted signatures (see \ifnum\masterthesis=0{\cref{ssec:keylift}}\else{\cref{sec:lifting_key}}\fi).
\end{remark}

\paragraph*{\Hashed/\Derived Confirmation Times}\label{cell:ls_hdconf} Besides using a different signature, the spending procedure is the same as in the current state, and so it has the same confirmation times.

\paragraph*{Transaction Size Increase}\label{cell:ls_tx} As we discuss in \ifnum\masterthesis=0{\cref{ssec:liftedsize}}\else{\cref{sec:lifting_size}}\fi, currently the size of a lifted signature is of the order of hundreds of kilobytes.

\paragraph*{High Spendability Threshold}\label{cell:ls_spendthres} Since the \utxo is spent directly, the only barrier to spendability is that it is valuable enough to cover the cost of a transaction fee. Since lifted signatures are several orders of magnitude larger than pre-quantum signatures (see \ifnum\masterthesis=0{\cref{ssec:liftedsize}}\else{\cref{sec:lifting_size}}\fi), the spendability threshold for using lifted signatures also grows several orders of magnitude larger.

\paragraph*{Works Without Post-Quantum \utxo} \label{cell:ls_pqutxo} The fee for spending coins is taken from the spent \utxo exactly like it is currently, so no additional coin is required.

\paragraph*{Requires a hard-fork}\label{cell:ls_softfork} Recall that in \ecdsa the public-key is recoverable from a signature, whereby it is insecure to include an \ecdsa signature in the transaction\footnote{In contrast, in some implementations of the Schnorr signature scheme, including the one used in Bitcoin, recovering the public-key from a valid signature requires inverting a hash function and is thus considered infeasible, even for quantum adversaries.}. This means that valid lifted spending transactions can not be shaped as valid \ecdsa transactions. Hence a hard-fork is required.
\section{Lifted FawkesCoin}\label{chap:lfc}\label{row:lfs}\label{sec:lfc}

The \emph{Lifted FawkesCoin} method combines FawkesCoin (see \cref{sec:fawkes}) and signature lifting (see \ifnum\masterthesis=0{\cref{ssec:zk}}\else{\cref{chap:sl}}\fi) to allow spending \hashed and \derived \utxos without using a side payment or a post-quantum \utxo, \emph{while keeping the transaction size small}.


As briefly touched upon in \cref{sec:fawkes}, a major difficulty in FawkesCoin is dealing with \emph{denial-of-service attacks}. Namely, we need to prevent adversaries from cheaply spamming the blockchain. However, we need to allow users to post commitments, and at the time of commitment, it is not known what the commitment is for and if it will ever be revealed. Since the user is only charged a fee during the reveal, this opens up attacks where fake commitments are posted, or commitments for the same \utxo are posted several times, wasting valuable blockchain space for free.

\ifnum\masterthesis=0{
In \cref{sec:fawkes}, we will explain how to overcome this issue by requiring the user to pay \emph{at commitment time} by using a post-quantum \utxo. However, it is desirable that users without access to post-quantum \utxos could quantum-cautiously spend their money.
}\fi

Bonneau and Miller~\cite{BM14} propose two approaches to the denial-of-service problem, but both fall short of solving the problem:

\begin{itemize}
    \item \textbf{Zero-Knowledge}\quad The authors suggest that the user sends a ZK proof to the miner that the \utxo is valid and pays a "useful amount" of fees. However, as long as the transaction has not been revealed, there is nothing forcing the user to finalize the transaction or preventing them from committing to the same \utxo several times.
    \item \textbf{Merkle Trees}\quad The authors suggest that the miner arranges all commitments in a Merkle tree and only posts the root, requiring users to post Merkle proofs in their reveal message. In this scenario, a spammer can not increase the size of the Merkle root, but by spamming the miner, they could increase the size of a Merkle \emph{proof}.
\end{itemize}

We propose a way to avoid denial-of-service attacks by using signature lifting. The idea is to limit the time the user has to post a commitment. If a user fails to post a commitment within this time, the miner can post the signature as a proof and claim the coin the user tried to spend as their own. To prevent users from committing to the same \utxo on several blocks, we require the user to reveal in advance the \utxo $u$ they are spending, and the miner has to post $u$ to the blockchain along the commitment. No other commitments for that \utxo will be accepted as long as the commitment does not expire. To prevent miners from abusing this power to censor \utxos, if the user failed to reveal a transaction spending $u$, the miner is \emph{required} to post a proof of ownership. If a commitment to a \utxo was posted, but neither a reveal nor a proof of ownership was posted before the \utxo expired, we assume the miner delayed the \utxo and penalize them to compensate the owner.

After the commitment has been posted, the spender has a limited time to post a reveal of the commitment to the blockchain. If they fail to do so, the miner can post the proof of ownership to claim the entire \utxo, making spam attempts costly. Note that the proof of ownership is never posted to the blockchain if both parties are honest and follow the protocol. Hence, the transaction size is about the same size as current transactions spending \emph{pre}-quantum \utxos.

However, the protocol described so far has a different problem: it allows \emph{censoring} \utxos. Note that a miner can commit a \utxo even without a proof of ownership, preventing the owner from spending it. To prevent this, we limit the time the miner has to post a proof of ownership, and require the miners to leave a \emph{deposit} proportional to the value of the \utxo, that will be awarded to the owner of the \utxo if the miner fails to post the proof. The deposit makes such delay attacks extremely costly for the miner, and reimburses the aggrieved user, effectively eliminating this drawback.

The resulting protocol for spending a \utxo $u$ is as follows:
\begin{enumerate}
    \item (Spender creates transaction) The spender prepares a transaction \textsf{tx} spending $u$ and paying fees as usual. 
    \item (Spender creates commitment) The spender prepares  $(\HH(\mathsf{tx}), \sigma, u, \alpha)$, where $\alpha$ is the amount of fees paid by \textsf{tx}, and $\sigma$ is a lifted signature on $(\HH(\mathsf{tx}),\alpha)$ which acts as a \emph{proof of ownership} of $u$. The type of the signature $\sigma$ depends on the type of \utxo being spent, as explained below.
    \item (Spender broadcasts commitment) The spender broadcasts $(\HH(\mathsf{tx}),\sigma, u, \alpha)$ to the mempool.
    \item (Miners validate commitment) If $\HH(u)$ appears in a previous commitment that has not yet expired (that is, $u$ is locked), the transaction is invalid. Moreover, if $\sigma$ is a key lifted (i.e. not seed lifted) signature and $u$ is not \hashed (or---more accurately---not \emph{leaked}, see \cref{ssec:leaked}), the transaction is also considered invalid. 
    \item (Miners include commitment) The miner posts $(\HH(\mathsf{tx}),\HH(u),\alpha)$ to the blockchain. The miner does not get a fee at this point, and the source of the fee is not yet determined, but an honest miner is guaranteed they will either get the fee or will have a chance to post $\sigma$ in order to claim all the coins in $u$.
    \item (Spender waits) After $(\HH(\mathsf{tx}))$ is posted to the blockchain, the user waits for 100 blocks.
    \item (Spender broadcasts reveal) The user must broadcast \textsf{tx} to the mempool soon enough so it will be included within the following 100 blocks. \textsf{tx} is considered invalid if it does not include a fee of $\alpha$. The fee paid for including \textsf{tx} is split equally between the miner who included \textsf{tx} and the miner who included $(\HH(\mathsf{tx}))$.
    \item (Receiver waits for confirmation) The receiver considers the transaction \textsf{tx} completed six blocks after the reveal was posted.
    \item (Miner posts proof of ownership) If the user fails to have \textsf{tx} included to the blockchain within 100 blocks, the miner must post the proof of ownership $\sigma$ to the blockchain within 100 blocks. If the miner posts the proof of ownership, they are given the entire value of $u$,
    \item (Compensation of delayed \utxos) If the miner fails to post the proof of ownership in time, they pay 3.35\% of the value of $u$ to the address of $u$.
\end{enumerate}

The disadvantage of Lifted FawkesCoin over FawkesCoin is that miners are not likely to include transactions whose value is less than the fee required to post a proof of ownership, as they risk losing money in case the commitment is not revealed. Hence, the spendability threshold (see p. \pageref{col:spend_thres}) of this method remains as high as the cost of posting a lifted signature.

Another issue that has to be addressed is obtaining commitment fees. The payment a miner earns for including a commitment is typically lower than the cost of posting a transaction, making it unspendable. Hence, creating a new \utxo paying the fee for every resolved commitment would increase the amount of dust (see \ifnum\masterthesis=0{\cref{par:dust}}\else{\cref{ssec:perlim_bitcoin_dust}}\fi). To address this, we propose ways to aggregate the payment of commitment fees in \cref{app:lfc_fees}.

\subsection{Properties of Lifted FawkesCoin}

\paragraph*{Spending \Hashed \utxos}\label{cell:lfc_hashed} To spend a \hashed \utxo, the user follows the protocol above using a key-lifted signature (see \ifnum\masterthesis=0{\cref{ssec:keylift}}\else{\cref{sec:lifting_key}}\fi). Additionally, the miner has to verify that the \utxo is not leaked (see \cref{ssec:leaked}).

\paragraph*{Spending \Derived \utxos}\label{cell:lfc_derived} To spend a \derived \utxo, the user follows the protocol above using a seed-lifted signature (see \ifnum\masterthesis=0{\cref{ssec:seedlift}}\else{\cref{sec:lifting_seed}}\fi), and includes $\msk,P$ in the payload of $\tx$. It must hold that $\pkec(\der(\msk,P)) = \pk$ for the revealing transaction to be considered valid.

\paragraph*{\Hashed/\Derived Confirmation Times}\label{cell:lfc_conf}The confirmation times are the same as in regular FawkesCoin.

\paragraph*{Transaction Size Increase}\label{cell:lfc_tx}The transaction size is the same size of a regular \emph{pre}-quantum transaction plus 64 bytes for commitment and the hash of the \utxo.

\paragraph*{Spendability Threshold}\label{cell:lfc_spendthres}While transaction sizes are small, they should be valuable enough to cover the cost of posting the proof of ownership. Including a less valuable \utxos imposes a risk on the miner, so it will probably not be included. We further discuss the size of lifted signatures in \ifnum\masterthesis=0{\cref{ssec:liftedsize}}\else{\cref{sec:lifting_size}}\fi.

\paragraph*{Works Without Post-Quantum \utxo} \label{cell:lfc_pqutxo} The protocol makes it secure to pay the fees from the original \utxo, so no additional coin is required.

\paragraph*{Limited Delay Attacks}\label{cell:lfc_delay} A miner can delay spending an arbitrary \utxo by posting its hash alongside a fake commitment. However, once the transaction is expired, they will compensate the user.

\paragraph*{Requires a Hard-Fork}\label{cell:lfc_softfork} In case the user fails to reveal the transaction in time, the miner is awarded the coins in the \utxo even though they do not have access to a valid signature. This requires nodes to consider a transaction valid even though it does not contain a signature for the \utxo it is trying to spend. As we explain in \ifnum\masterthesis=0{\cref{par:softfork}}\else{\cref{ssec:perlim_bitcoin_dust}}\fi, this requires a hard-fork.

\paragraph*{Offline Users Not Risked}\label{cell:lfc_offline} \Derived \utxo are only spendable using the derivation seed, so \naked and \lost \utxos can not be spent using this method.

\subsection{Obtaining Commitment Fees}\label{app:lfc_fees}

When processing a block, it is impossible to know how much fee the miner accrued for including Lifted FawkesCoin commitment. The most straightforward solution is, once the transaction is revealed, to create a \utxo spending the fee to an address specified in the block which contains the commitment. The problem with this solution is that it creates many \utxos whose value is comparable to the cost of spending them, thereby increasing the amount of dust (see \ifnum\masterthesis=0{\cref{par:softfork}}\else{\cref{ssec:perlim_bitcoin_dust}}\fi).

We propose to overcome this by delaying the coinbase processing. The Lifted FawkesCoin protocol guarantees that the exact amount of fees accrued by the miner is resolved after the expiry period is over (which amounts to 300 blocks using the 100 blocks waiting time we recommend in \cref{ssec:wait_time}). Hence, processing the coinbase transaction of the block could be delayed for this length of time. This could be implemented by having the coinbase transaction of each block pay to the block that appears three waiting periods before it, if that block contains Lifted FawkesCoin commitments. The \href{https://kaspa.org/}{Kaspa cryptocurrency} resolves a similar situation (where the final transaction fees are not known at the time the blocks are created) by implementing such an approach (see \href{https://github.com/kaspanet/docs/blob/main/Reference/Block\%20Rewards\%2C\%20Transaction\%20Fees\%2C\%20and\%20Coinbase\%20Transactions.md}{relevant documentation}). Note that coinbase cooldown \emph{starts} to take effect only after the expiry period is \emph{finished}.

A disadvantage of this approach is that the elongated wait times hold for the entire coinbase transaction, including the block reward and the fees for transactions not made using Lifted FawkesCoin. This could be resolved by allowing the coinbase transaction to include the block reward and all fees except the Lifted FawkesCoin fees, and having each block reward Lifted FawkesCoin fees to the block appearing three wait periods before it (if there are any).

\subsection{Increasing Lifted FawkesCoin's Throughput}\label{app:throughput}

In the current design, a miner only has a limited time to post the \myhl{proof of ownership}. This has the effect that the number of commitments risk-averse miners will agree to include is, at most, the number of fraud-proofs the network can support. Otherwise, \myhl{the miners} take the risk that they would have to post more fraud-proofs than possible, eventually causing them to pay delay fines albeit being honest.

Given that fraud proofs are much larger than commitments, this greatly decreases the commitment throughput. For example, if a block can only contain one fraud-proof, then no more than 300 commitments will be included during a $500$ blocks long FawkesCoin epoch.

Our approach to increasing the throughput is by extending the expiry time of Lifted FawkesCoin commitments in scenarios where many fraud-proofs are posted. If at the end of the Lifted FawkesCoin epoch (see \cref{ssec:combining}) sufficiently many fraud-proofs were posted, a new Lifted FawkesCoin epoch starts immediately, during which fraud-proofs for commitments from previous epochs are allowed. We refer to this new epoch as an \emph{extension}.

In more detail, we propose the following:
\begin{itemize}
    \item Currently, if the spender failed to post a reveal, the miner has a period of $100$ blocks to post a fraud-proof. We modify this rule such that the miner is allowed to post the fraud-proof at any time up to the \emph{end of the epoch}. However, we still do not allow commitments to be posted within the last $300$ blocks of the epoch, to allow miners a period of \emph{at least} $100$ blocks to post a proof.
    \item Say that $100$ blocks can contain up to $k$ fraud-proofs. At the end of the epoch:
    \begin{itemize}
        \item If more than $k/2$ fraud-proofs were posted in the last $100$ blocks, then:
        \begin{itemize}
            \item Do not pay the miner's deposits to the users,
            \item Start a new Lifted FawkesCoin epoch,
            \item During this epoch, allow miners to post fraud-proofs for any unrevealed commitment that has not yet expired (even if the commitment was included in a previous epoch).
        \end{itemize}
        \item Else: 
        \begin{itemize}
            \item For each non-expired unsettled commitment, pay the fine left for that commitment to the address of the \utxo, and consider the commitment expired.
            \item Do not initiate a new Lifted FawkesCoin epoch, but rather allow the epoch rotation to continue.
        \end{itemize} 
    \end{itemize}
\end{itemize}

More generally, one can set the threshold to trigger an extension at $kp$ fraud-proofs for any $0<p<1$. However, we now discuss possible attacks on this mechanism and conclude that $p=1/2$ is a natural choice (or more precisely, that $p$ should be \emph{at most} $1/2$, but should not be set too low either).

There are two ways to abuse the extension mechanism: either by \emph{forcing} or by \emph{denying} an extension.

In order to force an extension, an adversary must post $kp$ fraud proofs. The consequence of such an attack is threefold: 
\begin{itemize}
    \item It prevents users from using non-lifted FawkesCoin and in particular prevents spending \lost \utxos.
    \item Inability to use non-lifted FawkesCoin also makes it impossible to spend pre-quantum \utxos whose value is below the cost of posting a fraud-proof.
    \item It extends the length of ongoing delay attacks without increasing the delay attack fine eventually paid to the owners of the attacked \utxo.
\end{itemize}
Note that during the attack, users can still use Lifted FawkesCoin and spend post-quantum \utxos.

The adverse effects of this attack are disruptive but do not last once the attack is over, and maintaining the attack is very expansive: note that the total size of the maximal amount of fraud-proofs a block can contain is at least half a block, hence, the cost of posting $kp$ fraud-proofs is at least as high as the fee for consuming a space equivalent to $50p$ blocks (moreover, maintaining such an attack for an extended period wastes a lot of space, plausibly increasing the cost of posting fraud-proofs). 
In particular, if $p=1/2$, the cost of forcing \emph{each} extension is at least as high as the cost of block space equivalent to $25$ blocks. Setting $p$ too low might make such attacks affordable. 
We point out that the cost of an extension forcing attack could also be increased by increasing the length of the lifted FawkesCoin epoch, whereby increasing $k$.

An extension \emph{denying} attack is more dangerous since it may be \emph{profitable}. On the other hand, it can only be carried by a miner with a fraction of $q>1-p$ of the total hash rate. Such an adversarial miner can deny an extension by simply refusing to include fraud-proofs in her blocks. The adversarial miner can post many commitments to the mempool, so that they would be included by honest miners, and then deny the extension. By the end of the epoch, at most $(1-q)k$ of the corresponding fraud proofs will have been posted, and the adversarial miner would gain the delay attack fines for the remaining commitments.

Hence, we strongly recommend setting $p\le 1/2$. This assures that such an attack can be only carried by $>50\%$ attackers (recall that such powerful attackers can already severely damage the network in many other ways\ifnum\masterthesis=1{, as explained in \cref{ssec:prelim_bitcoin_revert})}\fi.

\subsection{Delay Attack Fines}\label{ssec:fines}

The fine for delay attacks should reflect the damage done to the user. However, it should not be chosen too large, as it would make it hard for miners to include commitments on highly valuable \utxos. We propose setting the fine to be the equivalent of an \emph{annual} 100\% interest, paid over the period the transaction was delayed.

\ifnum\masterthesis=0{Note that if rotation is employed with the periods we proposed in \cref{intro:proposal}, then a delay attack actually prevents the user from spending their \utxo in the current Lifted FawkesCoin epoch, forcing them to wait as much as 2,500 blocks. Hence, we propose that the interest should be calculated as an annual interest of 100\% accrued over a period of 25,000 minutes, which is 3.35\% of the value of the \utxo.}\else{
We propose that the interest should be calculated as an annual interest of 100\% accrued over the maximal period of time a \utxo could be locked for. That period of time depends on the implementation details and constants. In the concrete implementation we propose in \cref{sec:procrast_full}, this time is about 2,500 blocks. So, if our specification is to be followed, we propose that the fee should be an annual interest of 100\% accrued over a period of 25,000 minutes, which is 3.35\% of the value of the \utxo.
}\fi

The miner should be able to cover delay fines for all \utxos whose hash is used in a Lifted FawkesCoin commitment. Let the \emph{guaranteed coinbase value} be the value of the block reward and all transactions which are not Lifted FawkesCoin commitments (this is a lower bound on the block reward in case the miner is honest. Obviously, the final block reward could be lower if the miner performs delay attacks). If the guaranteed coinbase value is lower than the sum of all required deposits, then the miner needs to include an additional transaction covering the difference.

\section{Restrictive and Permissive FawkesCoin}\label{sec:fawkes}

FawkesCoin \cite{BM14} is a blockchain protocol that avoids public signature schemes altogether by employing the Fawkes signatures of \cite{ABC+98}. The core idea is that in order to spend a \utxo the user commits to a certain transaction and reveals the transaction once some predetermined period of time has passed (we further discuss the considerations for choosing the waiting time in \cref{ssec:wait_time}, where we propose a waiting time of 100 blocks). The security of the protocol follows from the observation that if the waiting time is chosen long enough, then after revealing the transaction, an attacker has a negligible chance to complete a commit-wait-reveal cycle before the honest user manages to include the revealed transaction in the blockchain.

In the original FawkesCoin design, the secret-key is a random string $r$, and the public address for spending to this key is $\HH(r)$, where $\HH$ is some agreed upon collision-resistant hash function. In order to spend a \utxo with address $\HH(r)$ to another address $\HH(s)$, the user posts $\HH(r,\HH(s))$ as a commitment, and $r$ as a reveal. Given $r$ and $\HH(s)$, anyone can verify that $\HH(r)$ and $\HH(r,\HH(s))$ evaluate correctly.

Bonneau and Miller note that their solution could be integrated into Bitcoin and that a \utxo could be spent by using a hash of a valid transaction as a commitment and the transaction itself as a reveal. This allows cautious spending of pre-quantum \hashed \utxos. They point out that this approach could mitigate "a catastrophic algorithmic break of discrete log on the curve P-256 or rapid advances in quantum computing."

The main obstacle to adopting FawkesCoin is in incentivizing miners to include commitment messages in the blocks in the first place. However, this issue could be completely circumvented if the user already has access to post-quantum \utxos that they could use to pay the fees. In this section, we assume that it is the case. 

While it is desirable that users could use the committed transaction to pay the miner fee, achieving this feature without allowing denial-of-service attacks proves far from trivial. In \cref{sec:lfc} we introduce \emph{Lifted FakwesCoin}, a variation of FawkesCoin which allows paying the transaction fee out of the spent \utxo by using lifted signatures (see \ifnum\masterthesis=0{\cref{ssec:zk}}\else{\cref{chap:sl}}\fi).

However, non-lifted FawkesCoin has its own benefits. Mainly, it has a lower spendability threshold so it is more suitable for small transactions. Also, using lifted signatures allows us to extend the applicability of non-lifted FawkesCoin to extend the set of spendable \utxos to also include \naked and \lost \utxos, we call the extended protocol \emph{permissive} FawkesCoin. We stress that while both permissive FawkesCoin and lifted FawkesCoin are constructed by combining FawkesCoin and lifted signatures, they are inherently different protocol. In particular, in permissive FawkesCoin neither the user nor the miner have to post a lifted signature at any point. Lifted signatures are only used in cases where there is an attempt to steal a \utxo.


\subsection{Overview of FawkesCoin modes}\label{ssec:cfc}

We extend the FawkesCoin protocol by adding two new modes of operation, namely \emph{unrestrictive} and \emph{permissive} modes. These modes allow spending \utxos that are not cautiously spendable using the original (restrictive) design of FawkesCoin. The two new modes of operation work similarly to optimistic rollups~\cite{Eth22}, so we follow the terminology used therein.

The original application of FawkesCoin to quantum cautious spending presented in \cite{SIZ+18} is equivalent to restrictive FakwesCoin further restricted to spending only \hashed \utxos.

\begin{itemize}
    \item \textbf{Restrictive\los{ (\cref{ssec:cfc})}{}}\label{row:cfc}\quad This mode allows users to spend \hashed \utxos using the public-key, and \derived \utxos using the derivation seed. The advantage of this mode of operation is that it minimizes quantum loot. The disadvantage is that it can only support spending \hashed and \derived \utxos (so implementing restrictive FawkesCoin as the \emph{only} way to spend pre-quantum \utxos will make all pre-quantum \utxos which are neither \hashed nor \derived unspendable indefinitely).
    
    \item \textbf{Unrestrictive\los{ (\cref{ssec:nfc})}}\label{row:nfc}\quad This mode further allows users to spend \naked \utxos. 
    However, to do so they are required to provide a \emph{deposit} as large as the value of the \utxos, and then the transaction goes into a long \emph{challenge period} \los{(in \cref{ssec:dispute} we discuss the length of the waiting period and propose to set it to one year)}{} during which the rightful owner can claim the coins in the \utxo \emph{and} the deposit by posting a \emph{fraud proof} showing that they hold the derivation seed. 
    The purpose of the deposit is to make it risky to try stealing a \utxo not known to be \stealable.
    
    The disadvantage of unrestrictive mode is that owners of \derived \utxos that are not \hashed have to occasionally (e.g. at least once a year) scan the blockchain for attempts to spend their money.
    
    \item \textbf{Permissive\los{ (\cref{ssec:pfc})}}\label{row:pfc}\quad This mode further allows users to spend \lost \utxos.
    To achieve this, we allow users to claim \naked \utxos without signing the transaction. This allows recovering of \lost \utxos. Like in unrestrictive mode, we require a deposit for spending a \naked \utxo, making it risky to try to claim a \naked \utxo without knowing that it is \lost.
    Permissive mode has the advantage that \lost \utxos are no longer quantum loot.
    This method allows spending \utxos without presenting a signature, whereby implementing it requires a hard-fork.
\end{itemize}

Using restrictive mode minimizes the quantum loot available to a quantum adversary. 
One might argue that even though unrestrictive and permissive modes allow more honest users to spend their coins, they have a negative side-effect, namely, unrestrictive and permissive modes increase the available quantum loot over restrictive mode. However, this is not quite the case, since holders of \derived \utxos can falsely declare their \utxos lost with the purpose of baiting attackers to place a deposit they could claim. We discuss this further in \cref{ssec:fc_loot}. 
Note that permissive mode could be modified to allow spending \hashed \utxos whose keys were lost. However, we advise against this approach as it would require \emph{all} holders of pre-quantum \utxos, including \hashed \utxos, to be actively cautious against attempts on their money. Also note that permissive mode could be implemented regardless of quantum adversaries as means for recovering lost funds, though this approach also has the drawback of requiring users to be online to maintain the safety of their money.

\subsection{Restrictive FawkesCoin}\label{row:fc}

Restrictive FawkesCoin allows spending a \hashed \utxo by creating a transaction spending it, posting the hash of the transaction as a commitment, and posting the transaction itself to reveal it. It also allows spending \emph{any} \utxo by committing to a derivation seed instead.

\paragraph*{Spending a \Hashed \utxo} \label{cell:cfc_hashed}\label{cell:nfc_hashed}\label{cell:pfc_hashed} The protocol for spending a \Hashed \utxo $u$ is as follows: 
\begin{enumerate}
    \item (Spender creates transaction) The spender creates a transaction \textsf{tx} spending $u$, which includes a standard fee.
    \item (Spender creates commitment) The spender creates a transaction having $\HH(\textsf{tx})$ as its payload, which is signed and pays fees using a post-quantum $\utxo$. We refer to this committing transaction as $\textsf{ctx}$.
    \item (Spender posts commitment) The spender posts \textsf{ctx} to the mempool.
    \item (Miners include commitment) The miners include \textsf{ctx} in their blocks, the fee for \textsf{ctx} goes to the miner who included it first, as usual.
    \item (Spender waits) Once \textsf{ctx} is included in the blockchain, the spender waits for 100 blocks to be mined above it (see \cref{ssec:wait_time}).
    \item (Spender posts reveal) The spender posts \textsf{tx} to the mempool.
    \item (Miners validate) The miners verify that:
    \begin{itemize}
        \item $\mathsf{\HH(tx)}$ appears in the payload of a committing transaction \textsf{ctx} at least 100 blocks old, and
        \item the \utxo was \hashed when the commitment was posted: the public-key used to sign \textsf{tx} does not appear in the blockchain before \textsf{ctx}.
    \end{itemize}
    If \textsf{tx} does not satisfy both conditions, it is considered invalid.
    \item (Miners include reveal) The miners include \textsf{tx} in their block. The fee for \textsf{tx} goes to the miner who included it first, subject to the two conditions above.
    \item (Receiver waits for confirmation) The receiver considers the transaction completed once \textsf{tx} accumulated six confirmations.
\end{enumerate}

\paragraph*{Spending a \Derived \utxo}\label{cell:cfc_derived}\label{cell:nfc_derived}\label{cell:pfc_derived} To spend a derived \utxo whose address is $\pk$ (or a hash thereof), we require the user to commit and reveal a parent extended secret-key $\xsk_{par}$ and a derivation path $P$ such that $\pkec(\der(\xsk_{par},P)) = \pk$ (see \ifnum\masterthesis=0\cref{sssec:ec}\else\cref{ssec:prelim_ds_ec}\fi and \ifnum\masterthesis=0{\cref{ssec:hdwallet}}\else{\cref{ssec:prelim_other_hd}}\fi). We stress that once $\xsk_{par}$ is revealed, anyone could use FawkesCoin to spend any \utxo whose address has been derived from $\xsk_{par}$. Hence, to maintain the safety of their funds, the user must commit to \emph{all} \utxos whose addresses were derived from $\xsk_{par}$ before they start revealing them.

The protocol for spending a \derived \utxo $u$ whose address corresponds to a public-key $\pk$ is as follows (the differences with the protocol for spending a \hashed \utxo are underlined):
\begin{enumerate}
    \item (Spender creates transaction) The spender creates a transaction \textsf{tx} spending $u$, which \ul{contains in its payload $(\xsk_{par},P)$ such that $\pkec(\der(\xsk_{par},P)) = \pk$ } and includes a standard fee.
    \item (Spender creates commitment) The spender creates a transaction having $\HH(\textsf{tx})$ as its payload, which is signed and pays fees using a post-quantum $\utxo$. We refer to this committing transaction as $\textsf{ctx}$.
    \item (Spender posts commitment) The spender posts \textsf{ctx} to the mempool.
    \item (Miners include commitment) The miners include \textsf{ctx} in their blocks, the fee for \textsf{ctx} goes to the miner who included it first, as usual.
    \item (Spender waits) Once \textsf{ctx} is included in the blockchain, the spender waits for 100 blocks to be mined above it (see \cref{ssec:wait_time}).
    \item (Spender posts reveal) The spender posts \textsf{tx} to the mempool.
    \item (Miners validate) The miners verify that:
    \begin{itemize}
        \item $\mathsf{\HH(tx)}$ appears in the payload of a committing transaction \textsf{ctx} at least 100 blocks old, and
        \item \ul{$\pkec(\der(\xsk_{par},P)) = \pk$} 
    \end{itemize}
    If \textsf{tx} does not satisfy both conditions, it is considered invalid.
    \item (Miners include reveal) The miners include \textsf{tx} in their block. The fee for \textsf{tx} goes to the miner who included it first, subject to the two conditions above.
    \item (Receiver waits for confirmation) The receiver considers the transaction completed once \textsf{tx} accumulated six confirmations.
\end{enumerate}

\begin{remark}
    Note that when spending a \derived \utxo, the pre-quantum signature is not actually required for validation, as suffices to verify $\pkec((\xsk_{par})_{(i,p)}) = \pk$. We include the signature in \textsf{tx} so that restrictive FawkesCoin could be implemented as a soft-fork (see \ifnum\masterthesis=0{\cref{par:softfork}}\else{\cref{ssec:perlim_bitcoin_forks}}\fi). If FawkesCoin is implemented in a hard-fork, the signature could be removed to conserve space.
\end{remark}

\paragraph*{\Hashed/\Derived Confirmation Times}\label{cell:cfc_hdconf}\label{cell:nfc_hdconf}\label{cell:pfc_hdconf} The confirmation time is the length of a single commit-wait-reveal cycle (which we propose setting to 100 blocks in \cref{ssec:wait_time}), followed by waiting the current number of confirmation blocks once the reveal message is included (e.g., six blocks in Bitcoin).

\paragraph*{Transaction Size Increase}\label{cell:nfc_tx}\label{cell:cfc_tx}\label{cell:pfc_tx} The component of the transaction dominating its size is the signature of the post-quantum \utxo, which is about an order of magnitude larger than transaction spending a pre-quantum \utxo\ifnum\acmtops=0 (see \cref{ssec:nist-qpds})\fi. However, a single post-quantum transaction can be used to spend several pre-quantum \utxos (either by committing to several transaction, or by committing to a transaction spending several \utxos), making the size increase additive.

\paragraph*{Spendability Threshold}\label{cell:nfc_spendthres}\label{cell:cfc_spendthres}\label{cell:pfc_spendthres} Since the fee is paid using the post-quantum \utxo, the only obstruction to spendability is if the \utxo is less valuable then the fees required to post it. Disregarding the post-quantum \utxo (see p. \pageref{col:spend_thres}), the size of a FawkesCoin transaction is slightly larger than spending it usually (exactly by how much depends on whether the transaction spends a \hashed or \derived \utxo, and on whether we are in a soft- or hard-fork), whereby the spendability threshold stays the same up to a factor close to $1$.

\paragraph*{Requires a Post-Quantum \utxo} \label{cell:fc_pqutxo} A post-quantum \utxo is required in order to pay the transaction fee.

\paragraph*{Can be Implemented in Soft-Fork}\label{cell:nfc_softfork}\label{cell:cfc_softfork} Commitments are ordinary post-quantum transactions whereas reveals are ordinary pre-quantum transactions, whereby they would also be accepted by outdated nodes. However, a hard-fork implementation could conserve block space as the pre-quantum signature can be removed from \derived \utxos.

\subsection{Unrestrictive FawkesCoin}\label{ssec:nfc}

\emph{Unrestrictive} mode extends the functionality of restrictive mode by providing a way to spend a \naked \utxo. \Naked \utxos are spent like \hashed \utxos, except the user must leave a \emph{deposit} as valuable as the \utxo they are trying to spend. After the transaction is revealed, it goes into a long \emph{challenge period} during which any user can post a \emph{fraud proof} by spending that same transaction in FawkesCoin \emph{using the derivation seed}. If a fraud proof is posted, the revealed transaction is considered invalid, and the deposit goes to the address of the user who posted the fraud proof.

\paragraph*{Spending \Naked \utxos}\label{cell:nfc_naked}\label{cell:pfc_naked}The protocol for spending a \naked \utxo $u$ is as follows: 
\begin{enumerate}
    \item (Spender creates transaction) The spender creates a transaction \textsf{tx} spending $u$, which also spends a post-quantum \utxo $d$ called the \emph{deposit}. The value of $d$ must be at least the value of $u$ plus the fees paid for including \textsf{tx}.
    \item (Spender creates commitment) The spender creates a transaction having $\HH(\textsf{tx})$ as its payload, which pays fees using a post-quantum $\utxo$. We refer to this committing transaction as $\textsf{ctx}$.
    \item (Spender posts commitment) The spender posts \textsf{ctx} to the mempool.
    \item (Miners include commitment) The miners include \textsf{ctx} in their blocks, the fee for \textsf{ctx} goes to the miner who included it first, as usual.
    \item (Spender waits) Once \textsf{ctx} is included in the blockchain, the spender waits for 100 blocks to be mined above it (see \cref{ssec:wait_time}).
    \item (Spender posts reveal) The spender posts \textsf{tx} to the mempool.
    \item (Miners validate) The miners verify that $\mathsf{\HH(tx)}$ appears in the payload of a committing transaction \textsf{ctx} at least 100 blocks old.
    If \textsf{tx} does not satisfy both conditions, it is considered invalid.
    \item (Miners include reveal) The miners include \textsf{tx} in their block. The fee for \textsf{tx} goes to the miner who included it first, subject to the condition above.
    \item (Challenge period start) The transaction \textsf{tx} enters a \emph{challenge period} of one year.
    \item (Owner can post fraud proof) During that period, any user holding the derivation seed for the address of $u$ can post a \emph{fraud proof}: a transaction \textsf{fp} spent using the FawkesCoin protocol for \derived \utxos spending the \utxo $u$. If a fraud proof is posted, the transaction \textsf{tx} is considered invalid. The miner is paid the fee they were supposed to for posting \textsf{tx} from the deposit, and the rest of the deposit goes to same address \textsf{fp} is spent to.
    \item (Receiver waits for confirmation) If no fraud proof was posted, the receiver consider the transaction \textsf{tx} completed six blocks after the challenge period is over.
\end{enumerate}

\begin{remark}
    For simplicity, we presented the method in a manner that requires a hard-fork: if a fraud proof is posted, then the deposit needs to be spent to a different address than the one specified in its output. This could be rectified by instead requiring that the deposit output is spent to an anyone-can-spend address, as explained in \ifnum\masterthesis=0{\cref{par:softfork}}\else{\cref{ssec:perlim_bitcoin_forks}}\fi.
\end{remark}

\paragraph*{Naked Confirmation Times}\label{cell:nfc_nconf} \label{cell:pfc_nconf} The challenge period is very large to allow honest users ample time to notice if anyone attempted to steal their transactions and respond accordingly (see \cref{ssec:dispute}).

\paragraph*{Offline Users Risked}\label{cell:nfc_offline}\label{cell:pfc_offline} Users holding naked derived \utxos have to monitor the network for attempts to spend their money. In particular, if an adversary knows of a user that would not be online for a duration longer than a dispute period (say, if he's denied modern society and became a recluse, joined an Amish community, or is just forever trapped beyond the event horizon of a black hole), they can safely steal their naked derived \utxos.

\subsection{Permissive FawkesCoin}\label{ssec:pfc}

\emph{Permissive mode} extends the functionality of unrestrictive mode by providing a way to spend \emph{any} leaked \utxo without providing any proof of ownership. That is, we allow anyone to spend any leaked \utxo. However, to spend a leaked \utxo without providing a valid signature or a derivation seed, the spender must provide a deposit and wait for a lengthy challenge period, exactly like \naked \utxos are spent in unrestrictive mode. The deposit acts to deter adversaries from attempting to steal \utxos, as they can not know whether the owner of the \utxo they are trying to steal can produce a fraud proof.

\paragraph*{Spending Lost \utxo}\label{cell:pfc_lost} The protocol for spending a \lost \utxo $u$ is as follows: 
\begin{enumerate}
    \item (Spender creates transaction) The spender creates a transaction \textsf{tx} spending $u$, which also spends a post-quantum \utxo $d$ called the \emph{deposit}. \emph{The \utxo need not include a signature on $u$}. The value of $d$ must be at least the value of $u$ plus the fees paid for including \textsf{tx}.
    \item (Spender creates commitment) The spender creates a transaction having $\HH(\textsf{tx})$ as its payload, which  pays fees using a post-quantum $\utxo$. We refer to this committing transaction as $\textsf{ctx}$.
    \item (Spender posts commitment) The spender posts \textsf{ctx} to the mempool.
    \item (Miners include commitment) The miners include \textsf{ctx} in their blocks, the fee for \textsf{ctx} goes to the miner who included it first, as usual.
    \item (Spender waits) Once \textsf{ctx} is included in the blockchain, the spender waits for 100 blocks to be mined above it (see \cref{ssec:wait_time}).
    \item (Spender posts reveal) The spender posts \textsf{tx} to the mempool.
    \item (Miners validate) The miners verify that $\mathsf{\HH(tx)}$ appears in the payload of a committing transaction \textsf{ctx} at least 100 blocks old.
    If \textsf{tx} does not satisfy both conditions, it is considered invalid.
    \item (Miners include reveal) The miners include \textsf{tx} in their block. The fee for \textsf{tx} goes to the miner who included it first, subject to the condition.
    \item (Challenge period start) The transaction \textsf{tx} enters a \emph{challenge period} of one year.
    \item (Owner can post fraud proof) During that period, any user holding the derivation seed for the address of $u$ can post a \emph{fraud proof}: a transaction \textsf{fp} spent using the FawkesCoin protocol for \derived \utxos spending the \utxo $u$. If a fraud proof is posted, the transaction \textsf{tx} is considered invalid. The miner is paid the fee they were supposed to for posting \textsf{tx} from $d$, and the remaining coin in $d$ goes to the same address \textsf{fp} is spent to.
    \item (Receiver waits for confirmation) If no fraud poof was posted, the receiver consider the transaction \textsf{tx} completed six blocks after the challenge period is over.
\end{enumerate}

\begin{remark}\label{rem:deposit}
    In the specification above we chose the value of the deposit $d$ to be as high as the value of the spent \utxo $u$ (plus the fees for including $\textsf{tx}$). This implies that an adversary has a negative expected profit from such an attack as long as they can't make an educated guess that a particular \utxo is \lost with a probability higher than $1/2$. It might be the case that such educated guesses are feasible, and if so the value of the deposit should be increased appropriately. Setting the value of the deposit to $\frac{p}{1-p}$ times the value of the spent output implies that attempting to steal a \utxo using permissive FawkesCoin has negative expected profit, unless knows that \utxo is \lost with probability at least $p$. Increasing $p$ makes such attacks less feasible at the cost of making spending \lost \utxos less affordable. 
\end{remark}

\paragraph*{\Lost Confirmation Times}\label{cell:pfc_lconf} Same as \naked confirmation times.

\paragraph*{Requires hard-fork}\label{cell:pfc_softfork} This method allows spending \utxos without producing a signature, which can not be implemented in a soft-fork (see \ifnum\masterthesis=0{\cref{par:softfork}}\else{\cref{ssec:perlim_bitcoin_forks}}\fi).

\subsection{Salvaged Coins and Loot} \label{ssec:fc_loot}


In the current state, all but post-quantum and \hashed \utxos can be stolen by a quantum adversary (and \hashed \utxos cannot be spent cautiously).
After FawkesCoin kicks in, all \hashed and \derived \utxos can be spent cautiously. The fate of a \naked, \lost, or \stealable \utxo depends on the chosen mode of FawkesCoin, where each such \utxo can be in one of the following statuses:

\begin{itemize}
    \item Burnt: no longer spendable by anyone.
    \item Spendable: spendable by the owner. A (quantum) adversary can successfully steal the \utxo, but would not do that due to the risk of losing their deposit. Hence, spendable \utxos are less secure than cautiously spendable \utxos.
    \item Loot: the \utxo will be stolen.
    \item Quantum loot: the \utxo will be stolen by a quantum adversary.
\end{itemize}

\cref{table:salvage} summarizes the status \naked \utxos according to the FawkesCoin variant.

\setlength{\tabcolsep}{6pt}
\begin{table}[!htb]
\begin{adjustbox}{max width=1.1\textwidth,center}

\begin{tabular}{l|ccc}
 & \textbf{\Naked} & \textbf{\Lost} & \textbf{\Stealable} \\ \hline
\textbf{Restrictive}   & Burnt     & Burnt    & Burnt \\
\textbf{Unrestrictive} & Spendable      & Unspendable     & Quantum Loot \\
\textbf{Permissive}    & Spendable  & Spendable & Loot
\end{tabular}

\end{adjustbox}

\caption{
    The status of different types of \naked \utxos according to the variant of FawkesCoin used (see \cref{ssec:fc_loot}).
}
\label{table:salvage}
\end{table}

One of the purposes of the canary mechanism is to decrease the set of spendable \utxos, and the purpose of the deposit and fraud-proof mechanism is to create \emph{risk} for adversaries attempting to steal them:

\begin{itemize}
    \item As long as the protocol did not kick in, an owner of a \stealable \utxo who knows the secret-key can migrate their funds to a \utxo that is not \stealable. They could be incentivized to do so by providing sufficient heads-up, e.g. in the form of a killed canary.
    \item \Lost \utxos are protected by the observation that an adversary cannot tell them apart from any other \naked \utxo, so attempting to steal them incurs a risk of losing their deposit.
\end{itemize}

To conclude, the different modes of Fawkes coin provide the following tradeoffs:
\begin{itemize}
    \item Unrestrictive FawkesCoin allows spending \naked \utxos, but \stealable \utxos will be stolen by quantum adversaries.
    \item Permissive FawkesCoin further allows spending \lost \utxos, but implementing the permissive mode requires a hard-fork whereas the unrestrictive mode only requires a soft-fork.
\end{itemize}


We note that while \stealable \utxos become loot in unrestrictive and permissive mode, this loot can only be safely exploitable if the adversary has great confidence that the transaction is indeed stealable. This is augmented by the fact that \utxos can be falsely announced \stealable to bait the adversary to lose their deposit.

\subsection{The Wait Time}\label{ssec:wait_time}

In the specifications above, we used a default value of 100 blocks for the length of a commit-wait-reveal cycle. Longer waiting times make the economy more secure but less usable.

\ifnum\acmtops=0
The main consideration when choosing the wait time is \emph{forking security}: the wait time should be longer than any plausible reorganization of the blockchain. Recall that the discussion in \ifnum\masterthesis=0{p. \pageref{par:cooldown} }\else{\cref{ssec:perlim_bitcoin_dust} }\fi points out that coinbase transactions require a considerably longer confirmation than a standard transaction. The justification for that is that if a reorganization reverts a coinbase transaction, then all transactions which spend money minted by that coinbase transaction become invalid too. By setting the cooldown time of a coinbase transaction high enough to guarantee that a coinbase transaction is \emph{never} reverted, it is guaranteed that no reorganization can make a valid transaction invalid (as long as the owners of the \utxos used as input therein do not attempt double spending them).

Our considerations are more similar to coinbase transactions. Once a reveal is posted, an attacker could use it to post competing commit and reveal messages. If the reorg is deep enough to revert the honest \emph{commit} message, then it is possible that the adversarial commit and reveal will be included itself, making the honest transaction invalid.
    
Thus, taking a cue from Bitcoin, we propose a 100 blocks waiting time.
\else
A similar consideration appears in Bitcoin when considering the fact that the reversal of coinbase transactions is more detrimental than the reversal of standard transaction (see the full version \cite{SW23} for a discussion). As a precaution, the Bitcoin protocol imposes a 100 blocks cooldown on coinbase transaction. Taking a cue from Bitcoin, we propose a wait time of 100 blocks.
\fi

One might argue that 100 waiting blocks is too long as it slows down the spending time too much. It is also arguable that this concern becomes more pressing as it applies to many users (whereas coinbase cooldown only applies to miners). We have several responses to this objection:
    
\begin{itemize}
        \item It is better to err on the side of caution. The consequence of choosing the wait time too short is that spending pre-quantum \utxos becomes insecure, and quantum-cautious spending becomes \emph{impossible}. This is arguably more detrimental to the economy than a predictable slowdown.
        \item The effects of quantum mining on fork rates and reorganization depths are still not yet understood. There is evidence that quantum mining increases the orphan rate \cite{LRS19,Sat20}, and it could be the case that there are more consequences we are not yet aware of.
        \item Even if no forks occur, there is still the risk that a quantum attacker will somehow manage to postpone the inclusion of the reveal message in the blockchain long enough for her to post a competing commit \emph{and} reveal. Such scenarios could also occur e.g. if congestion causes transactions to stall in the mempool for extended periods of time.
        \item Users who prefer short waiting times over small transaction sizes (that is, prefer paying faster over paying less) could use post-quantum signatures at confirmation speeds similar to today.
        \item It is possible to allow users to set their own waiting time when creating the \utxo. It is important to impose a \emph{default} waiting time to prevent front-running attacks. However, there is no harm in allowing users to set a lower waiting time if they choose (conversely, users who feel that the 100 blocks waiting time is insufficient could specify higher waiting times).
\end{itemize}

\subsection{Length of the Challenge Period}\label{ssec:dispute}

We propose a challenge period of \emph{one year}. We argue that the challenge period should be very long for two reasons:
\begin{itemize}
    \item We want to prevent a situation where a \utxo was successfully stolen because the user did not notice an attempted steal in time or had no time to gather the resources required to post a proof of fraud. A year-long period would give users ample time to notice the attempt and make preparations.
    \item Any \utxo that could be spent in restrictive mode could still be spent in the other modes without any challenge period. Even if an adversary is trying to steal the \utxo, proving fraud is just done by regularly spending the \utxo. The only effect of allowing permissive mode on leaked \derived \utxo is that a steal attempt might \emph{force} a user to spend a \utxo earlier than they desired to, and in that case, a long period allows the user more flexibility choosing when to spend it.
\end{itemize}


\section{Our Proposed Protocol}\label{intro:proposal}

In this section, we review considerations crucial for combining FawkesCoin, lifted FawkesCoin and quantum canaries into a cohesive and concrete protocol which arguably answers the two questions posed in the introduction.

\subsection{Leaked \textsf{UTXO}s}\label{ssec:leaked}

So far, we have made a distinction between \hashed and non-hashed \utxos (which we further divided into several sets). In particular, some quantum-cautious spending methods treat \hashed \utxos differently than other \utxos.

The problem with this is that the operation of any spending method should only depend on data available on the blockchain. However, it is impossible to read off the fact that a \utxo is \hashed from the blockchain (as its public-key may have been leaked in other ways).

We thus approximate the set of non-\hashed \utxo by using the set of \emph{leaked} \utxos. A \utxo is considered \emph{leaked} if its public-key appears anywhere \emph{on the blockchain}. For the rest of the section, we use the term \hashed \utxo to mean a \utxo which is not leaked.

To account for the gap between the definitions, we propose a \emph{good Samaritan} mechanism which allows users to post public-keys not already on the blockchain to the mempool, and Samaritan miners to include them on the blockchain, thus transforming non-leaked \utxos into leaked ones. The miners obtain no fees for including the public-keys. By allocating 1 Kilobyte for good Samaritan reports, each block could contain about 30 reported addresses. The increase in both block size and the computational overhead of verifying the block is negligible.

Note that the good Samaritan mechanism should only be available in the pre-quantum era, as a quantum adversary could listen for reported public-keys and attempt to steal them.

The purpose of the good Samaritan mechanism is twofold. First, it allows users to report public-keys they encountered outside the blockchain. Second, it allows users who lost their secret-key but still hold their public-key to report it, making their \utxo \lost so they could spend it using permissive FawkesCoin when the quantum era arrives (see \cref{ssec:nfc}).

\secorssec{Compatibility of Methods}\label{intro:compat}

The properties in \cref{table:intro} only hold when each method is implemented separately. Carelessly implementing several methods together might compromise their security. For example, allowing lifted spending and FawkesCoin simultaneously renders FawkesCoin insecure: a quantum attacker who sees the revealed FawkesCoin transaction in the mempool can try to front-run by creating a competing lifted signature transaction.

Another example is when using FawkesCoin and Lifted FawkesCoin together. Spending a \derived \utxo in Lifted FawkesCoin requires revealing the master secret-key $\msk$. An adversary listening to the mempool can use $\msk$ to derive parent-keys she could use to spend any \utxo in the wallet using FawkesCoin. Hence, if FawkesCoin and Lifted FawkesCoin are implemented together, FawkesCoin must be modified such that the seed from which $\msk$ was derived must be posted instead of the parent key, which has the disadvantage of requiring the user to commit to their \emph{entire wallet} before they can start revealing it.

For that purpose, we propose \emph{not} to implement lifted spending directly, and to segregate FawkesCoin and Lifted FawkesCoin into \emph{epochs}, where in each epoch only one of the solutions is available.

We further require users of lifted FawkesCoin to complete a commit-wait-reveal cycle within the epoch. This acts to prevent users from committing to a \utxo during a lifted FawkesCoin epoch, and then spending it during a non-lifted FawkesCoin epoch, denying compensation form the miner who posted their lifted FawkesCoin commitment. We recall that in non-lifted FawkesCoin the user pays for the commitment as it is posted, regardless of whether the transaction is eventually revealed, so requiring that non-lifted FawkesCoin transactions are committed and revealed in the same epoch is unnecessary.

As the example above shows, segregating the solutions into epochs is insufficient, and some more steps are required for safely combining them. We discuss this further in \cref{ssec:combining}.

We stress that as long as any method that requires users to be online is implemented, this requirement holds even if other methods without this requirement are implemented.

\subsection{Combining FawkesCoin and Lifted FawkesCoin}\label{ssec:combining}

As will be briefly touched upon in \cref{intro:compat}, allowing FawkesCoin and Lifted FawkesCoin to operate in tandem compromises their security. To overcome this problem, we suggest segregating the operation of these solutions into epochs, where we expect Lifted FawkesCoin users to complete a commit-wait-reveal cycle within an epoch.

However, this still allows a vulnerability when spending derived transactions using Lifted FawkesCoin, since such a spend requires exposing $\msk$, which could then be used to spend any other transaction created by the same wallet using non-Lifted FawkesCoin.

The most immediate solution to that is to modify non-Lifted FawkesCoin such that spending a \derived \utxo requires committing and revealing the \emph{seed}. But this has the unfortunate implication that a user must commit to their entire wallet before they could reveal it.

To avoid this issue, we keep a \emph{registry} of known keys in the hierarchy. Let $\xsk_{par}$ be a key used to spend a derived transaction (either as a parent key in non-Lifted FawkesCoin or as $\msk$ in Lifted FawkesCoin). It is considered invalid to use any key derivable from $\xsk_{par}$ to spend a \derived \utxo in non-Lifted FawkesCoin.

The problem is that it is impossible to exhaust the entire space of derivable keys. Typically, hardware wallets use a very limited space of addresses. We can thus agree on a reasonably small set of \emph{regular} derivation paths $\sD$ assuming that most derived addresses of the blockchain are of the form $\msk_{P}$ where $P\in \sD$ and $\msk$ is a master secret-key of an existing HD wallet.

To account for users who might have a \utxo whose derivation path is not in $\sD$, we can allow users to post messages of the form $(\HH(\xsk),P_1,\ldots,P_k)$ to the blockchain. Once the key $\xsk$ is revealed, the paths $P_1,\ldots,P_k$ will be checked along with the regular derivation paths. Let $\sD_\xsk$ contain the set $\sD$ along with any irregular path which appeared alongside $\HH(\xsk)$ on the blockchain, and let $\sK_\xsk$ contain all keys of the form $\xsk_{P'}$ where $P'$ is a prefix of some $P \in \sD_\xsk$. The registry will store a lost of tuples of the form $(\sK_\xsk,b_\xsk)$ where $b_\xsk$ is the height of the block where $\xsk$ was included. Any FawkesCoin transaction spending a \derived \utxo using the key $\xsk_{par}$ will be considered invalid if there is some $\xsk$ such that $\xsk_{par}\in \sK_\xsk$ and the transaction was \emph{committed} in a block whose height is at least $b_\xsk$.

In addition, any \utxo whose $\pk$ is in $\sK_\xsk$ for some $\xsk$ is considered leaked, and cannot be spent as if it is hashed.

\secorssec{The Final Protocol}

\ifnum\masterthesis=0{The purpose of this section is to propose a concrete solution to questions 1 and 2 above by combining the methods described above.}\else{We now combine all the ingredients into one concrete protocol. For concreteness, our proposal contains some "magic numbers" such as the length of the delay period for FawkesCoin, the funding of the canary, the length of the dispute period, the delay between claiming the canary and commencing the quantum era, etc. We have done our best to pick reasonable values and justify their magnitude (based on the magic numbers used in Bitcoin, estimations of the size of the \utxo set and post-quantum signatures, and so on). We note that all these values could be easily tuned, as they only mildly affect the functionality and security of the network.}\fi

We propose the following:
\begin{itemize}
    \item Set up a quantum canary (see \ifnum\masterthesis=0{\ifnum\acmtops=0\los{\cref{sec:quantum_canaries}}{\cref{intro:canaries}}\else \cref{intro:canaries}\fi}\else{\cref{sec:procrast_canary}}\fi) with a bounty of 20,000 BTC. The bounty should be freshly minted for that purpose\footnote{We stress once again that the quantum canaries could be replaced with any other method to determine when the quantum era has started. Doing so does not affect the rest of the specification.}.
    \item Define the \emph{quantum era} to the period starting 8,000 blocks (about two months) after the canary is killed.
    \item Once the quantum era starts:
    \begin{itemize}
        \item Activate FawkesCoin (see \los{\cref{sec:fawkes}}{\cref{row:fc}}) and Lifted FawkesCoin (see \los{\cref{sec:lfc}}{\cref{row:lfs}}). Set up a rotation (see \cref{intro:compat}) such that the FawkesCoin epoch is 1,900 blocks long and the Lifted FawkesCoin epoch is 500 blocks long.
        \item Directly spending pre-quantum \utxos becomes prohibited. That is, FawkesCoin and Lifted FawkesCoin become the only methods for spending pre-quantum \utxos.
        \item The mode of operation for FawkesCoin is restrictive for any \naked \utxo whose \emph{address} was posted to the blockchain prior to 2013, and permissive for the rest of the \utxos. 
        \end{itemize}
\end{itemize}

A core component in our solution is Lifted FawkesCoin, which requires a hard-fork \los{ (see \ifnum\masterthesis=0{\cref{par:softfork}}\else{\cref{ssec:perlim_bitcoin_forks}}\fi)}{}. The hard-fork could be implemented at any stage and does not require waiting for the quantum era.\ifnum\acmtops=0 \los{ A hard-fork is also required to implement permissive FawkesCoin, and makes it easier to set up the bounty for the canary (see \cref{ssec:canary_incent}).}{}\fi

Once the quantum era commences, all \naked \utxos whose public-keys leaked before 2013 would be burned, unless they were spent before the quantum era. Hence, the number of blocks we should wait for after the canary has been killed before the quantum era should be sufficient for all such \utxos to be migrated. On the other hand, we do not want to wait too long since longer waiting times increase the risk that the quantum adversary will scale sufficiently to loot pre-quantum \utxos before the quantum era starts. At the start of 2013, the \emph{entire} \utxo set included about 4 million \utxos. We hence use 4 million as a rough upper bound on the number of \utxos at risk of being burned. \ifnum\acmtops=0 Assuming about half of the block space is spent on signatures (see \ifnum\masterthesis=0{\cref{ssec:nist-qpds}}\else{\cref{ssec:perlim_bitcoin_pqsig}}\fi), spending 4 million \utxos should require about 800 blocks. \else In the full version \cite{SW23} we argue that currently about half of the block space is used for signature, whereby spending 4 million \utxos should require about 800 blocks. \fi However, we should not spend the waiting time too close to 800 blocks since 1. we should account for a possible increase in the number of \utxos whose public-keys leaked before 2013, and 2. it is plausible that once the canary is killed, there will be a rush to spend pre-quantum \utxos (including those who are not at the risk of burning) to post-quantum addresses, and there should be sufficient leeway for owners of soon-to-be-burned \utxos to spend them in this scenario. We thus suggest a waiting period of 8,000 blocks (or about two months) as, on the one hand, it seems unreasonable that quantum computers will increase in scale during this period, and on the other, it gives owners of soon-to-be-burned \utxos ample time to spend their transactions.

The 1,900 block epoch for FawkesCoin gives a 1,800 block window for posting commitments. Commitments can not be posted during the last 100 blocks, as this will not leave enough blocks in the epoch to complete the required waiting time.

For similar reasons, the last 300 blocks of a Lifted FawkesCoin epoch could not contain reveals since we have a 100 blocks wait period, 100 blocks during which the user can post a reveal, and 100 blocks during which the miner can post a proof of ownership in case the user failed to post a reveal. Thus, a 500 blocks Lifted FawkesCoin period gives a 200 block window for posting commitments.

We propose implementing FawkesCoin along with Lifted FawkesCoin, with the intention that a user with no access to post-quantum \utxos could use Lifted FawkesCoin initially to migrate pre-quantum \utxos to post-quantum \utxos, and then use these post-quantum \utxos to pay the FawkesCoin fees in future transactions. We thus expect that users will gradually shift from using Lifted FawkesCoin to using non-Lifted FawkesCoin.

We do not recommend enabling permissive mode before the quantum era since that would require users to be online\los{ (see \cref{ssec:nfc})}{}. We stress that permissive (or even unrestrictive) mode should not be used at all for \utxos prior to 2013 since HD wallets were only introduced in 2013, so \emph{all} \naked \utxos prior to 2013 are known to be non-derived and are thus \stealable.

Enabling unrestrictive mode in the quantum era for \utxos made after 2013 is desirable, as it increases the set of spendable pre-quantum \utxos while not increasing the available quantum loot\los{ (see \cref{ssec:fc_loot})}{}. Since unrestrictive mode already requires users to be online and Lifted FawkesCoin already requires a hard-fork, there is no downside to enabling permissive mode.

Note that we do \emph{not} recommend directly implementing lifted spending, as this method does not provide a meaningful benefit over Lifted FawkesCoin, and is incompatible with non-Lifted FawkesCoin (see \cref{intro:compat}).

\ifnum\masterthesis=0{
\begin{remark}\label{rem:magic}
    For concreteness, our proposal contains some "magic numbers" such as the length of the delay period for FawkesCoin, the funding of the canary, the length of the dispute period, the delay between claiming the canary and commencing the quantum era, etc. We have done our best to pick reasonable values and justify their magnitude (based on the magic numbers used in Bitcoin, estimations of the size of the \utxo set and post-quantum signatures, and so on). We note that all these values could be easily tuned, as they only mildly affect the functionality and security of the network.
\end{remark}}\fi


\ifnum\lipics=0
\ifnum\acmtops=0
\section{Open Questions and Further Research}\label{sec:disc}
\begin{itemize}
    \item \textbf{Reducing the size of lifted signatures}\quad The most pressing issue is the size of lifted signatures, discussed in \cref{ssec:liftedsize}. The size of the signature directly determines the spendability threshold one using Lifted FawkesCoin. Estimating and optimizing the sizes of the signatures we use is an important next step toward implementing our solution. 
    
    \item \textbf{Seed-lifting without exposing $\msk$} \quad The seed-lifted scheme we present in \cref{ssec:seedlift} requires exposing the master secret-key of the HD wallet. While this requirement does not affect the security of our solution, it does cause some inconveniences, so devising a method for cautiously spending \derived \utxos without exposing $\msk$ might be desirable. There are two main drawbacks to exposing $\msk$. First, exposing $\msk$ implies that once a single commitment has been revealed, the entire wallet must be spent using lifted FawkesCoin (see \cref{ssec:combining}). Another concern is that HD wallets use the same seed to derive keys to many different cryptocurrencies, so carelessly exposing $\msk$ on one currency might compromise \utxos from different currencies whose address was derived from the same seed.
    
    The most direct approach to spend a \derived \utxo with derivation path $P$ is to instantiate $\picnic$ with the function $\der(\kdf(\cdot),P)$. However, this requires $\picnic$ to be secure when using the same secret-key in different instantiation. It is unclear whether $\picnic$ is secure against such attacks, though it is very easy to modify $\picnic$ such that it remains $\seufcma$ secure, but becomes completely broken against such attacks (e.g. by modifying the signature to contain the $i$th bit of the secret-key, where $(i,s)$ is the first step in $P$). Hence, making this approach secure requires introducing a formal security notion that prohibits such attacks, and proving that $\picnic$ (or some modification thereof) satisfies this stronger form of security. We point out that another drawback of this approach is that it greatly increases signature sizes.

    We leave the problem of removing the need to expose $\msk$, either following the approach above or by coming up with a different solution, to future research.

    \item \textbf{Instantiating the canary}\quad The discussion in \cref{ssec:canarypuzzle} suggests that the canary puzzle should be forging a signature of an \ecdsa scheme instantiated with an elliptic curve similar to \secp. However, it still remains to choose a particular curve, and provide a way to sample a nothing-up-my-sleeve public-key for that curve.
    \nextver{\item \textbf{Reducing lifted signature sizes}\quad All of our lifted spending methods currently produce prohibitively large signatures. We believe it is possible to optimize the \picnic scheme to our particular choice of a one-way function similarly to how it was optimized for LowMC in \cite{KZ20}.}
    \item \textbf{Further analysis of canary entities}\quad 
    Our game theoretic analysis in \cref{ssec:canarygame} assumes the parties have perfect information about their advesary's capabilities. One step towards a more realistic model is to consider the game theory of canaries in the setting of imperfect information.
\end{itemize}
\fi
\fi

\ifnum\shortver=0



\fi


\ifnum\anonymous=0
\fi
\ifnum\sigconf=1
    \bibliographystyle{ACM-Reference-Format}
\else
    \ifnum\lipics=1
        \bibliographystyle{plainurl}
    \else
        \bibliographystyle{alphaabbrurldoieprint}
    \fi
\fi

\ifnum\lipics=0
    {\footnotesize\bibliography{main}}
\else
    \bibliography{main_lipi}
\fi

\ifnum\shortver=1
    \appendix

\fi
\ifnum\shownomenclature=1
\printnomenclature[1in]
\fi

\ifnum\acmtops=0
\appendix
\section{Preliminaries and Definitions}\label{sec:pre}

\subsection{Digital Signatures}\label{prelim:ds}

A \emph{digital signature scheme} is a cryptographic primitive that allows users to sign messages such that only they can sign the message, but anyone can verify the authenticity of the message. In digital signature schemes, the signer creates a \emph{secret signature key} that can be used to sign messages and a \emph{public-key} that can be used to verify these messages.

\subsubsection{Security of Digital Signatures}\label{ssec:dssec}

The security notion we consider is \emph{existential unforgeability under chosen message attack} ($\mathsf{EUF\mhyphen CMA}$). In this notion, the adversary is given access to the public-key as well as access to a \emph{signing oracle} which she could use to sign any message she wants. Her task is to produce a signature for any message she did not use the oracle to sign. A scheme is \emph{\seufcma secure} if an efficient adversary cannot achieve this task with more than negligible probability (for a more formal treatment of the security of digital signatures we refer the reader to \cite{Gol04,KL14}). 

Extending the security notions of digital signatures to accommodate general quantum adversaries is not quite straightforward. The main difficulty is in the setting where the adversary has access to sign messages of their own, and they can sign a \emph{superposition} of messages. The established classical notions of security do not generalize directly to this setting, since the notion of "a message she did not use the oracle to sign" becomes ill-defined when discussing superimposed queries.

However, in the setting of cryptocurrencies, the signature oracle reflects the adversary's ability to read transactions signed by the same public-key off the blockchain. This ability is captured even when assuming that the adversary, albeit quantum, may only ask the oracle to sign classical messages. A scheme that remains secure against quantum adversaries with classical oracle access is called \emph{post-quantum} $\mathsf{EUF\mhyphen CMA}$ secure. For brevity, we use the term $\mathsf{EUF\mhyphen CMA}$ security to mean post-quantum $\mathsf{EUF\mhyphen CMA}$ security, unless stated otherwise.

\subsubsection{Elliptic Curve based Signatures}\label{sssec:ec}

The \ecdsa and Schnorr signature schemes are based on a particular mathematical object called an \emph{elliptic curve}. Bitcoin uses the \ecdsa signature scheme instantiated with the \secp curve, which is considered to admit 128-bit security against \emph{classical} attackers \cite{Bit22b}. The Taproot update \cite{WNT20} replaces \ecdsa with Schnorr signatures instantiated with the same \secp curve, which is also considered to admit 128-bit security \cite{WNR20}.

While the way signatures are produced and verified in both these schemes is beyond the scope of the current work, we do require some understanding of how keys are generated. Fortunately, the key generation procedure is the same for \ecdsa and Schnorr signatures (provided they were instantiated with the same curve and the same basis point).

Given an elliptic curve, one can give the element of the curve the structure of an abelian group known as the \emph{elliptic curve group}. Recall that given any group $\Gamma$, we can define for any element $G\in \Gamma$ and any natural number $n$ the group element $n\cdot G$ by defining $0\cdot G = 0_\Gamma$ (where $0_\Gamma$ is the identity element of $\Gamma$) and $n\cdot G = (n-1)\cdot G + G$. We let $o(\Gamma)$ be the order of the group and say that $G$ is a \emph{generator} of $\Gamma$ if $\{n\cdot G\mid n=0,\ldots,o(\Gamma)\}=\Gamma$. The $\secp$ curve has the property that it has a prime order, whereby \emph{any} element besides the identity is a generator.

Instantiating a signature scheme requires specifying not only an elliptic curve group $\Gamma$, but also a fixed generator $G$ of the group. Given the curve and generator, the secret-key is a uniformly random number $\sk\gets \{0,\ldots,o(\Gamma)\}$ and the matching public-key is $\pk = \sk\cdot G$. As we further discuss in \cref{ssec:canarypuzzle}, the considerations behind choosing the curve and generator are highly involved, and are relevant both to the security and the efficiency of the resulting scheme. 

To recover the public-key from a secret-key, one needs to be able to compute $\sk$ from $G$ and $\pk = \sk\cdot G$. Differently stated, one needs to compute the \emph{discrete logarithm} $\log_G(\pk)$. The \emph{hardness of discrete logarithm} is the assumption that solving such equations is infeasible in the average case. As we'll shortly review, quantum computers disobey this assumption, which is the root of the problem at hand.

In practice, elements of $\ecdsa$ are encoded into binary strings, for a group element $H\in\Gamma$, let $\widetilde{H}$ denote its binary representation. We define the function $\pkec(\sk) = \widetilde{\sk\cdot G}$, and note that it is an injective function with domain $\{0,\ldots,o(\Gamma)\}$. Furthermore, $\pkec$ is a \emph{group homomorphism} from $\ZZ_{o(\Gamma)}$ to $\Gamma$: $\pkec(k+l) = (k+l)\cdot G = k\cdot G + l\cdot G = \pkec(k) + \pkec(l)$.

In practice the secret key $\sk$ is not sampled uniformly from $\{0,\ldots,o(\Gamma)\}$ but from $\zo^{\log_{2}\left\lceil o(\Gamma)\right\rceil }$ (that is from strings of the minimal length required so that each possible number in $0,\ldots,o(\Gamma)$ has a unique representation. However, the elliptic curve group is typically chosen such that almost all keys have a unique representation. For example, in the \secp curve $\sk$ is given as a string of $256$ bits whereas there are more than $2^{256} - 2^{64}$ group elements, so the fraction of elements of $\ZZ_{o(g)}$ which admit two representations is smaller than $2^{-192}$. In practice, this is overlooked and a key is generated by choosing a uniformly random string. We also overlook this detail and assume that $\pkec$ is injective.

The binary representation $\tilde{H}$ of a curve point $H$ has several different formats with different properties. In practice in general, and Bitcoin in particular, different formats appear in different contexts. Most of the current work work is agnostic to how the string is actually formatted, and treat $\pkec$ as a bijection onto its domain. The only place where we appeal to the format is in \cref{ssec:hdwallet}, where we use the fact that in HD wallets following the specification of BIP-32 \cite{Wui13}, $\sk$ is encoded as a $32$ byte string whereas $\pk$ is encoded as a $33$ byte string (and in particular it is impossible that $\sk = \pkec(\sk)$).

\subsubsection{Quantum Attacks on \textsf{ECDSA}}

Shor's algorithm \cite{Sho94} was the first to break previously considered unbreakable cryptographic schemes by efficiently factoring large numbers. Shor's techniques were generalized in \cite{BL95} to solve the discrete logarithm problem in arbitrary groups, including elliptic groups, thus proving that \ecdsa and Schnorr signatures are \emph{not} post-quantum (i.e., the number of bits of security of these schemes over any group is polylogarithmic in the number of bits required to describe an arbitrary element of the group). We refer the reader to \cite{LK21} for a fairly comprehensive survey of the state-of-the-art algorithms optimized to solve the discrete logarithm problem in several types of elliptic curve groups\footnote{The "type" of a curve is determined by the field it is defined above, the form in which the parametric equation is given, and the types of coordinates used}.

When designing quantum canaries (see \cref{sec:quantum_canaries}), we are concerned with curves similar to the \secp curve used in Bitcoin. The \secp curve is over a prime field and is given in Weierstrass form. Fortunately, quantum algorithms for solving discrete logarithms over prime curves given in Weierstrass form were analyzed by several authors. The first concrete cryptanalysis is given by Roetteler et al. \cite{RNSL17}, which provide an explicit algorithm optimized to minimize the required number of logical qubits. The algorithm in \cite{RNSL17} is vastly improved by Häner et al. \cite{HJN+20}, who provide a more efficient algorithm and also consider space-time trade-offs. They provide algorithms minimizing either the number of logical qubits, the number of $T$ gates, or the depth of the circuit. Unfortunately, \cite{HJN+20} do not analyze any of these metrics as a function of the size of the field, but rather compute them for several known curves and compare the result with \cite{RNSL17}. 

\subsubsection{Post-Quantum Signature Schemes}
To address the quantum vulnerability of contemporary signature schemes, the United States National Institute of Standards and Technology (\textsf{NIST}) invited protocol designers to submit post-quantum signature schemes. In July 2022, \textsf{NIST} announced that three post-quantum signatures will be standardized: CRYSTALS-Dilithium \cite{DKL+18}, FALCON \cite{FHK+18}, and SPHINCS+ \cite{BHK+19}.

A survey of the chosen schemes and their performance is available in \cite{RWC+21}. The blog posts \cite{Wes21,TC22,Wig22} provide a more approachable (yet less formal) survey of the state-of-the-art post-quantum signature schemes in general, as well as the \textsf{NIST} standardized schemes.
\subsubsection{\textsf{Picnic} Signatures}\label{ssec:picnicov}

The principal tool we use for signature lifting is the \emph{\picnic} signature scheme of Chase et al. \cite{CDG+17}. The \picnic scheme signature scheme can be instantiated using \emph{any} post-quantum one-way function $f$  to obtain a signature scheme which is post-quantum \seufcma secure in the \textsf{QROM} (recall \cref{ssec:rom}). In the obtained scheme, a secret-key is a random point $x$ in the domain of $f$, and the corresponding public key is $f(x)$.

The \picnic scheme, instantiated with a particular block cipher called \emph{LowMC}, was submitted to NIST for standardization. Their designed prevailed the first two rounds of the competition. It was decided that \picnic will not proceed to the third round due to the novelty of techniques it applies compared with other candidates, however, it was decided to retain \picnic as an alternative candidate \cite{AASA+20}. This means that, while \picnic was not chosen to be NIST standardized, it successfully withstood heavy scrutiny.

\subsection{The Random Oracle Model}\label{ssec:rom}


When considering the security of constructions which involve explicit hash functions such as \sha or \shaf, it is difficult to make formal arguments about their security. A common way to overcome this is using the \emph{random oracle model}. The random oracle model assumes that there is an oracle $\HH$ accessible to anyone, such that $\HH(x) = f_{|x|}(x)$ where for each $n$ the function $f_n$ was uniformly sampled from the set of functions from $\zo^n$ to $\zo^\ell$ (where $\ell$ is fixed). The \emph{quantum} random oracle model further assumes that users have access to the unitary $\ket{x,y}\mapsto \ket{x,y\oplus \HH(x)}$. 

It is common practice in cryptography to analyze constructions assuming that random oracles are used instead of hash functions, and then instantiate them with hash functions considered secure for the application at hand. For a more extensive introduction to the random oracle model and the quantum random oracle model we refer the reader to \cite[Section~8.10.2]{BS20b} and \cite{BDF+11} respectively.

Random oracles have a particularly useful property that if applied to a "sufficiently random" distribution, the output is indistinguishable from uniformly random. The measure of randomness appropriate for our settings is that of \emph{guessing probability}:

\begin{definition}[Guessing probability]\label{defn:guessing}
    Let $\sS$ be a finite distribution, the \emph{guessing probability} of $\sS$ is $\gamma_\sS\max_s\PP\left[t=s\mid t\gets \sS\right]$. 
\end{definition}

That is, the guessing probability is the probability of the most likely output of sampling from $\sS$\footnote{Readers familiar with the min-entropy function $H_{min}$ might notice that $\gamma_\sS = 2^{-H_{\min}(\sS)}$}. It is called the \emph{guessing} probability as it is also the best probability with which an adversary can guess the outcome of sampling from $\sS$, as one can prove that always guessing the most likely outcome is an optimal strategy.

Note that a distribution can have a very small guessing probability and still be easily distinguishable from uniformly random. For example, consider the distribution over $(s,pw)$ where $s$ is a uniformly random string of length $n$ and $pw$ is some fixed string of length $k$. Then $\gamma_{(s,pw)} = 2^{-n}$, though it is easily distinguishable from a random string of length $n+k$.

\begin{proposition}[\cite{BS20b} Theorem 8.10, adapted]\label{prop:kdf}
    Let $\sS$ be a finite distribution. If $\HH$ is modeled as a random oracle and $\gamma_\sS = \negl$ then the distribution $\HH(\sS)$ is computationally indistinguishable from a random distribution.
\end{proposition}


Throughout this work, we use $\HH$ to represent the hash used in the construction we analyze (which is always either \sha or \shaf), and assume $\HH$ is modeled as a random oracle. In practice, \sha and \shaf are based on a construction called the Merkle–Damgård transform, which is known to be an unsuitable replacement for a random oracle in some particular settings. In \cref{app:mdrom} we expand on the properties of Merkle–Damgård and justify our modeling thereof as a random oracle.

\subsubsection{Merkle–Damgård as a Random Oracle}\label{app:mdrom}

In \cref{ssec:rom}, we stated that throughout this paper we replace the \sha and \shaf hashes with a random oracle \HH. More generally, these hashes are constructed by applying a technique called the \emph{Merkle–Damgård transform}, which transforms a collision-resistant compressing function $f$ (e.g., $f:\{0,1\}^{2n}\to\{0,1\}^n$) to a collision-resistant function $\HH_f$ with a variable input length.

The purpose of this appendix is to isolate the properties we require of $\HH$, and argue that they also hold for $\HH_f$, when $f$ is modeled as a random oracle.

The "bad" property of $\HH_f$ is that it admits a \emph{length extension attack}: given a hash $\HH_f(m)$ of some \emph{unknown} string $m$, and oracle access to $f$, an adversary can easily compute $\HH_f(m\|m')$ for a large variety of strings $m'$. One can easily abuse this property to construct a scheme that is secure in the \textsf{ROM} but becomes broken when $\HH_f$ replaces the oracle (see e.g. the "prepend the key" \textsf{MAC} construction in \cite[Section 8.7]{BS20b}).

Our analysis appeals to two properties of a random oracle: that it is collision-resistant and that the output of a "sufficiently random" input is indistinguishable from uniform.

The original purpose for the Merkle–Damgård transform was to provide collision resistance, and it is indeed known to hold (see \cite[Theorem~8.3]{BS20b}).

We assume the second property two cases: for key derivation from a seed (see \ifnum\masterthesis=0{\cref{ssec:kdf}}\else{\cref{ssec:prelim_other_kdf}}\fi) where it is known that the input starts with sufficiently many random bits, and when hashing a string sampled from a distribution known to be indistinguishable from uniform (typically generated in the first case). Hence, it suffices to prove the following:

\begin{proposition}\label{prop:mdkdf}
    Let $f:\zo^{2n}\to \zo^n$ be a random function. Let $\sS$ be a distribution of strings of length at least $n$ such that the distribution on the first $n$ bits is computationally indistinguishable from uniform, then $\HH_f(\sS)$ is computationally indistinguishable from uniform for an adversary with oracle access to $f$.
\end{proposition}

\begin{proof}
    We first observe that if a distribution $\sT$ on strings of length $n$ is indistinguishable from uniform then $\gamma_\sT=\mathsf{negl}(n)$: say that $\PP_\sT(x)=\alpha$, and consider an adversary $\sA$ that samples $x'$ and outputs $1$ iff $x'=x$. Clearly $\sA$ is polynomial, and we have $\PP[\sA^{\sT}()=1]= 1/\alpha$ and $\PP[\sA^{\mathcal{U}}()=1]= 2^{-n}$. By hypothesis we have that $|1/\alpha - 2^{-n}| = \mathsf{negl}(n)$. It follows that 
    $\PP_\sT(x) = \mathsf{negl}(n)$
    for all strings $x$, hence $\gamma_\sT = \mathsf{negl}(n)$.

    We also observe that if $\sT = (\sT_1,\sT_2)$, then $\gamma_{\sT} \le \gamma_{\sT_1}$: let $(x_1,x_2)$ that maximizes $\PP_{\sT}[(x_1,x_2)]$, then $\gamma_{\sT_1} \ge \PP_{\sT_1}[x_1] = \sum_{x_2'} \PP_{\sT}[(x_1,x_2')] \ge \PP_{\sT}[(x_1,x_2)] = \gamma_{\sT}$.

    Recall that the first step of applying Merkle-Damgård to an input $x$ is to create a padded $\hat{x}$ whose length is a multiple of $n$ (the specifics of the padding are crucial for collision resistance, but irrelevant to the current proof), and writing $\hat{x} = x_1\|\ldots\|x_k$ where each $x_i$ is a string of length $n$. Crucially, if $|x|\ge n$ then $x_1$ is exactly the first $n$ bits of $x$.

    We then compute $y_1 = f(0^n,x_1)$ and $y_j = f(y_{j-1},x_j)$ for $j=2,\ldots,k$, until finally we obtain $\HH_f(x) := y_k$.

    By hypothesis if $x\gets \sS$ then $\gamma_{x_1} = \mathsf{negl}(n)$. From the observation above, we get that $\gamma_{(0,x_1)} = \mathsf{negl}(n)$. Since $f$ is random it follows from \cref{prop:kdf} that $y_1 = f(0^n,x_1)$ is indistinguishable from uniform. Hence $\gamma_{y_1} = \mathsf{negl}(n)$, so we can repeat the argument to get that $\gamma_{(y_1,x_2)} = \mathsf{negl}(n)$ and $y_2$ is indistinguishable from random. We repeat the process $k$ times to obtain that $y_k = \HH_f(x)$ is also indistinguishable from uniform.
\end{proof}

For completeness, we point out the places where we use the assumption that $\HH$ is modeled as a random oracle, and observe that we only use the two properties above:
\begin{itemize}
    \item In \ifnum\masterthesis=0{\cref{ssec:kdf} }\else{\cref{ssec:prelim_other_kdf} }\fi we argue that $2048$ successive applications of $\shaf$ constitute a \textsf{PBKDF}. By successive applications of \cref{prop:mdkdf} we get that $\shaf^{2048}$ is indeed a \textsf{PBKDF}, as long as we assume that the first $n$ bits of the input are indistinguishable from uniform. In practice, the key-generation procedure of HD wallets uses an input that starts with a long uniform string, so this assumption is justified.
    \item Throughout \ifnum\masterthesis=0{\cref{ssec:hdwallet} }\else{\cref{ssec:sl_seed_hd} }\fi we prove that if $\HH$ is a random oracle, then several functions derived from $\HH$ have some collision resistance properties. These arguments are readily adaptable to the weaker assumption that $\HH$ is collision-resistant.
    \item In \ifnum\masterthesis=0{\cref{ssec:zk} }\else{\cref{ssec:sl_seed_proof} }\fi we assume that $\HH$ is modeled as a random oracle. We only use this fact directly to argue that $\picnic(\HH)$ is secure, which only requires $\HH$ that is one-way. The fact that $\HH$ is one-way already follows from the assumption that it is collision-resistant. All other appeals to the randomness of $\HH$ are to apply the statements of \ifnum\masterthesis=0{\cref{ssec:kdf} }\else{\cref{ssec:prelim_other_kdf} }\fi and \ifnum\masterthesis=0{\cref{ssec:hdwallet}}\else{\cref{ssec:sl_seed_hd}}\fi.
\end{itemize}


\subsubsection{Key-Derivation functions}\label{ssec:kdf}

A \emph{password-based key-derivation function} (\textsf{PBKDF}) is a function used to derive a secret-key for a cryptographic application from a string that is not necessarily uniformly random (e.g. the key also contains a user selected password, without assuming anything about the distribution the password was sampled from) called the \emph{seed}. For a formal introduction we refer the reader to \cite[Section~8.10]{BS20b}.

The property we require from a \textsf{PBKDF} is that the resulting key is indistinguishable from random, provided that the input seed was sampled from a distribution with sufficient guessing probability (recall \cref{defn:guessing}). \cref{prop:kdf} establishes that a random oracle is a \textsf{PBKDF}.

Most Bitcoin wallets use a particular \textsf{PBKDF} specified in \cite{PRV+13}, where the seed is comprised of a user-chosen password and a uniformly random binary string (encoded in the form of a \emph{mnemonic phrase} for the sake of human readability). The \textsf{PBKDF} therein serves another purpose: to provide some protection for the user from an adversary which has access to the binary string but not to the password (as the string is usually stored inside the wallet, but the user is required to input the password with each use). This is achieved by making the \textsf{PBKDF} computationally heavy to make dictionary attacks on the password less feasible. Towards this end, the \textsf{PBKDF} described therein applies 2048 iterations of \textsf{SHA-512} to the input.

\subsection{Bitcoin and Blockchain}\label{prelim:bitcoin}

In this section, we overview some of the aspects of Bitcoin relevant to our discussion. We assume the reader is familiar with core concepts of Bitcoin such as transactions and \utxos, chain reorganizations, etc. For a review of these concepts, we refer the reader to \cite{NBF+16}.

\subsubsection{Coinbase Cooldown}\label{par:cooldown} In a reorg scenario, most transactions removed by the reorg can be rebroadcast to the mempool, and any transaction broadcast by an honest user will eventually be included in the new chain. The only scenario where a transaction becomes invalid is if a dishonest user attempts a double spend by broadcasting conflicting transactions to the two sides of the fork. However, this is no longer the case for reverted coinbase transactions, as their validity relies on the block that included them. Reverting a coinbase transaction renders it, and all transactions spending coins minted in that coinbase transaction, invalid. To avoid a scenario where a valid transaction made by an honest user becomes invalid, Bitcoin imposes a cooldown of one hundred blocks before the coinbase transaction can be spent (see the \href{https://developer.bitcoin.org/devguide/block_chain.html}{Bitcoin developer guide}).

\subsubsection{Software Forks}\label{par:softfork} A \emph{software fork} is a change to the code that nodes are expected to run that affects the conditions under which a block is considered valid. There are two types of software forks, a \emph{soft-fork} and a \emph{hard-fork}. The difference is that in a hard-fork, there exist blocks that the new version considers valid while the old version does not. In order to adopt a soft-fork, only the miners are required to update their nodes, and non-mining nodes will operate correctly even with the outdated version. A hard-fork, on the other hand, causes the chain to split into two separate chains that can not accept each other's blocks. Notable examples of hard-forks include Bitcoin Cash (forked from Bitcoin) and Ethereum Classic (forked from Ethereum).

We illustrate the two types of forks with two scenarios which occur in our methods:
\begin{itemize}
    \item Paying the coin in a \utxo to a different address than the one specified in the transaction (e.g., paying the \utxo of a non-revealed commitment to a miner who posted a proof of ownership in Lifted FawkesCoin, see \cref{sec:lfc}). This behavior can be implemented in a soft-fork by not spending to transaction to its intended address, but rather spending it to a new anyone-can-spend \utxo. Updated miners who enforce the rules of the protocol will know to consider attempts to spend the anyone-can-spend \utxo to any address but the one dictated by the protocol invalid. However, from the point of view of non-updated nodes, spending the \utxo to \emph{any} address is valid.
    \item Spending a \utxo without including a valid signature (e.g., spending a \lost \utxo in permissive FawkesCoin, see \cref{ssec:pfc}). This scenario can not be implemented in a soft-fork, as currently there is no way to spend a \utxo without including a signature.
\end{itemize}

\subsubsection{Dust} \label{par:dust}  \utxos whose value is too small to cover the costs of spending them are called \emph{Dust}. The difference between dust and unspendability is that a \utxo can be spendable with respect to one spending method but not the other. For a \utxo to be considered \emph{dust}, it must not be spendable using \emph{any} method available. Since the transaction fees fluctuate, the set of dust \utxos changes in time.

Any changes to the operation of Bitcoin that increase the transaction sizes will also incur an increase in the amount of dust. In particular, current post-quantum signature schemes require much larger keys and/or signature sizes than the  \ecdsa scheme currently used in Bitcoin. Moreover, variations on the Bitcoin protocol can cause the dust threshold to increase beyond the transaction fee, as indeed happens in Lifted FawkesCoin (see \cref{sec:lfc}).

\subsubsection{Quantum Threats on Blockchains}\label{ssec:qthreats}

Roughly speaking, quantum computers affect Bitcoin on two different fronts: quantum mining and attacks on pre-quantum cryptography.

A common misconception is that quantum computers have no drastic effects on Bitcoin mining (beyond increased difficulty due to Grover's quadratic speedup). This was debunked in~\cite{Sat20,LRS19}. Our work is orthogonal to aspects related to quantum mining, especially since quantum attacks on \secp are projected to occur a few years before quantum mining starts~\cite{ABL+17}.

The most immediate risk is in the form of \utxos with leaked public-keys. A quantum adversary could use the public-key to sign arbitrary messages and, in particular, spend any \utxo whose address is this public-key. The public-key can be exposed in a variety of ways, including (but not limited to):
\begin{itemize}
    \item \ptpk \utxos: when Bitcoin just launched, all transactions would contain the \ecdsa public-key of the recipient. Hence, all \utxos in the \utxo set would contain the public-key used to verify transactions spending them. That is, \emph{none} of the \utxos were \hashed. In 2009, in order to conserve space, \ptpkh transactions which only contain a hash of the public-key were introduced. \ptpkh soon became the standard, and as of today, no \ptpk \utxos are created. However, the \utxo set still contains a few old \ptpk \utxos whose balance totals about 2 million bitcoin \cite{Del22}.
    \item Reused addresses: while considered bad practice, reusing the same address for several \utxos is a very common habit. A reused address in itself is not quantumly compromised, but once such a \utxo is spent, the remaining \utxos with the same address become compromised. A survey by Deloitte estimates \cite{Del22} that at least 2 million bitcoins are stored in reused \utxos whose key has been exposed this way. A survey by BitMex \cite{Bit22} estimates that about half of the Bitcoin transaction throughput is to reused addresses\footnote{It should be noted that reusing addresses poses security risks beyond quantum vulnerabilities -- particularly compromising the anonymity of users -- and is generally considered a bad security practice.}.
    \item Taproot: the Taproot update to Bitcoin \cite{WNT20} uses a new form of \utxo called \textsf{P2TR} in which the public-key itself is posted like in \ptpk \utxos. Taproot was created with the intention to become the new standard and replace other spending methods, which might greatly increase the amount of leaked public-keys. Since its deployment in November 2021 the adoption of Taproot has been steadily increasing and as of August 2023, about 23\% of newly created \utxos are \textsf{P2TR} \cite{Tra23}.
    \item Software forks: when a chain forks into two independent chains (such as Bitcoin and Bitcoin Cash, or Ethereum and Ethereum classic), any \utxo from before the splitting point coexists on both chains. Thus, in order to spend money on one of the chains, the owner must expose her key, which could be used by an attacker to spend the owner's corresponding \utxo on the other chain. See \cite{IKK20} for an analysis of address reuse across Bitcoin and Bitcoin Cash.
\end{itemize}

\subsubsection{Using Post-Quantum Signatures on the Blockchain}\label{ssec:nist-qpds}
The most straightforward way to address \ifnum\masterthesis=1 the threats described in the previous section\else quantum threats\fi is to introduce post-quantum signatures. However, this requires that all users migrate their \utxos to post-quantum \utxos \emph{before} a quantum adversary emerges.

According to bitcoinvisuals.com, the average Bitcoin transaction in the last two years is about 600 bytes long and contains three inputs. Each such input contains a compressed \ecdsa signature and an \ecdsa public-key whose combined length is at most 103 bytes. Hence, at full capacity, about half of the block space is used for signature data.

As of the time of writing, there are about 85 million \utxos in the Bitcoin \utxo set \cite{Blo22}. Assuming a block size of one megabyte and a block delay of 10 minutes, it follows that migrating every single \utxo to a post-quantum address would require about five months. Though, even if it was possible to rally all users to migrate their \utxos, a quantum adversary could still loot \lost \utxos.

Another crucial aspect of post-quantum signature schemes is their increased storage and computational requirements compared to current signatures. The most immediate consequence is the increase in public-key and signature lengths, which directly decreases the throughput of Bitcoin.

The most space-efficient post-quantum signature endorsed by NIST is the Falcon512 scheme, whose combined public-key and signature size is 1532 bytes, which is about 15 times larger than the \ecdsa equivalent, making the average transaction about eight times larger. Crystal-DILITHIUM is the second-best scheme in terms of combined public-key and signature storage, with a combined length of 2740 bits, inflating the size of an average transaction by a factor of about 14.

\begin{remark}
    One concern about implementing post-quantum primitives too soon is that they have not undergone the same scrutiny as contemporary cryptosystems, and weaknesses might be revealed therein. One solution for that is to use an \emph{hybrid} approach, where (as long as pre-quantum cryptography is still considered safe), the signer must produce \emph{both} an \ecdsa (or any other widely used signature scheme) and a signature produced by the newly implemented post-quantum scheme. This way, users are still protected by the pre-quantum scheme should such vulnerabilities arise. Using a hybrid approach to transition to post-quantum cryptography was discussed in \cite{BHMS17}. Such a hybrid approach was recently applied in \cite{GKP+22}.
\end{remark}

\section{Hierarchical Deterministic wallets}\label{ssec:hdwallet}

Bitcoin HD wallets \cite{Wui13} allow a user to only store one \emph{master secret-key} from which many $\ecdsa$ key-pairs can be derived. The master secret-key itself is usually derived using a \textsf{PBKDF} (see \cref{ssec:kdf}). In this appendix, we present and analyze the construction of \cite{Wui13}. Our goal is proving \cref{prop:preimage}, which is essential for proving \cref{thm:seedlift}, establishing the security of seed-lifting.

All keys derived by an HD wallet are actually \emph{extended keys}. An extended secret-key is of the form $\xsk = (\sk,c)$ where $\sk$ is an \ecdsa secret-key and $c$ is a pseudo-random string. The corresponding public-key is $\xpk = (\pk,c)$ where $\pk = \pkec(\sk)$ (recall \ifnum\masterthesis=0{\snote{ref}}\else{\cref{ssec:prelim_ds_ec}}\fi).


Given a parent extended secret-key, $\xsk = (\sk,c)$ we define the \emph{$i$th non-hardened} child $\HH^{nh}(\xsk)$ as following:
\begin{itemize}
    \item Let $\HH$ be collision resistant with output length 512 (BIP-32 uses \shaf).
    \item Let $\HH_L$ and $\HH_R$ be the first 256 bits and last 256 bits of $\HH$ respectively.
    \item Set $c_i = \HH_R(c,\pk,i)$ and $\sk_i = \HH_L(c,\pk,i) + \sk \mod N$ where $N = |\secp|$.
    \item Set $\pk_i = \pkec(\sk_i)$.
\end{itemize}

In $i$th \emph{hardened} child $\HH^h(\xsk)$ is defined similarly except $\sk$ is used instead of $\pk$ in the third step.

Differently stated, we define $\HH^h(c,x,i) = (\HH_L(c,x,i) + x, \HH_R(c,x,i))$ and $\HH^{nh}(c,x,i) = (\HH_L(c,PK^{EC}(x),i)+x,\HH_R(c,PK^{EC}(x),i))$ and get that $\xsk_i = \HH^h(c,x,i)$ in the hardened case and $\xsk_i = \HH^{nh}(c,x,i)$ in the non-hardened case. 

We call a pair $s=(i,p)$ where $p\in\{h,nh\}$ a \emph{derivation step}. A \emph{derivation path} $P$ is a sequence of derivation steps. We recursively define $\xsk_\emptyset = \msk$ and $\xsk_{P,(i,p)} = \HH^{p}(\xsk_P,i)$. For each derivation step $s=(i,p)$ we define $\HH^s(x) = \HH^p(x,i)$. Finally, for each derivation path $P$ we define $\HH^P$ recursively as $\HH^\emptyset(x) = x$ and $\HH^{P,s}(x) = \HH^{s}(\HH^P(x))$.

Note that given $\xpk = (c,\pk)$ one can compute $\pk_i^{nh} = \pkec(\HH_L(c,\pk,i)+\sk) = \pkec(\HH_L(c,\pk,i)) + \pkec(\sk) = \pkec(\HH_L(c,\pk,i)) + \pk$ (where we used the fact that $\pkec$ is a group homomorphism). Also note that given $\sk_i^{nh}$ and $\xpk = (c,\pk)$ one can compute $\sk = \sk_i - \HH_L(c,\pk,i)$, so it is not secure to provide an extended public-key along with a non-hardened non-extended child secret-key. 

We will need the following:

\begin{proposition}\label{prop:derstap}
    If $\HH$ is a random oracle and $s$ is a derivation step then $\HH^{s}$ is a random oracle.
\end{proposition}

\begin{proof}
    Let $s = (i,p)$. By assumption, $\HH$ is random, which implies that $\HH_L$ is random. Since $\pkec$ is injective, it follows that $(c,x,i)\mapsto \HH_L(c,\pkec(x),i) = \HH^{nh}_R$ is also random. Finally, we note that if $f$ is uniformly random and $g$ is fixed then $f+g$ is uniformly random, from which it follows that $\HH^{nh}_R$ is also random. It follows that $\HH^{nh}$ is random as needed. The proof for $\HH^h$ is similar.

    We now note that $\HH^{(i,p)}$ is simply $\HH^p(\cdot,i)$, but the result of fixing a part of the input of a random function is also a random function (with a smaller domain), which completes the proof.
\end{proof}

From this follows by induction:

\begin{corollary}\label{prop:derpath}
    If $\HH$ is a random oracle and $P$ is a non-empty derivation path then $\HH^{P}$ is collision-resistant.
\end{corollary}

Finally we define the function $\der(\msk,P) = \HH^{P}(\msk)$. The function $\der$ is not collision-resistant: if $P_1$ and $P_2$ are two non-empty derivation paths, and $P=P_1\|P_2$ is their concatenation, then it holds for any $\msk$ that $\der(\msk,P) = \der(\der(\msk,P_1),P_2)$. However, these are the only forms of collision a bounded adversary could feasibly produce. This property will be useful to us in \ifnum\masterthesis=0\cref{ssec:seedlift}\else\cref{ssec:sl_seed_proof}\fi, so we define and prove it formally.

\begin{definition}\label{defn:nscrh}
    Let $f(x,y)$ be a two-variable function. A point $(x',y')$ is \emph{an $f$-suffix} of $(x,y)$ if there exists $y''$ such that $y=y''\|y'$ and $x'=f(x,y'')$.

    A pair of inputs $(x,y),(x',y')$ is a \emph{non-suffix collision} if $f(x,y)=f(x',y')$ but neither of the input is an $f$-suffix of the other.

    The function $f$ is \emph{collision resistant up to suffixes} if it is infeasible to find a non-suffix collision.
\end{definition}

\begin{proposition}\label{prop:dernscrh}
    If $\HH$ is a random oracle then $\der$ is collision resistant up to suffixes.
\end{proposition}

\begin{proof}
    Let $\sA$ be a \qpt adversary which outputs a non-suffix collision  $(\msk,P),(\msk',P')$ with probability $\varepsilon$. We prove that we can recover from $(\msk,P),(\msk',P')$ a collision in a function known to be a random-oracle, and in particular collision-resistant. From this will follow that $\varepsilon = \negl$, completing the proof.

    In most cases, we will find a collision in $\HH^Q$ for some non-empty path $Q$, which is a random oracle by \cref{prop:derpath}. In the remaining case we will find a collision in $\HH_L$, which is a random oracle by the hypothesis that $\HH$ is a random oracle.

    If $P=P'$ we get that $\msk\ne \msk'$ is a collision for $\HH^{P}$. Assume that $P\ne P'$.

    If there is a $P''$ such that $P = P''\|P'$, then by the assumption that $(\msk,P),(\msk',P')$ is non-suffix we have that $\msk' \ne \HH^{P''}(\msk)$. However it does hold that $\HH^{P'}(\msk') = \HH^{P'}(\HH^{P''}(\msk))$, so we found a collision in $\HH^{P'}$.

    The remaining case is that neither $P,P'$ is a suffix of the other. Let $S$ be the longest shared suffix of $P,P'$. Let $Q,Q'$ be paths such that, $P=Q\|S$ and $P'=Q'\|S$ (note that both $Q$ and $Q'$ are necessarily non-empty). If $\HH^Q(\msk) \ne \HH^{Q'}(\msk')$ then these points constitute a collision of $\HH^S$. Otherwise, let $s$ and $s'$ be the last steps in $Q$ and $Q'$ respectively, and let $\tilde{Q}$ and $\tilde{Q'}$ be $Q$ and $Q'$ with the last step removed. By the maximality of $S$ we have that $s\ne s'$. However, we have that $\HH^s(\HH^{\tilde{Q}}(\msk)) = \HH^Q(\msk) = \HH^{Q'}(\msk') = \HH^{s'}(\HH^{\tilde{Q'}}(\msk'))$. Let $s=(i,p)$ and $s'=(i',p')$. If $p=p'$ we get from $s\ne s'$ that $i\ne i'$, and since $\HH^p(\HH^{\tilde{Q}}(\msk),i)=\HH^p(\HH^{\tilde{Q'}}(\msk'),i')$ we found a collision in $\HH^p$. If $p\ne p'$ assume without loss that $p = h$ and $p' = nh$, and also set $\HH^Q(\msk) = \HH^{Q'}(\msk') = (\sk,c)$, then we have that $H_R(c,\sk,i) = H_R(c,\pkec(\sk),i')$. However, recall that in \cite{Wui13} secret-keys are encoded as 32 byte strings whereas public-keys are encoded as 33 byte strings (see \ifnum\masterthesis=0\cref{sssec:ec}\else\cref{ssec:prelim_ds_ec}\fi), whereby $(c,\sk,i)\ne (c,\pkec(\sk),i')$ so we found a collision in $\HH_R$.
\end{proof}

So far we have considered $\HH$ to have fixed output length of $512$ bits. The next property we require is stated more naturally when we think of $\HH_\ell$ as having output length $2\ell$ and of $\pkec_\ell$ as an arbitrary efficiently computable injection from $\zo^\ell$ to $\zo^{\ell + c}$ for some $c>0$ (when instantiated with $\HH = \shaf$ and $\pkec$ specified in BIP-39 we have that $\ell = 256$ and $c=8$). For a string $x\in \zo^{2\ell}$ we use $x_L$ and $x_R$ to denote its first and last $\ell$ bits respectively.

\begin{proposition}\label{prop:preimage}
    Fix a path $Q$. Let $\sA$ be a \qpt adversary with oracle access to $\HH_\ell$ that gets as input $x\in\zo^{2\ell}$ and outputs a path $Q$ and a string $x'$ of length at least $\ell$. If $\HH_\ell$ is a random oracle, then $$\PP[\HH_\ell(x)=\HH_\ell^Q(x')\mid x\gets \zo^{2\ell}, (x',Q)\gets \sA(x)]=\mathsf{negl}(\ell)\text{.}$$
\end{proposition}

\begin{proof}
    We show that this is infeasible for $Q$ which only contains a single hardened derivation step. The proof is identical for a non-hardened derivation step and easily extends to a general $Q$ by induction, which is omitted.
    
    For the sake of readability let $\HH$ denote $\HH_\ell$. Suppose an adversary finds $i$ such that $\HH(x) = \HH^{(i,h)}(x')$. Then in particular $\HH_R(x) = \HH_R(x',i)$. Since $\HH_R$ is a random oracle and is thus collision-resistant, it follows that with overwhelming probability $x=(x',i)$. Thus, to have an equality, we must have that $\HH_L(x) = \HH_L(x) + x_L$, implying that $x_L=0$. However, since $x$ is uniformly random, this only happens with probability $2^{-\ell}=\mathsf{negl}(\ell)$.
\end{proof}





\section{Seed Lifting Security Proof}

We provide a full proof in \cref{thm:seedlift}. Loosely speaking, \cref{prop:dernscrh} implies that we need only consider three types of attacks:
\begin{itemize}
    \item The adversary can create their own seed and master secret-key, and find a derivation path from their master secret-key to an existing public-key. We show that such an attack implies a collision in the hash function.
    \item The adversary manages to \emph{extend} the derivation path. This is similar to the previous attack, except the derivation path in the signature is a \emph{suffix} of the derivation path produced by the adversary. The analysis of this attack is similar to the previous one.
    \item In the third type of attack, the adversary \emph{truncates} the given derivation path. That is, the derivation path they provide is a \emph{suffix} of the derivation path in the signature. Given a derivation path and a master secret-key, it is trivial to create a master secret-key for a suffix (by only applying a part of the path), so the argument used in the previous forms of attack does not carry over. In order to prove this attack infeasible, we first prove in \cref{ssec:preimage} that given an adversary that can forge a \picnic signature, it is possible to extract a preimage of the master secret-key with respect to the key-derivation function. We use this property to show that in this scenario it is also possible to recover a collision in the hash function.
\end{itemize}

\subsection{Preimage Extractability from \textsf{Picnic} Signatures}\label{ssec:preimage}

In the next section, we introduce \emph{seed-lifted} schemes. As we explain therein, the seed-lifted scheme is constructed such that an adversary with access to valid signatures with respect to some keys can't feasibly create a signature that is valid with respect to a \emph{different} (but related) public-key. This type of security can not be implied by the \seufcma security of $\picnic$ alone, as such security only prohibits creating signed documents verifiable by \emph{the same} public-key.
In order to overcome this, we want to exploit the relations between the two public-keys -- the one given to the adversary, and the one with respect to which their output passes verification -- to produce a collision in a function known to be collision resistant. However, to do that, we need to compute the \emph{secret}-key corresponding to the public-key produced by the adversary.

Roughly speaking, the property we need is that if it is feasible to produce a signature verifiable by some public-key, then it is also feasible to produce the corresponding secret-key.

To see this holds for the scheme $\picnic(f)$, we note that in this scheme a signature is an accepting transcript of a slightly modified version of the Unruh transform \cite{Unr15}, applied to a particular $\Sigma$-protocol used as an argument of knowledge for a preimage of $f$ used to instantiate the scheme\footnote{For an overview of $\Sigma$-protocols and how they are used to construct signature schemes, we refer the reader to \cite[Chapter~8]{Kat10}}. The purpose of the modification is twofold: to incorporate the message $m$ into the transcript, and to slightly generalize Unruh's transform, whose original formulation does not apply to the particular $\Sigma$-protocol used in \cite{CDG+17}.

In particular, part of the proof of \cite[Corollary~5.1]{CDG+17} extends \cite[Theorem~18]{Unr15}, proving that for any post-quantum one-way function $f$, a valid $\picnic(f)$ signature (on any message) is an \emph{argument of knowledge} whose error is negligible in the length of the input.

We summarize this discussion in the following:

\begin{proposition}[\picnic extractability]\label{prop:zkpicnic}
    Let $f$ be a post-quantum one-way function with input length $\secpar$. Assume there exists a \qpt adversary $\sA$ that, on a uniformly random input $\pk\gets\zo^\secpar$, outputs with probability $\varepsilon$ a signed document $(m,\sigma)$ such that $\picnic(f).\ver_{\pk}(m,\sigma)$ accepts. Then there exists a \qpt \emph{extractor} $\sE$ that, on a uniformly random input $\pk\gets\zo^\secpar$, outputs $\sk$ such that $f(\sk)=\pk$ with probability $\varepsilon - \negl$.
\end{proposition}
\subsection{Security Proof}


\ifnum\masterthesis=0{In this section we prove }\else{We now turn to prove }\fi\cref{thm:seedlift}. We do so in two steps: we show that the seed-lifting is a post-quantum lifting, and then infer from that that it is actually a strong lifting. The first step requires most of the effort, so we start with the short and straightforward second step.

\newcommand{\chng}[2]{\sout{#1} \colorbox{Gray!40!white}{\textcolor{Red}{#2}}}

\begin{proposition}
    If the seed-lifting is \seufcma secure, then it is \sseuflcma secure.
\end{proposition}

\begin{proof}
        The idea is that all signatures provided by the seed-lifted schemes contain information that could be used to perfectly simulate the signature oracle of the modified scheme.
    
        Let $\sA$ win the \sseuflcma game with probability $\varepsilon$, and consider the adversary $\sA'$ for the \seufcma game which operates as follows:

        \begin{itemize}
            \item $\sA'$ initiates the simulation of $\sA$ by providing her the $\pk$ she obtained from the challenger.
            \item $\sA'$ queries the signature oracle (recall that since $\sA'$ is an \seufcma adversary, she only has access to the signature oracle of the lifting) on a uniformly random $m_0\gets \zo^\secpar$ to obtain $(\sigma,\msk,P)$.
            \item $\sA'$ resumes simulation of $\sA$, using her own oracle access to answer queries on $\widetilde{\sign}_\sk$, and responding to queries of the form $\sign_\sk(m)$ with $\sigma \gets \ecdsa.\sign_{\der(\msk,P)}(m)$.
            \item Outputs the output of $\sA$
        \end{itemize}

        Note that the view of $\sA$ in the real and simulated games is identical.
    
        By hypothesis, $\sA$ outputs $(\sigma,m)$ such that $\ver_{\pk}(m,\sigma)$ accepts and $\sA$ never made a query on $m$ is $\varepsilon$. In this event, $\sA'$ wins the game unless $m=m_0$. However, since $m_0$ was chosen uniformly, and the view of $\sA$ is independent of $m_0$, it follows that $m\ne m_0$ with overwhelming probability. Hence $\sA_p$ wins the \seufcma game with probability $\varepsilon - \negl$. But by hypothesis the seed lifting is post-quantum, so it follows that $\varepsilon = \negl$, whereby the lifting is strong.
\end{proof}

The first step is the following statement:

\begin{proposition}
    If $\HH$ is modeled as a random oracle, the seed-lifting is \seufcma secure.
\end{proposition}

\begin{proof}
We prove this by a sequence of hybrids.

The first hybrid $H_0(\sA,\secparam)$ is just a restatement of the $\seufcma$ game, explicated for the seed-lifting:

\begin{enumerate}
    \item $\sC$ samples $(\sk,\pk)\gets \keygen(\secparam)$ where $\sk=(s,pw,P)$
    \item Set $\sk_p = \kdf^{pre}(s,pw)$
    \item Set $\msk = \kdf^{pq}(\sk_p)$, note that by definition $\pk = \der(\msk,P)$
    \item $\sC$ send $\pk$ to $\sA$
    \item Let $q$ be the number of queries made by $\sA$, for $i=1,\ldots,q$:
    \begin{enumerate}
        \item $\sA$ sends $m_i$ to $\sC$
        \item $\sC$ computes $\sigma \gets \picnic(\kdf^{pq}).\sign_{\sk_p}(m,P)$
        \item $\sC$ sends $(\sigma,\msk,P)$ to $\sA$
    \end{enumerate}
    \item $\sA$ sends $(m,(\sigma,\msk',P'))$ to $\sC$
    \item if $\exists i : m_i = m$, output 0
    \item if $\der(\msk',P')\ne \pk$, output 0
    \item output $1$ if $\picnic(\kdf^{pq}).\ver_{\msk'}((m,P'),\sigma)$ accepts, $0$ otherwise.
\end{enumerate}

In the next hybrid $H_1(\sA,\secparam)$ we choose $\sk_p$ uniformly at random:

\begin{enumerate}
    \setcounter{enumi}{1}
    \item \chng{Set $\sk_p = \kdf^{pre}(s,pw)$\\}{$\sC$ samples $\sk_p \gets \zo^\ell$} (where $\ell$ is the secret-key length) 
\end{enumerate}

\begin{claim}
    For any \qpt $\sA$ it holds that 
    $|\PP[H_0(\sA,\secparam)=1] - \PP[H_1(\sA,\secparam)=1]| = \negl$.
\end{claim}

\begin{proof}
    Consider the distributions $\DD_0$ and $\DD_1$ where $\DD_0$ is uniform on $\zo^\ell$ and $\DD_1$ is sampled by sampling $(\sk=(s,pw,p),\pk)\gets \keygen(\secparam)$ and outputting $\kdf^{pre}(s,pw)$.

    Recall that $\kdf^{pre} = \HH^{2047}$. That $\HH$ is uniformly random does not imply that $\kdf^{pre}$ is uniformly random. However, \cite{BDD+17} prove that the distribution of iterating a random function (with an exponentially large domain) a constant number of times is computationally indistinguishable from uniformly sampling a function. Combined with the fact that $\gamma_{(s,pw)}=\negl$ (recall \cref{defn:guessing}) we get that if from \cref{prop:kdf} that $|\PP[\sA_D^{\DD_0}(\secparam)=1] - \PP[\sA_D^{\DD_0}(\secparam)=1| = \negl$, where $\sA_D^\DD$ is any \qpt procedure that is given oracle access to a distribution $\DD$.

    Now consider the following procedure $\sA_D^{\DD}(\secparam)$:
    \begin{itemize}
        \item $\sA_D$ samples $\sk_p \gets \DD$
        \item $\sA_D$ samples $(\sk,\pk)\gets\keygen(\secparam)$, she discards all data but $P$
        \item $\sA_D$ computes $\msk = \kdf^{pq}(\sk_p)$ and $\pk = \der(\msk,P)$
        \item $\sA_D$ simulates the $\seufcma$ game with $\sA$ by exactly emulating $\sC$, and outputs the output of $\sC$
    \end{itemize}

    Note that the view of $\sA$ in the simulation carried by $\sA_D^{\DD_b}(\secparam)$ is identical to the view in the hybrid $H_b(\sA,\secparam)$. Hence $\PP[H_b(\sA,\secparam)=1]=\PP[\sA_D^{\DD_b}(\secparam)=1]$ and it follows that $|\PP[H_0(\sA,\secparam)=1] - \PP[H_1(\sA,\secparam)=1]| = \negl$.
\end{proof}

In the hybrid $H_2$ we add a new failure condition that trivially never holds, it will be useful for the following hybrids:

\begin{enumerate}
     \setcounter{enumi}{6}
    \item ...
    \item if $\der(\msk',P')\ne \pk$, output 0
    \item \chng{}{Set $E = \mathsf{True}$}
    \item \chng{}{if $E = \mathsf{False}$, output $0$}
    \item output $1$ if $\picnic(\kdf^{pq}).\ver_{\msk'}((m,P'),\sigma)$ accepts, $0$ otherwise.
\end{enumerate}

Since $E$ is always true, it follows that $H_1$ and $H_2$ execute identically, proving:

\begin{claim}
    For any \qpt $\sA$ it holds that $\PP[H_1(\sA,\secparam)=1] = \PP[H_2(\sA,\secparam)=1]$.
\end{claim}

In the hybrid $H_3$ we decompose $E$ into a conjunction of several events $E_1,\ldots,E_4$:

\begin{enumerate}
     \setcounter{enumi}{8}
    \item \chng{Set $E = \mathsf{True}$\\}{Set $E = E_1 \vee E_2 \vee E_3 \vee E_4$ where}:
    \begin{enumerate}
        \item \chng{}{$E_1 = \mathsf{True}$ if $(\msk,P)\ne (\msk',P')$ and neither is a $\der$-suffix of the other}
        \item \chng{}{$E_2 = \mathsf{True}$ if $(\msk,P)$ is a $\der$-suffix of $(\msk',P')$}
        \item \chng{}{$E_3  = \mathsf{True}$ if $(\msk',P')$ is a $\der$-suffix of $(\msk,P)$}
        \item \chng{}{$E_4  = \mathsf{True}$ if $(\msk,P) = (\msk',P')$}
    \end{enumerate}
    \item if $E = \mathsf{False}$, output $0$
\end{enumerate}

\begin{claim}\label{test}
    For any \qpt $\sA$ it holds that $\PP[H_2(\sA,\secparam)=1] = \PP[H_3(\sA,\secparam)=1]$.
\end{claim}

\begin{proof}
    It suffices to show that in $H_3$ it holds with certainty that $E=\mathsf{True}$.

    Consider the strings $(\msk,P)$, $(\msk',P')$. If they are equal, then $E_4 = \textsf{True}$. If $\der(\msk,P) \ne \der(\msk',P')$ then neither can be a suffix of the other (since a suffix is a type of collision), so $E_1 = \textsf{True}$. If $\der(\msk,P) = \der(\msk',P')$ then they constitute a collision in $\der$. If it is a suffix collision, then either $E_2 = \textsf{True}$ or $E_3 = \textsf{True}$ (depending on which of the strings is a suffix of the other). If it is not a suffix collision, then $E_1 = \textsf{True}$.
\end{proof}

In the hybrid $H_4$ we set $E_1$ to false:

\begin{enumerate}
     \setcounter{enumi}{8}
    \item Set $E = E_1 \vee E_2 \vee E_3 \vee E_4$ where:
    \begin{enumerate}
        \item \chng{$E_1 = \mathsf{True}$ if $(\msk,P)\ne (\msk',P')$ and neither is a $\der$-suffix of the other\\} {$E_1 = \mathsf{False}$}
        \item $E_2 = \mathsf{True}$ if $(\msk,P)$ is a $\der$-suffix of $(\msk',P')$
        \item $E_3  = \mathsf{True}$ if $(\msk',P')$ is a $\der$-suffix of $(\msk,P)$
        \item $E_4  = \mathsf{True}$ if $(\msk,P) = (\msk',P')$
    \end{enumerate}
    \item if $E = \mathsf{False}$, output $0$
\end{enumerate}

\begin{claim}
    For any \qpt $\sA$ it holds that $|\PP[H_3(\sA,\secparam)=1] - \PP[H_4(\sA,\secparam)=1]| = \negl$.
\end{claim}

\begin{proof}

    We note that $H_3$ conditioned on $E_1 = \textsf{False}$ is identical to $H_4$. Hence, $|\PP[H_3(\sA,\secparam)=1] - \PP[H_4(\sA,\secparam)=1]|$ is at most the probability that $E_1 = \textsf{True}$ in $H_3$.
    
    Consider a \qpt adversary $\sA_c$ that perfectly simulates $H_3(\sA,\secparam)$ until receiving an output. $\sA_c$ then outputs $((\msk,P),(\msk',P'))$ where $(\msk',P')$ was taken from $\sA$'s output. The event that $((\msk,P),(\msk',P'))$ is a non-suffix collision in $\der$ is exactly the event that $E_1 = \textsf{True}$. \cref{prop:dernscrh} asserts this event has negligible probability.
\end{proof}

In the hybrid $H_5$ we set $E_2$ to false:

\begin{enumerate}
     \setcounter{enumi}{8}
    \item Set $E = E_1 \vee E_2 \vee E_3 \vee E_4$ where:
    \begin{enumerate}
        \item $E_1 = \mathsf{False}$
        \item \chng{$E_2 = \mathsf{True}$ if $(\msk,P)$ is a $\der$-suffix of $(\msk',P')$ \\ } {$E_2 = \mathsf{False}$}
        \item $E_3  = \mathsf{True}$ if $(\msk',P')$ is a $\der$-suffix of $(\msk,P)$
        \item $E_4  = \mathsf{True}$ if $(\msk,P) = (\msk',P')$
    \end{enumerate}
    \item if $E = \mathsf{False}$, output $0$
\end{enumerate}

\begin{claim}
    For any \qpt $\sA$ it holds that $|\PP[H_4(\sA,\secparam)=1] - \PP[H_5(\sA,\secparam)=1]| = \negl$.
\end{claim}

\begin{proof}
    Like in the previous claim, we have that $|\PP[H_4(\sA,\secparam)=1] - \PP[H_5(\sA,\secparam)=1]|$ is bound by the probability that $E_2 = \mathsf{True}$ in $H_4$.
    
    Consider the following adversary $\sA_s$ that given a uniformly random $\sk_p$ does the following:
    \begin{itemize}
        \item calculates $\msk=\kdf^{pq}(\sk_p)$
        \item uses $\keygen$ to sample a path $P$
        \item calculates $\pk = \der(\msk,P)$
        \item simulates $\sA$ using $\pk$ as input and $\sk_p,P$ to answer oracle calls
        \item recovers $(\msk',P')$ from the output of $\sA$
        \item if $(\msk,P)$ is a $\der$-suffix of $(\msk',P')$, let $Q$ such be a path such that $P'=Q\|P$, and output $(\msk',Q)$, otherwise output $\bot$.
    \end{itemize}

    Note that the event that $E_2=\mathsf{True}$ in $H_4$ is exactly the event that $\sA_s$ didn't output $\bot$. In this case, the output $(\msk',Q)$ if $\sA$ satisfies that $\der(\msk',Q) = \kdf^{pq}(\sk_p)$. Recalling that $\sk_p$ is uniformly random, $\kdf^{pq} = \HH$ and $\der(\msk',Q) = \HH^Q(\msk')$, it follows from \cref{prop:preimage} that this could only happen with negligible probability.
\end{proof}

In the hybrid $H_6$ we set $E_3$ to false:

\begin{enumerate}
     \setcounter{enumi}{8}
    \item Set $E = E_1 \vee E_2 \vee E_3 \vee E_4$ where:
    \begin{enumerate}
        \item $E_1 = \mathsf{False}$
        \item $E_2 = \mathsf{False}$
        \item \chng{$E_3  = \mathsf{True}$ if $(\msk',P')$ is a $\der$-suffix of $(\msk,P)$ \\}{$E_3  = \mathsf{False}$}
        \item $E_4  = \mathsf{True}$ if $(\msk,P) = (\msk',P')$
    \end{enumerate}
    \item if $E = \mathsf{False}$, output $0$
\end{enumerate}

\begin{claim}
    For any \qpt $\sA$ it holds that $|\PP[H_5(\sA,\secparam)=1] - \PP[H_6(\sA,\secparam)=1]| = \negl$.
\end{claim}

\begin{proof}
    Keeping the same line of reasoning, it suffices to show that the event that $E_3 = \textsf{True}$ in $H_5$ has negligible probability.

    The event $E_3$ is the event that $\sA$ managed to produce an output $((\sigma,\msk',P'),m)$ with the following properties:
    \begin{itemize}
        \item $(\msk',P')$ is a $\der$-suffix of $(\msk,P)$. That is, there is some non-empty path $Q$ such that $P=Q\|P'$ and $\msk'=\der(\msk,Q)$.
        \item $\picnic(\kdf^{pq}).\ver_{\msk'}((m,P'),\sigma)$ accepts.
        \item $\sA$ has never made a signature query on $m$.
        \item $\der(\msk',P') = \pk$.
    \end{itemize}

    Say this happens with probability $\varepsilon$. Since $P$ is sampled independently of $\sk_p$, there has to be a particular path $\tilde{P}$ such that the probability of $E$ conditioned on $P=\tilde{P}$ is at least $\varepsilon$.

    Let $q+1$ be the length of $P$ (recall that $P$ must have length at least two for the event $E_3$ to hold), for any $t=1,\ldots,q$ let $P_t$ be the suffix of $\tilde{P}$ of length $t$. Then there must exist some $\tilde{t}$ such that $P' = P_{\tilde{t}}$ with probability at least $\varepsilon/q$.

    Let $\sA_z$ be the procedure that on input $\sk_p$ does the following:
    \begin{itemize}
        \item Calculates $\msk = \HH(\sk_p)$ and $\pk = \pkec(\der(\msk,\tilde{P})_L)$ and gives $\pk$ as input to $\sA$
        \item Simulates the hybrid $H_5$ responding to signature queries with $(\sigma,\msk,\tilde{P})$ where $\sigma \gets \picnic(\HH).\sign_{\sk_p}(m,\tilde{P})$
        \item Resumes operations until obtaining output $((\sigma,\msk',P'),m)$.
        \item Outputs $\sigma$ if all the following hold: $P'=\tilde{P}_{\tilde{t}}$, $\msk' = \der(\msk,Q)$ where $\tilde{P}=Q\|P'$, $m$ was never the input to a query and $\picnic(\HH).\ver_{\msk'}((m,P),\sigma)$ accepts. Otherwise, output $\bot$.
    \end{itemize}

    By design. If $\sk_p$ is uniformly random, then the probability that $\sA_z$ doesn't output $\bot$ is at least $\varepsilon/q$. In this case, $\sA_z$ managed to output a signature that passes the $\picnic(\HH)$ verification with respect to the public-key $\msk'$.

    It follows from \cref{prop:zkpicnic} that there exists an extractor $\sE$ which outputs an $\HH$-preimage $x$ of $\msk'$ with probability $\varepsilon/q - \negl$. We thus have that $\HH(x) = \HH^Q(\msk)$. By \cref{prop:preimage} this can only happen with negligible probability. We get that $\varepsilon/q-\negl$ is negligible and so $\varepsilon$ is also negligible.
\end{proof}

In the last hybrid $H_7$ we set $E_4 = \mathsf{False}$:

\begin{enumerate}
     \setcounter{enumi}{8}
    \item Set $E = E_1 \vee E_2 \vee E_3 \vee E_4$ where:
    \begin{enumerate}
        \item $E_1 = \mathsf{False}$
        \item $E_2 = \mathsf{False}$
        \item $E_3  = \mathsf{False}$
        \item \chng{$E_4  = \mathsf{True}$ if $(\msk,P) = (\msk',P')$ \\}{$E_4  = \mathsf{False}$}
    \end{enumerate}
    \item if $E = \mathsf{False}$, output $0$
\end{enumerate}

\begin{claim}
    For any \qpt $\sA$ it holds that $|\PP[H_5(\sA,\secparam)=1] - \PP[H_6(\sA,\secparam)=1]| = \negl$.
\end{claim}

\begin{proof}
    Keeping the same line of reasoning, it suffices to show that the event that $E_4 = \textsf{True}$ in $H_6$ has negligible probability.

    Let $\varepsilon$ be the probability that $E_4 = \mathsf{True}$ in $H_6$, let $\sA_p$ be the following \seufcma adversary for $\picnic(\kdf^{pq})$:
    
    \begin{itemize}
        \item Let $(\sk_p,\pk_p)$ denote the keys generated by $\sC$.
        \item After receiving $\pk_p$, $\sA_p$ samples $((s,pw,P),\pk)\gets \keygen(\secparam)$ and discards all data but $P$ (recall that $P$ is sampled independently of $(s,pw)$.
        \item $\sA_p$ uses $\pk = \der(\pk_p,P)$ as input to $\sA$.
        \item $\sA_p$ responds to a signature query on $m$ by querying the $\picnic(\kdf^{pq}).\sign_{\sk_p}$ oracle on $(m,P)$ to obtain $\sigma$ and outputting $(\sigma,\pk_p,P)$.
        \item $\sA_p$ resumes the simulation until obtaining an output $((\sigma,\msk',P'),m)$. If $(\pk_p,P)=(\msk',P')$, $m$ was never queried and $\picnic(\kdf^{pq}).\ver_{\pk_p}(m,\sigma)$ accepts, $\sA_p$ outputs $m$. Otherwise $\sA_p$ outputs $\bot$.
    \end{itemize}
    
    Note that the view of $\sA$ in the simulation and in $H_5$ is identical. Also note that $\sA_p$ either wins the \seufcma game or outputs $\bot$, and the former happens exactly when $\sA$ outputs a winning output. That is, the probability that $\sA_p$ wins the game is exactly $\varepsilon$.

    Now note that, since $\kdf^{pq} = \HH$ is modeled as a random oracle, it follows from \cite[Theorem~4]{CLQ20} that it is post-quantum one-way. It then follows from \cite[Corollar~5.1]{CDG+17} that $\picnic(\kdf^{pq})$ is \seufcma secure. Since $\sA_p$ is a \qpt adversary that wins the \seufcma game for $\picnic(\kdf^{pq})$ with probability $\varepsilon$, it follows that $\varepsilon = \negl$.
        
\end{proof}

By stringing all claims above together and applying the triangle inequality, we obtain that it holds for any \qpt adversary $\sA$ that $|\PP[H_0(\sA,\secparam) = 1] - \PP[H_7(\sA,\secparam) = 1] | =  \negl$.

We conclude the proof by noting that in $H_7$ it always holds that $E=\textsf{False}$, whereby, for any \qpt adversary $\sA$ we have that $\PP[H_7(\sA,\lambda)=1]=0$. It follows that $\PP[H_0(\sA,\secparam) = 1] =  \negl$ as needed.

\end{proof}


\section{Quantum Canaries}\label{sec:quantum_canaries}


In this section, we introduce \emph{quantum canaries}: puzzles designed to be intractable for classical computers but solvable for quantum computers whose scale is significantly smaller than required to compromise \ecdsa signatures. A solution to the puzzle will then act as a heads-up for the network that the quantum era is near. To incentivize quantum entities to solve the puzzle, thus alerting the network to the presence and compromising the quantum loot, we propose awarding a monetary prize to the first solver. In \cref{ssec:canarygame}, we provide a game theoretic analysis of two entities competing for the canary bounty and quantum loot.

Once the canaries are set up, they could be used to implement policies in consensus, such as: any non quantum-cautious transaction posted more than 10,000 blocks after the canary puzzle was solved is considered invalid.

The idea of using cryptographic canaries to award bounties for discovering exploits was first discussed in \cite{Dra18}. We discuss in more depth the specific application of canaries to detect quantum entities. In particular, we analyze the behavior of two entities competing for the quantum loot and discuss the option to fund the bounty for the canary based on \utxos that would be burned.

\subsection{Properties of Good Canaries}

We expect a quantum canary to satisfy the following properties:

\begin{itemize}
    \item \textbf{Similarity}\quad The task of killing the canary should be as similar as possible to the task of forging an \ecdsa signature for a selected message given access to the public-key. If the challenge we choose for killing the canary is vastly different from breaking \ecdsa, it becomes more plausible that a future optimization could reduce the scale required for one task over the other, leading to a scenario where killing the canary becomes nearly as hard or even harder than breaking \ecdsa, allowing a future attacker to claim the bounty \emph{and} the loot.
    There are several incomparable complexity metrics for "hardness" in quantum computing\ifnum\masterthesis=1{ (see \cref{sec:prelim_qc} for an overview)}\fi, the most studied being the circuit size, and circuit depth (directly related to the time complexity), T-count, and T-depth (which are the important measure when dealing with certain quantum error-correcting codes, in which Clifford gates are cheap to implement) and the number of qubits (space complexity).
    Different problems offer different trade-offs, and since it is impossible to predict which quantum resource will develop more quickly, it is desirable that killing the canary will have trade-offs similar to breaking \ecdsa.
    
    \item \textbf{Incentive}\quad 
    There should be a clear incentive for quantum entities that \emph{can} kill the canary to actually do so, rather than hiding their quantum capabilities until they mature enough to loot pre-quantum coins.
    
    \item \textbf{Security}\quad The mechanism for posting the solution to the challenge should be secure against forking attacks. In particular, it should be infeasible for anyone listening to the mempool to claim the solution as their own.
    \item \textbf{Nothing Up My Sleeve}\quad The procedure for generating the challenge should be publicly known. Any randomness used for generating the challenge should be sampled from a verifiably random source. All steps in the generating procedure should be justified, and arbitrary choices (that could conceal backdoor solutions) should be avoided.
\end{itemize}

\subsection{Choosing the Puzzle}\label{ssec:canarypuzzle}

Since our ultimate concern is entities that are able to break \ecdsa signatures over the \secp curve, a natural contender for a puzzle is to forge an \ecdsa signature over a curve with less security\footnote{Another natural contender is to skip the signing mechanism and simply require the solver to solve a logarithmic equation over a suitable elliptic curve. However, it is hard to argue that the difficulty of this task scales the same as creating an \ecdsa signature. An \ecdsa signature contains many steps besides solving the logarithmic equation, so implementing the entire signing mechanism is better in terms of similarity.}.

A necessary (but not sufficient) condition for this solution to have an untrusted setup is that we are able to generate a public-key without learning anything about the matching secret-key. Fortunately, \ecdsa has the nice feature that the public-key is merely a random point on the elliptic curve. This leads to the following approach:
\begin{itemize}
    \item select an appropriate curve,
    \item sample a random point $\pk$ on the curve and a random string $r$,
    \item post $(\pk,r)$ to the blockchain,
    \item to kill the canary, post $\sigma$ which passes \ecdsa verification as a signature of $r$ with respect to the chosen curve. 
\end{itemize}

The missing components to fully specify the canary are the curve itself, and a method for generating a public-key which does not provide information about the corresponding secret-key.

It is desirable that the family of curves we choose from bears as much similarity as possible to \secp, or we risk the scenario previously described where an optimization is found to computing discrete logarithm on \secp but not on the challenge curve (or vice versa).

The \secp curve is a curve over a prime field whose order is also prime, given in Weierstrass form. Recent works \cite{RNSL17,HJN+20} provide concrete cryptanalysis for computing discrete logarithms in such elliptic groups as a function of the order of the prime field above which they are defined.

The discussion above suggests that a good choice for a curve family might be the family of all Weierstrass curves over prime fields of prime order whose binary representation has length $b$, for some appropriately chosen $b$.

However, simply choosing an arbitrary curve of convenient parameters and counting our blessings is not a sensible approach. Appropriately choosing a secure elliptic curve is a very difficult task, and a carelessly chosen curve could easily exhibit classically feasible exploits (for an overview of selecting elliptic curves, see e.g. \cite{BCL+16}). Unfortunately, our current application is unusual in the sense that we deliberately seek out curves whose bit security is \emph{intermediate}, in the sense that it is lower than currently used curves but still high enough to withstand classical attacks. Virtually all available literature focuses on curves chosen to be as \emph{strong} as possible.

Moreover, the \secp curve was chosen partly due to its convenient parameterization, which allows for optimizing arithmetic operations over the curve group. This has no effect on the asymptotic computational cost of any algorithm computing the discrete logarithm, but has a significant impact on the constants. It is plausible that computing a discrete logarithm over \secp might be more efficient than computing a discrete logarithm over a general curve with fewer bits of security due to these optimizations. The curve selection process must be aware of these nuances. It might be desirable to require that the canary curves also furnish a compact binary representation. We leave the task of choosing a concrete curve to further discussion and future work.

\subsection{Funding the Bounty}\label{ssec:canary_incent}

The bounty could be funded from several sources: it could be raised from the community, freshly minted for that purpose, or "borrowed" from future inflation (e.g. allocating 5\% of all future block rewards). Another funding source unique to our setting is the funds destined to burn. \ifnum\masterthesis=0{Recall that in our proposed solution (see \cref{intro:proposal})}\else{In the final specification of the protocol (see \cref{sec:procrast_full})}\fi, all \utxos whose addresses were posted to the blockchain prior to 2013 will become forever unspendable two months after the canary is killed. The bounty for killing the canary could be taken from these funds.

However, we do not recommend using the burned funds, as there is uncertainty regarding how much of these funds will remain by the time the quantum era starts: if all users will switch to post-quantum addresses during the period between the canary is killed and the quantum era starts, no bounty will be left. This uncertainty could deter entities from claiming the bounty, and encourage them to wait for the loot.

\subsection{Adaptation to Taproot}

So far we implicitly assumed that quantum loot is locked behind \ecdsa signatures. However, as the newly deployed Taproot update \cite{WNT20} gains adoption, more quantum loot is accumulated behind Schnorr signatures (see \ifnum\masterthesis=0{\cref{ssec:qthreats}}\else{\cref{ssec:intro_bitcoin_attacks}}\fi). Fortunately, the implementation of Schnorr signatures in Taproot is instantiated with the same \secp curve used to instantiate the \ecdsa scheme.

However, one might argue that Schnorr and \ecdsa are not similar enough, and by specializing the puzzle to \ecdsa we take the risk that a quantum entity will be able to loot \textsf{P2TR} \utxos before they would be able to collect the bounty.

Fortunately, in both the \ecdsa and Schnorr schemes, the public-key is a uniformly random point on \secp. Hence, a wide adoption of \secp-instantiated Schnorr signatures (e.g. via Taproot) could be addressed by modifying the puzzle such that a valid solution is a signature which passes \emph{either} \ecdsa or Schnorr verification (w.r.t. to the sampled public-key).

\subsection{Game Theoretic Analysis}\label{ssec:canarygame}

\newcommand{\fw}[1]{
\begin{minipage}{2.3ex}
  \centering$#1$
\end{minipage}
}

\newcommand{\fcolor}[1]{\textcolor{ForestGreen}{#1}}
\newcommand{\scolor}[1]{\textcolor{Bittersweet}{#1}}

\newcommand{\fstyle}[1]{\fcolor{\fw{#1}}}
\newcommand{\sstyle}[1]{\scolor{\fw{#1}}}

\newcommand{\fl}{\fstyle{\LEFTCIRCLE}}
\newcommand{\fm}{\fstyle{\CIRCLE}}
\newcommand{\fr}{\fstyle{\RIGHTCIRCLE}}
\newcommand{\sle}{\sstyle{\ll}}
\newcommand{\sm}{\sstyle{\texttt{\DOWNarrow}}}
\newcommand{\sr}{\sstyle{\gg}}
\newcommand{\spa}{\hspace{2.3ex}}
\newcommand{\FE}{\fstyle{E}}
\newcommand{\FL}{\fstyle{L}}
\newcommand{\SE}{\sstyle{E}}
\newcommand{\SL}{\sstyle{L}}

\newcommand{\fln}{\fcolor{\LEFTCIRCLE}\xspace}
\newcommand{\fmn}{\fcolor{\CIRCLE}\xspace}
\newcommand{\frn}{\fcolor{\RIGHTCIRCLE}\xspace}
\newcommand{\slen}{\scolor{$\ll$}\xspace}
\newcommand{\smn}{\scolor{\texttt{\DOWNarrow}}\xspace}
\newcommand{\srn}{\scolor{$\gg$}\xspace}

\newcommand{\spaces}[1]{
    \foreach \n in {1,...,#1}{\spa}
}

\ifnum\masterthesis=0{
In this section we argue that quantum-capable entities, including dishonest ones, are incentivised to claim the bounty, even at the price of relinquishing the loot. To do so, we consider an idealized setting in which exactly two quantum entities exist. Furthermore, for each entity, it is known (to both entities) exactly when they would be able to claim the loot and the bounty. 
}\else{
In this section, we consider two quantum-capable entities that are competing for Bitcoin profits. We assume the two entities do not collude, and that their only utility is maximizing their profit, with no regard to ethical considerations. We further assume that the entities are completely aware of the quantum capabilities of each other. We present a simple model capturing this scenario (in particular, the "perfect knowledge" is modeled by assuming that it is known for each entity exactly when both entities would be able to claim the loot and the bounty). We map out the possible scenarios in this setting and find that most of them are "good" scenarios in which quantum-capable entities are incentivized to relinquish the loot and claim the bounty. Furthermore, we that all scenarios converge to "good" scenarios as the value of the bounty increases and the waiting time between claiming the bounty and prohibiting quantum unsafe transaction decreases. The scenarios and the convergences between them are detailed in \cref{tab:timeline_and_payoffs}.
}\fi

We consider two strategies: the \emph{early strategy}, in which the entity claims the bounty as soon as possible, and the \emph{late strategy}, in which the entity claims the bounty as soon as possible \emph{without relinquishing the loot}. That is, if an entity can claim the bounty at time $t_b$ and the loot at time $t_\ell$, then the early strategy is to try to claim the bounty at time $t_b$; whereas the late strategy is to try to claim the bounty at time $\max\{t_b, t_\ell - w \}$ and the loot at time $t_\ell$, where $w$ is the quantum adjustment period between the time the canary is killed and the time the quantum loot is burned (Of course, whether one party's attempt is successful depends on the other party's strategy).

Note that if $t_b > t_\ell - w$ then both strategies coincide, and the entity becomes \emph{degenerate}.

In \cref{tab:timeline_and_payoffs} we analyze two non-degenerate entities with all their possible timelines. We use \fln and \frn (resp. \slen, \srn) to denote the points in time in which the first (resp. second) entity is able to claim the bounty and loot, respectively.
Note that by the design of the canary, \fln always happens \emph{before} \frn, since killing the canary is strictly easier than forging a standard \ecdsa signature, which is required for claiming (part of) the loot.

Since the pay-offs are symmetric, we assume without loss of generality that the \fcolor{first (faster)} player is faster, and can claim the bounty before the \scolor{second (slower)} player. We use \fmn (resp. \smn) to denote the point in time in which the entity would claim the bounty if they follow the late strategy. That is, the time difference between \fmn and \frn (resp. \smn and \srn) is always $w$.

\renewcommand{\arraystretch}{1.5}

\begin{table}[!htb]
    \centering
    
\begin{tabular}{l||l}
\toprule
\textbf{Timelines} & \textbf{Payoffs}\\
\midrule
$\begin{array}{l}
\tikzmark{TL1}\fl\sle\spaces{5}\sm\spaces{3}\sr\fm\spaces{3}\fr\\
\tikzmark{TL2}\fl\spaces{4}\sle\spaces{5}\sm\fm\spaces{2}\sr\fr
\end{array}
$
 & 
 $ \begin{array}{c|cc} & \SE & \SL \\
 \hline 
 \fw{\FE} & \underline{(\fcolor{b},\scolor{0})} &  \underline{(\fcolor{b},\scolor{0})} \\[3pt]
 \fw{\FL} & (\fcolor{0},\scolor{b}) & (\fcolor{0},\scolor{b+}\,\,\textcolor{red}{\ell}) 
 \end{array}$
\\
\midrule
$ \begin{array}{l}
\tikzmark{TL3}\fl\spaces{4}\sle\spaces{3}\fm\spa\sm\spa\fr\sr\\
\tikzmark{TL4}\fl\spaces{7}\sle\fm\spaces{3}\fr\sm\spaces{3}\sr
\end{array}
$

 & 
 $\begin{array}{c|cc} & \SE & \SL \\
 \hline 
 \fw{\FE} & \underline{(\fcolor{b},\scolor{0})} &  (\fcolor{b},\scolor{0}) \\[3pt]
 \fw{\FL} &  (\fcolor{0},\scolor{b})& (\fcolor{b+}\,\textcolor{red}{\ell},\scolor{0}) 
 \end{array}$
\\
\midrule

 $ \begin{array}{l}
\tikzmark{TL5}\fl\spaces{8}\fm\sle\spa\sm\fr\spaces{2}\sr\\
\tikzmark{TL6}\fl\spaces{8}\fm\spaces{2}\sle\fr\sm\spaces{3}\sr\\
\tikzmark{TL7}\fl\spaces{8}\fm\spaces{3}\fr\sle\spa\sm\spaces{3}\sr
\end{array}
$

 & 
 $ 
 \begin{array}{c|cc} & \SE & \SL \\
 \hline 
 \fw{\FE} & (\fcolor{b},\scolor{0}) &   (\fcolor{b},\scolor{0})\\[3pt]
 \fw{\FL} & \underline{(\fcolor{b+}\,\textcolor{red}{\ell},\scolor{0})} & \underline{(\fcolor{b+}\,\textcolor{red}{\ell},\scolor{0})} 
 \end{array}$
\\
\bottomrule

\end{tabular}

 \begin{tikzpicture}[remember picture, overlay]
   \draw[->,bend left] (TL2.east) to (TL1.east);
   \draw[->,bend right] (TL3.east) to (TL4.east);
   \draw[->,bend left] (TL5.east) to (TL3.north east);
   \draw[->,bend right] (TL5.east) to (TL6.east);
   \draw[->,bend left] (TL6.west) to (TL4.west);
   \draw[thick, dotted,->,bend right, transform canvas={xshift=.7em}] (TL7) to (TL6);
   \draw[dotted,->,bend right, transform canvas={xshift=.7em}] (TL6) to (TL4);
   \draw[thick, dotted,->,bend right, transform canvas={xshift=.7em}] (TL5) to (TL3);
 \end{tikzpicture}
    \label{tab:timeline_and_payoffs}

    \caption{Best viewed in color. An analysis of two non-degenerate dishonest quantum entities with respect to all possible timelines. The timelines are grouped into scenarios, where each scenario corresponds to a different payoff matrix. Each timeline advances in time from left to right. \fln (resp. \slen) and \frn (resp. \srn) represent when the \fcolor{faster} (resp. \scolor{slower}) entity is capable of claiming the bounty and the quantum loot respectively. \fmn (resp. \smn) represents the earliest point in which the \fcolor{faster} (resp. \scolor{slower}) entity can claim the bounty without forfeiting the loot, and is always at distance $w$ from \frn (resp. \srn). The payoff matrices list the outcome of the game for both players given their timeline and the strategies they chose. The headings of the matrix rows and columns represent the strategies chosen by both players, where "E" stands for the early strategy, and "L" stands for the late strategy. The rows (resp. columns) of the matrices represent the strategy chosen by the  \fcolor{faster} (resp. \scolor{slower}) entity. Each entry of the matrix specifies the outcome for the \fcolor{faster} and \scolor{slower} players, in that order. Pure Nash equilibria are denoted by an underline. Solid arrows describe how timelines transform as the waiting time decreases, and dotted arrows describe how timelines transform as the bounty increases. We find that in the first two scenarios, the equilibria strategies do not claim the loot.}
\end{table}

The matrices on the right column show the payoffs for each entity in each possible timeline given the strategy of both entities. The rows of the matrices represent the strategy of the \fcolor{fast} entity, whereas the columns represent the strategy of the \scolor{slow} entity. The headings "E" and "L" denote the early and late strategies, respectively. The utility of the bounty is denoted by $b$ and the loot by $\ell$. The pure Nash equilibria of the payoff matrices are underlined.

We find that in the first two scenarios, the equilibria strategies do not claim the loot. Hence, the network is given the desired heads-up before quantum looting becomes feasible. Furthermore, we find that the third scenario transforms into the second scenario as we decrease $w$ or increase the bounty.
Decreasing $w$ is modeled by moving \fmn to the right (towards \frn) and, similarly, \smn towards \srn, where both are moved the same distance and the other points remain fixed. The solid arrows show how timelines change as $w$ is decreased. We find that as $w$ decreases, all timelines but the last converge to scenarios where equilibria strategies do not claim the loot. The last timelime models a faster entity that is strong enough to claim the loot before the slow entity is able to claim the bounty, so it will win both the bounty and the loot for any value of $w$.

However, note that increasing the bounty might cause a player to increase the investment in building a small quantum computer, sufficiently large to win the bounty. This could be interpreted as moving \fln and \slen to the left. Since the bounty is orders of magnitude smaller than the loot, increasing the bounty does not seem to expedite the emergence of looting-capable quantum entities, so \fr and \sr (and consequentially \fm and \sm) remain fixed. Dashed arrows describe how the different timelines change as the bounty increases. In particular, the last timeline converges to a timeline whose equilibria strategy does not claim the loot. This motivates picking a fairly large bounty, to increase the chances that this dynamic would occur.




We do not provide a detailed analysis of the scenario where one (or both) entities are degenerate, as it does not provide much insight, and degenerate entities transform to non-degenerate ones as $w$ decreases.

\fi

\ifnum\masterthesis=1
    \bibliography{main}
\fi

\end{document}